\documentclass[a4paper,12pt]{amsart}
\usepackage[lmargin=2.4cm,rmargin=2.25cm,bmargin=3cm,tmargin=2.5cm]{geometry}

\usepackage[utf8]{inputenc}
\usepackage{amsmath} 
\usepackage{amssymb}
\usepackage{amsthm}
\usepackage{textcomp}
\usepackage{enumerate}
\usepackage[sort,numbers]{natbib}
\usepackage{array}
\usepackage{verbatim}
\usepackage{hyperref}
\usepackage{graphicx}
\usepackage{booktabs}
\usepackage{algpseudocode}
\usepackage{stmaryrd} 
\usepackage{xcolor}
\hypersetup{
  colorlinks,
  linkcolor={red!50!black},
  citecolor={blue!50!black},
  urlcolor={blue!80!black}
}

\usepackage{tcolorbox}
\usepackage{soul}
\usepackage[foot]{amsaddr}

\usepackage{subcaption}

\usepackage{xcolor}
\usepackage[linesnumbered,ruled,vlined]{algorithm2e}
\usepackage{setspace}
\SetKwInput{KwInput}{Input}                
\SetKwInput{KwOutput}{Output}              
\usepackage{lipsum}
\makeatletter
\newcommand{\algorithmfootnote}[2][\footnotesize]{%
  \let\old@algocf@finish\@algocf@finish
  \def\@algocf@finish{\old@algocf@finish
    \leavevmode\rlap{\begin{minipage}{\linewidth}
    #1#2
    \end{minipage}}%
  }%
}

\newtheorem{theorem}{Theorem}

\newtheorem{proposition}[theorem]{Proposition}
\newtheorem{lemma}[theorem]{Lemma}

\theoremstyle{definition}

\theoremstyle{remark}

\begin{document}

  \title[De-biasing particle filtering for a continuous time hidden Markov model]{De-biasing particle filtering for a continuous time hidden Markov model with a Cox process observation model}

  \author{Ruiyang Jin}
  \address[RJ]{Department of Engineering, University of Cambridge}

  \author{Sumeetpal S.~Singh}
  \address[SSS]{Department of Engineering, University of Cambridge}

  \author{Nicolas Chopin}
  \address[NC]{ENSAE, Institut Polytechnique de Paris}
  
\begin{abstract} 
   
    We develop a (nearly) unbiased particle filtering algorithm for a specific
    class of continuous-time state-space models, such that (a) the latent
    process $X_t$ is a linear Gaussian diffusion; and (b) the observations
    arise from a Poisson process with intensity $\lambda(X_t)$.  The likelihood
    of the posterior probability density function of the latent process
    includes an intractable path integral. Our algorithm relies on Poisson
    estimates which approximate unbiasedly this  integral. We show how we
    can tune these Poisson estimates to ensure that, with large probability,
    all but a few of the estimates generated by the algorithm are positive.
    Then replacing the negative estimates by zero leads to a much smaller bias
    than what would obtain through discretisation.  We quantify the probability
    of negative estimates for certain special cases and show that our particle
    filter is effectively unbiased. We apply our method to a challenging 3D
    single molecule tracking example with a Born and Wolf observation model.
\end{abstract} 


\maketitle
\sloppy

\section{Introduction} 
\subsection{Background}
Diffusion processes have been extensively used for modelling continuous-time
phenomena in a range of scientific areas, including finance
\citep{black2019pricing, merton1975optimum}, biochemistry 
\citep{mcadams1997stochastic, gillespie1976general, gillespie1977exact},
physics \citep{obukhov1959description} and engineering 
\citep{pardoux1984etude}. These processes are usually applied to model both the
observed process and an unobserved signal/state process in a hierarchical model. 

This paper develops novel methods for optimal filtering of multivariate
diffusion processes observed at irregular time instances, which follow a  Cox
process whose intensity is a (non-negative) function of the state process. 
The complete data likelihood of such a model includes a path integral of the
state trajectory (through the intensity function), which is intractable. 
This precludes the use of standard particle filters. 

Another common recognised problem in continuous time filtering for diffusion
processes is the unavailability of transition densities
\citep{jasra2020multilevel, fearnhead2010random}. In our problem though, the
hidden state is described by a linear SDE and thus state transition density is
available, but the likelihood still remains intractable for the reason
mentioned above. 
Estimators proposed by  \cite{nicolau2002new} replaces the path integral (with
respect to time) with a Riemann approximation based on a number of intermediate
points. This technique is further used to construct the transition density
estimator of \cite{durham2002numerical} and implemented in filtering context to
approximate the weights.

To remove the time-discretisation error in the numerical approximation of the
path intergral, the so-called Poisson estimator is often deployed. The path
integral estimate is computed using a (infinite) series expansion which is
expressed as a random finite series where the random truncation is given by a
Poisson random variable.  The first Poisson estimator was introduced in the
field of statistical physics by \cite{wagner1988unbiased, wagner1989undiased}.
This was subsequently further developed in the computational statistics
literature, to provide an unbiased estimation of diffusion transition densities
by \cite{beskos2006exact}, to its generalisation forms and its use in
sequential importance sampling by \cite{fearnhead2008particle}, and to its
variant using power series expansion by \cite{papaspiliopoulos_2011}. One
drawback of using a Poisson estimator is that it may return negative values,
which can result in an overall negative likelihood estimate and thus
prohibiting the use of the likelihood estimate for model calibration via
Particle MCMC \citep{andrieu2010particle}. A naive way to ensure positive
estimates is to truncate all negative estimates to zero, which obviously comes
at the expense of introducing a bias to the estimate.
\cite{fearnhead2010random} use Wald's identity (for martingales) to generate an
unbiased estimate of the path integral that are guaranteed to be positive.
However, this method does not seem to yield an unbiased  estimate of the
likelihood itself (see Section \ref{sec:furthercomments} for an elaboration on this
point), and has a bias which appears difficult to quantify.

\subsection{Contributions} 

The approach we pursue in this work is to employ the standard Poisson estimate and retain only the positive part of the returned estimate. (In Section \ref{sec:furthercomments}, we discuss the retaining the absolute value which will allow us to completely de-bias the estimate.) We are able to quantify the probability of encountering a negative weight (in certain idealised scenarios) and show that this probability decreases exponentially with the inverse of the time interval size over which the estimate is computed. (For some typical experimental settings in our numerical work the probability is exceptionally small, of the order $10^{-50}.$)  This exponential decrease in the probability of a negative estimate yields a few extra boons. The first being a rapidly diminishing mean square error, for the likelihood estimate, in the available CPU time. The second being the probability a complete run of an $N$-particle approximation for $T/\Delta$ time steps encountering a negative estimate (thus needing truncation) being equally rare and straightforward to control using our proposed (heuristic) tuning procedure. (Here $[0,T]$ is the time interval for smoothing, $\Delta$ is the interval over which the path integral is estimated and thus there are $T/\Delta$ path integrals to estimate for each particle.) To control a negative weight event, the extra simulation cost per-particle per-time step  is $\mathcal{O}(\Delta)$ and thus the total extra cost is $N \times (T/\Delta)\times \Delta$, which does not blow up $\Delta$ tends to zero.

As for our second contribution, we apply our methodology to a challenging model calibration problem arising from  single molecule fluoresencence microscopy, which is a very popular live cell imaging technology. We combine our likelihood estimate with the particle marginal Metropolis-Hastings algorithm  \citep{andrieu2010particle} to estimate the model parameters for data that arises from observing a diffusing molecule in 3D via a Cox process and a Born-Wolf observation model. We show how our particle filter significantly outperforms the conventional time discretisation based approach for the intractable path integral as implemented in \cite{d2022limits}. 
Our method is shown to have a negligible bias due to our tuning heuristic that controls the occurrence of a negative Poisson path integral estimate (and thus the truncation induced bias).


The paper is organised as follows. Section \ref{sec:problem_formulation} presents the problem formulation while  
Sections \ref{sec:paticle_filtering} and \ref{sec:discretisation_free} present the particle filtering methodology in continuous time. Section \ref{sec:paticle_filtering} presents the particle filter that employs a simple time discretisation of the path integral and Section \ref{sec:discretisation_free} the more sophisticated particle filter that employs the Poisson estimator of the path integral. Our proposed algorithm and accompanying theoretical results on its performance are also presented in Section 4.  Experiments including likelihood estimation, state estimation (smoothing) and parameter estimation are presented in Section \ref{sec: numerical experiments}. Proofs and additional algorithms can be found in the Appendix.
\section{Problem Formulation}\label{sec:problem_formulation}

\subsection{Notation}

We consider a  latent continuous time Markov process  $\{X_t\}_{t\geq0}$ which
takes values in $\mathcal{X}\subset \mathbb{R}^n$, has a time-inhomogeneous
Markov transition density, $X_{t_k}\vert(X_{t_{k-1}}=x_{k-1}) \sim
f^{\theta}_{t_{k-1},t_k}(x_{k}|x_{k-1})$, and initial density $\nu^{\theta}$.
The superscript $\theta$ is the parameter of the model, and will be defined
subsequently. 
By $X\sim \mathcal{N}(\mu,\Sigma)$, we mean that $X$ has the distribution of a
Gaussian random variable with mean $\mu$ and covariance $\Sigma$, whereas
$\text{N}(x; \mu,\Sigma)$ is the evaluation this Gaussian density at $x$. We
use the standard sequence notation $i:j=i,i+1, \ldots, j-1, j$, and
$\left\lceil x\right\rceil$ to denote the smallest integer number greater than
or equal to $x\in \mathbb{R}$.   The $\mathcal{Y}\subset \mathbb{R}^m$-valued
stochastic process $\{Y_k\}_{k\in\mathbb{Z}_+}$ corresponds to the observed
process with observation density $g^\theta(y_k\vert x_k)$. 
A realisation of a Poisson point process on the positive real line is a
sequence of increasing time points $0<t_1<t_2<\ldots$ generated according to a
non-negative intensity function $t \mapsto \lambda_t$. For our application, the
intensity function is stochastic and state dependent, i.e. $\lambda_t =
\lambda(X_t)\geq 0$.

\subsection{Hidden Markov Model Formulation} \label{sec: discrete time HMM}

Let $(t_1, y_{t_1}),\ldots, (t_{n_p}, y_{t_{n_p}})$ be an observed sequence of
non-negative increasing arrival times $0<t_i<T$ and arrival locations $y_{t_i}$
of a marked Poisson point process on the real line, recorded in the time
interval $[0, T]$. The arrival times are generated by a Poisson point process
on $[0,\infty)$ with stochastic intensity function $\lambda(X_t)$, which is
determined by a latent continuous time Markov process $\{X_t\}_{t\geq0} \subset
\mathcal{X}$ and a non-negative real valued function
$\lambda:\mathcal{X}\rightarrow\mathbb{R}$.  The locations $y_{t_i}\in
\mathcal{Y}$ are marks of the point process and are generated according to the
conditional (on $X_{t_i}=x$) probability density function,
\begin{equation*}
    Y_{t_i}|(X_{t_i}=x) \sim g^\theta (y |x)dy, \qquad i\in\left\{1:n_p\right\}.
\end{equation*}

The {\it exact} likelihood is
\begin{equation}
    \mathcal{L} = 
    \mathbb{E}\left\{\left(\prod_{i=1}^{n_p}\lambda\left(X_{t_i}\right)g^\theta\left(y_{t_i}\vert X_{t_i}\right)\right\}
\times \exp\left(-\int_{0}^{T}\lambda\left(X_s\right)ds\right)\right\}
    \label{eq: true likelihood}
\end{equation}
where the expected value is computed with respect to the law of
$\{X_t\}_{0\leq t \leq T}$.

\section{Particle filtering}\label{sec:paticle_filtering}

We adopt a discretisation of the positive real axis which is divided into segments of maximum length $\Delta$ defined sequentially as follows:
\begin{align}
    &t_0^\Delta=0,\nonumber\\ &t_k^\Delta=t^\Delta_{k-1}+\min\big\{\Delta, T-t_{k-1}^\Delta, \min_{t_i>t_{k-1}^\Delta} t_i-t_{k-1}^\Delta\big\},\quad  k>1
    \label{eq:time_step}
\end{align}
where $t_i$ is the (observed) arrival time. Thus \eqref{eq:time_step} defines  an increasing sequence of time points
$t_0^\Delta = 0<t_1^\Delta<\ldots< t_{m-1}^\Delta<t_m^\Delta=T$ spaced $\Delta$
apart unless the spacing is narrowed to coincide with the arrival of the
observation $y_{t_i}$ at time $t_i$ and ensures 
$\{t_1,\ldots,t_{n_p}\}\subset\{t_1^\Delta,\ldots,t_{m-1}^\Delta\}$.
The {\it exact} likelihood \eqref{eq: true likelihood} may be re-expressed using 
time points $t_i^\Delta$ as 
\begin{equation}
    \label{eq:true_likelihood_2}
    \mathcal{L}=\mathbb{E}\left\{
        \left(\prod_{i=1}^{n_p}\lambda(X_{t_{i}})g^\theta(y_{t_{i}}\vert
        X_{t_{i}})\right)
    \times\left(\prod_{j=1}^{m}\exp\left(-\int_{t_{j-1}^{\Delta}}^{t_{j}^{\Delta}}\lambda(X_{s})\mathrm{d}s\right)\right)\right\}.
\end{equation} 

The exact likelihood is not straightforwardly (using an approach such as in
Algorithm \ref{alg: bootstrap particle filter_discrete_time_version}) amenable
to unbiased estimation using particle filtering due to the path-integrals of
$\lambda(X_s)$. A simple approach is to replace the path-integral over $[0,T]$
with the following Reimann approximation
\begin{equation}
    \mathcal{L}_{\Delta}=\mathbb{E}\left\{ \left(\prod_{i=1}^{n_p}\lambda(X_{t_{i}})g^\theta(y_{t_{i}}\vert X_{t_{i}})\right)
        \times \prod_{j=1}^{m}\exp
        \left(-\lambda(X_{t_{j-1}^{\Delta}})(t_{j}^{\Delta}-
    t_{j-1}^{\Delta})\right)\right\}       
    \label{eq:approx_likelihood_2}
\end{equation}
(The subscript $\Delta$ denotes the dependence on the time discretisation and
emphasises that $\mathcal{L}_{\Delta}\neq\mathcal{L}$.) The posterior density
function of $(X_0, X_1, \ldots, X_m)=(X_{t_0^\Delta}, X_{t_1^\Delta}, \ldots,
X_{t_m^\Delta})$  for this time discretised model is 
\begin{align}
   &\int p^{\theta}_{\Delta}(x_0, \ldots, x_m) h(x_{0:m})dx_{0:m}\nonumber\\
   &\propto \mathbb{E}\left\{h\left(X_{t_0^{\Delta}}, \ldots, X_{t_m^{\Delta}}\right)\times 
       \left(\prod_{i=1}^{n_p}\lambda(X_{t_{i}})g^\theta(y_{t_{i}}\vert X_{t_{i}})\right) 
   \times \prod_{j=1}^{m}\exp \left(-\lambda(X_{t^{\Delta}_{j-1}})(t_{j}^{\Delta}-t_{j-1}^{\Delta}) \right)\right\} 
   \label{eq: integral form posterior discrete time}
\end{align}
This posterior density function  and its likelihood can be estimated using a
conventional particle filter  as described in Algorithm \ref{alg: bootstrap
particle filter_discrete_time_version} \citep{d2022limits}. 

The estimate $\hat{\mathcal{L}}_{\Delta}$ returned by Algorithm \ref{alg:
bootstrap particle filter_discrete_time_version} is an unbiased estimate of the
time-discretised likelihood ${\mathcal{L}}_{\Delta}$. In the next section we
will develop a particle method that approximates the exact (not-discretised)
likelihood, and in the numerical section (Section \ref{sec: numerical
experiments}) we will extensively contrast its estimation accuracy compared to
Algorithm \ref{alg: bootstrap particle filter_discrete_time_version} applied to
model \eqref{eq: integral form posterior discrete time}.

\begin{algorithm}[t!]
\setstretch{1.}
\SetAlgoLined
\DontPrintSemicolon
  

\For{$i\in\{1: N\}$}{Sample ${X}_0^{(i)}\sim \nu^\theta(\cdot)$.\\
Set 
$W_0^{(i)} = \exp\left(-X_0^{(i)}(t_{1}^{\Delta}-t_{0}^{\Delta})\right)$.\\
Resample $\{X_{0}^{(i)},W_0^{(i)}\}$ to obtain $\{\Tilde{X}_{0}^{(i)},\frac{1}{N}\}$.\\
}
\For{$k\in \{1:m-1\}$}{
\For{$i\in\{1:N\}$}{
Sample $X_k^{(i)}\sim f_{t_{k-1}^{\Delta},t_k^{\Delta}}^\theta(\cdot|\Tilde{X}_{k-1}^{(i)})$ and
set $X_{0:k}^{(i)}=(\Tilde{X}_{0:k-1}^{(i)},X_k^{(i)})$. \\
Set 
\[
W_k^{(i)} = \exp\left(-X_k^{(i)}(t_{k+1}^{\Delta}-t_{k}^{\Delta})\right)
\times \prod_{j=1}^{n_p}\left(
\lambda(X_k^{(i)})g^\theta(y_{t_{j}}\vert X_{k}^{(i)})
\right)^{\mathbb{I}[t_k^{\Delta}\leq t_j < t_{k+1}^{\Delta}]}.
\] \\
\hfill {\tt {\% Find all $y_{t_j}$ with $t_j \in [t_k^{\Delta},t_{k+1}^{\Delta})$.}} \\
Resample $\{X_{0:k}^{(i)},W_k^{(i)}\}$ to obtain $\{\Tilde{X}_{0:k}^{(i)},\frac{1}{N}\}$.\\
}}
Compute the (unbiased) estimate of the likelihood in \eqref{eq:approx_likelihood_2}:
\begin{equation}
\hat{\mathcal{L}}_{\Delta}=\prod_{k=0}^{m-1}\left\{\frac{1}{N}\sum_{i=1}^N W_k^{(i)} \right\}.
\label{eq: pf_biased likelihood estimates}
\end{equation}
\caption{Bootstrap particle filter}
\label{alg: bootstrap particle filter_discrete_time_version}
\end{algorithm}
\section{Particle Filtering to Mitigate Model Discretisation Error}
  \label{sec:discretisation_free}
  We propose a simple method to nearly unbiasedly estimate the true likelihood
$\mathcal{L}$. The idea to discretise the path integrals into smaller
$\Delta$ length time integrals, $\exp\left(-\int_{t}^{t+\Delta}
\lambda(X_s)ds\right)$, which are amenable to simple unbiased estimation and
whose probability of being positive approaches 1 rapidly as $\Delta$ tends to
0. We truncate a negative estimate to 0, and when combined with the rarity of
such events, it is simple to quantify the bias, which is also shown to be
rapidly decreasing as $\Delta$ tends to 0. This estimate can be used within
particle filtering and Particle MCMC; such methods are known particle filtering
with ``random weights'' as in \cite{rousset2006discussion},
\cite{fearnhead2008particle} and \cite{fearnhead2010random}.

Specifically, we are going to construct real valued random variables $E_1,
\ldots, E_{m}$, which are conditionally independent given $X_{t_0^\Delta},
\ldots, X_{t_{m}^\Delta}$ (in the manner made precise below in \eqref{eq: E})
and each unbiasedly estimates the corresponding term
$\exp(-\int^{t_{i}^\Delta}_{t_{i-1}^\Delta}\lambda(X_s)ds)$ in the manner of
\eqref{eq:unbiased_E}:
\begin{equation}
    p(e_1, \ldots, e_{m}|x_{t_0^\Delta},\ldots, x_{t_{m}^\Delta})=\prod_{i=1}^{m}p_{t_{i-1}^\Delta,t_{i}^\Delta}(e_i|x_{t_{i-1}^\Delta},x_{t_{i}^\Delta})
    \label{eq: E}
\end{equation}
\begin{equation}
    \int^\infty_{-\infty}e_ip_{t_{i-1}^\Delta, t_{i}^\Delta}(e_i|x_{t_{i-1}^\Delta},x_{t_{i}^\Delta})de_i=\mathbb{E}\Big\{\exp\big(-\int^{t_{i}^\Delta}_{t_{i-1}^\Delta}\lambda(X_s)ds\big)\big\vert X_{t_{i-1}^\Delta}=x_{t_{i-1}^\Delta}, X_{t_{i}^\Delta}=x_{t_{i}^\Delta}\Big\}. \label{eq:unbiased_E}
\end{equation}

With these random variables $E_1,\ldots, E_{m}$, we retain the unbiasedness of
the estimate of the numerator and denominator (the likelihood),
\begin{align*}
    &\int p_T(x_{0:m})h(x_{0:m})dx_{0:m}\\&\propto \mathbb{E}\Big\{h(X_{t_0^\Delta},X_{t_1^\Delta},\ldots, X_{t_m^\Delta})\times \Big(\prod^{n_p}_{i=1}\lambda(X_{t_i})g^\theta(y_{t_i}|X_{t_i})\Big)\times \prod^{m}_{j=1}E_j\Big\}
\end{align*}
which follows through a conditioning expectation argument. For $k\in\{1:m\}$,
let
\begin{align}
    &\int p_{t_k^\Delta}(x_0,\ldots, x_k)h_k(x_{0:k})dx_{0:k}\nonumber\\&\propto \mathbb{E}\Big\{h_k(X_{t_0^\Delta},X_{t_1^\Delta},\ldots,X_{t_k^\Delta})\times \Big(\prod^{n_p}_{i=1}[\lambda(X_{t_i})g^\theta(y_{t_i}|X_{t_i})]^{\mathbb{I}[t_i\leq t_k^\Delta]}\Big)\times \prod^{k}_{j=1}E_j\Big\}
    \label{eq:sequential_unbiased}
\end{align}
where, recall, $t_m=T$. Once we have defined \eqref{eq: E}, it will be
straightforward to construct a particle approximation of the conditional
probability density functions \eqref{eq:sequential_unbiased}. These posterior
densities, unlike \eqref{eq: integral form posterior discrete time}, do not
have a time discretisation bias. Our particle filtering algorithm, detailed in
Algorithm \ref{alg: bootstrap particle filter}, also returns an estimate the
exact likelihood \eqref{eq:true_likelihood_2}.

The next subsection explains how to construct these variables $E_i$ using the
Poisson estimate approach. The following subsections will explain how to ensure
that the probability of $E_i<0$ may be made negligible. 
\subsubsection{The Poisson Estimator}
We first consider a fixed trajectory $\{X_s\}_{0<s\leq t_1^\Delta}$, then
\begin{align}
    \exp\left(-\int_{0}^{t_{1}^{\Delta}}\lambda(X_{s})\mathrm{d}s\right)
    &=\exp(c)\exp(I-c)\nonumber\\
    &=\exp(c)\sum^\infty_{k=0}\frac{(I-c)^k}{k!}\nonumber\\
    &=\exp(c+\eta)\sum^\infty_{k=0}\exp(-\eta)\frac{\eta^k}{k!}\big(\frac{I-c}{\eta}\big)^k\nonumber\\
    &=\exp(c+\eta)\sum^\infty_{k=0}\exp(-\eta)\frac{\eta^k}{k!}\prod^k_{i=1}
    \mathbb{E}_{\tau_i}\left(\frac{-t_1^\Delta \lambda(X_{\tau_i}) -c}{\eta}\right)\nonumber\\
    &=\exp(c+\eta)\mathbb{E}_\kappa\left[\prod^\kappa_{i=1}
    \mathbb{E}_{\tau_i}\left(\frac{-t_1^\Delta \lambda(X_{\tau_i})-c}{\eta}\right)\right],
    \nonumber
\end{align}
where $I=-\int^{t_1^\Delta}_0\lambda(X_s)ds$ and $-t_1^\Delta \lambda(X_{\tau_i})$'s are the unbiased estimates of $I$. The above derivation follows the approach outlined in \cite{papaspiliopoulos_2011}.

The inclusion of constant $c$ is to optimise the resulting estimator. The
inclusion of $\mathcal{P}$o$(\eta)$ distribution is to allow an unbiased
estimate to be based on a truncated sum and finally $\mathbb{E}_\kappa(\cdot)$
and $\mathbb{E}_{\tau_i}(\cdot)$ denote expectation with respect to
$\kappa\sim\mathcal{P}o(\eta)$ and $\tau_i\sim\mathcal{U}(0,t_1^\Delta)$. The final line yields the resulting unbiased estimator: 

\begin{equation}
    E_1=\exp(c+\eta)\Big[\mathbb{I}_{\{\kappa=0\}}
    +\mathbb{I}_{\{\kappa>0\}}\Big(\prod^\kappa_{i=1}\frac{-t_1^\Delta\lambda(X_{\tau_i})-c}{\eta}\Big)\Big]
\label{eq:poisson_estimate}
\end{equation}
as the sample from $E_1\sim p(e_1|x_0,x_{t_1^\Delta})$. 

\cite{papaspiliopoulos_2011} discussed how to choose $c$ and $\eta$ in order to
make the variance of the estimate as small as possible. In particular, he
showed that $c^\star=I-\eta$ is the value of $c$ that minimises the variance
(for a fixed $\eta$). Our approach is slightly different: we aim at controlling
the probability of the estimate being negative. For that purpose, we set 
$c=-t_1^\Delta \lambda(X_0)-\eta$ (which can also be seen as a tractable
approximation of $c^\star$). This yields:
\begin{equation*}
    E_1=\exp\left\{-t_1^\Delta\lambda\left(X_0\right\}\right)\left[\mathbb{I}_{\{\kappa=0\}}+\mathbb{I}_{\{\kappa>0\}}\left(\prod^\kappa_{i=1}\left[1+\frac{t_1^\Delta}{\eta}\left(\lambda\left(X_{0}\right)-\lambda\left(X_{\tau_i}\right)\right)\right]\right)\right].
\end{equation*}
We postpone the discussion on how to control the probability of a negative
estimate to the next sub-section. 

The Poisson estimator for any time interval $t_{i-1}^\Delta\leq t\leq t_{i}^\Delta$ is detailed in Algorithm \ref{alg: Random Weight Estimator Algorithm for One Particle}. Note that we assume we can exactly sample $X_{\tau_j}$ from 
$p(x_\tau |x_{\tau_{j-1}})$ for $j\in\{1:\kappa\}$. This is possible for linear
Gaussian diffusions, as discussed in the introduction; see Appendix \ref{sec: bridge density} for
details. 

The particle filter with the Poisson estimator is described in Algorithm
\ref{alg: bootstrap particle filter}. Step 8 of this algorithm makes a call to
Algorithm \ref{alg: Random Weight Estimator Algorithm for One Particle} to get
the desired samples $E_k^{(i)}$ from $p(e_k\vert X_{t_{k-1}^\Delta}^{(i)},
X_{t_k^\Delta}^{(i)})$.

\begin{algorithm}[t!]
\DontPrintSemicolon
  
 \KwInput{$\eta, t_{i-1}^\Delta, t_i^\Delta, X_{t_{i-1}^\Delta}$}

    Generate $\kappa\sim\mathcal{P}o\big(\eta\big)$.\\
    Generate $\tau_1,\tau_2,\ldots, \tau_\kappa\sim\mathcal{U}(t_{i-1}^\Delta,t_{i}^\Delta)$, sort them in ascending order and relabel them so that $\tau_1<\tau_2<\ldots,<\tau_\kappa$.\\
    Sequentially sample $X_{\tau_j}$ from $p(x_\tau|x_{\tau_{j-1}})$ for $j\in\{1: \kappa\}$ where $\tau_0=t_{i-1}^{\Delta}$. Sample $X_{t_{i}^\Delta}$ from $p(x_{t_{i}^\Delta}|x_{\tau_{\kappa}})$.\\
  Compute and return the estimate: 
\begin{align}
\scriptsize
    E=&\exp\left(-(t_{i}^\Delta-t_{i-1}^\Delta)\times \lambda\left(X_{t_{i-1}^\Delta}\right)\right)\times\nonumber\\
    &\Bigg[\mathbb{I}_{\{\kappa=0\}}+
    \left.\mathbb{I}_{\{\kappa>0\}}\left(\prod^{\kappa}_{j=1}\left[1+\frac{t_{i}^\Delta-t_{i-1}^\Delta}{\eta}\left(\lambda(X_{t_{i-1}^\Delta})-\lambda(X_{\tau_j})\right)\right]\right)\right].
    \nonumber
\end{align}

  \KwOutput{$(E,X_{t_{i}^\Delta})$ \qquad {\tt \% The sample from $p(e, x_{t_{i}^\Delta} \vert X_{t_{i-1}^\Delta} )$}}
\caption{PE($\eta, t_{i-1}^\Delta, t_i^\Delta, X_{t_{i-1}^\Delta}$)}
\label{alg: Random Weight Estimator Algorithm for One Particle}
\end{algorithm}
\begin{algorithm}[tb]
\small
\DontPrintSemicolon
Find $\Delta$ \eqref{eq:find_Delta} and define time steps \eqref{eq:time_step}.\\
\For{$i\in\{1:N\}$}{Sample $\Tilde{X}_{t_0^\Delta}^{(i)}\sim \nu^\theta(\cdot)$ and set $W_0^{(i)}= \frac{1}{N}$.\\
}
Estimate $\hat{l}_0$. \hfill {\tt \% See Section \ref{sec:design_choice}.}\\
\For{$k\in\{1:m\}$}{
\For{$i\in\{1:N\}$}{
Set $\eta_k=\left(t_k^\Delta-t_{k-1}^\Delta\right)\hat{l}_{k-1}$.\\
Sample $(E^{(i)}_k,X_{t_k^\Delta}^{(i)})\leftarrow$PE($\eta_{k}, t_{k-1}^\Delta, t_k^\Delta, \tilde{X}_{t_{k-1}^\Delta}^{(i)})$ and set $\left(X_{t_0^\Delta}^{(i)}, \ldots, X_{t_k^\Delta}^{(i)}\right)=\left(\Tilde{X}_{t_0^\Delta}^{(i)}, \ldots, \Tilde{X}_{t_{k-1}^\Delta}^{(i)}, X_{t_k^\Delta}^{(i)}\right)$.\\
Update $\hat{l}_k$ using \eqref{eq:update_lipschitz}. \\
Set $W_k^{(i)} = \max\{E_k^{(i)},0\}
\times \prod_{j=1}^{n_p}\left(
\lambda(\tilde{X}_{t_{k-1}^\Delta}^{(i)})g^\theta(y_{t_{j}}\vert \tilde{X}_{t_{k-1}^\Delta}^{(i)})
\right)^{\mathbb{I}[t_j=t_{k-1}^{\Delta}]}.$\\
\hfill {\tt \% Incorporating $y_{t_j}$ with $t_j=t_{k-1}^{\Delta}$. }\\
Resample $\big\{(X_{t_0^\Delta}^{(i)},\ldots, X_{t_k^\Delta}^{(i)}), W_k^{(i)}\big\}$ to obtain $\big\{(\Tilde{X}_{t_0^\Delta}^{(i)},\ldots, \Tilde{X}_{t_k^\Delta}^{(i)}),\frac{1}{N}\big\}$.}}
Compute the likelihood estimate:
\begin{equation}
\hat{\mathcal{L}}=\prod_{k=1}^{m}\left\{\frac{1}{N}
\sum_{i=1}^N W_k^{(i)}
\right\}.
\label{eq: pf_unbiased likelihood estimates}
\end{equation}
\caption{Bootstrap particle filter in continuous time}

\label{alg: bootstrap particle filter}
\footnotetext{Remark: $\hat{l}_k$ and $\eta_k$ in step 8 and 10 are updated according to the design choices described in Section \ref{sec:design_choice}. PE used in step 9 is the Poisson estimator algorithm described in Algorithm \ref{alg: Random Weight Estimator Algorithm for One Particle}}
\end{algorithm}
\subsection{Negative Poisson Estimate Control}
Although the Poisson estimator can return negative values, the following lemma shows that the probability of this happening is controllable by adjusting $(\eta,\Delta)$ and in particular decays exponentially fast in $\Delta$.

\begin{lemma}
Let $\left\{X_s\right\}_{0\leq s\leq \Delta}$ be one dimensional Brownian motion which starts at $X_0=x_0$. Consider the estimate \eqref{eq:poisson_estimate} (with $t_1^{\Delta}=\Delta$) of the path integral
$\mathbb{E}\left \{ \left. \exp\left(-\int_{0}^{\Delta} \lambda(X_s)ds\right) \right|X_\Delta=x_\Delta\right\}$. Let $\lambda(\cdot)$ be a non-negative $l$-Lipschitz function, then 
the following bound holds when $\eta> \Delta l \left \vert x_\Delta-x_0\right\vert$,
\begin{align}
    \Pr \Big(E_1<0 \vert \kappa >0, X_\Delta=x_\Delta \Big)< 2\exp\Big\{-\frac{\frac{2\eta}{\Delta l}(\frac{\eta}{\Delta l}-|x_\Delta-x_0|)}{\Delta}\Big\}.
    \label{eq: bound for probability}
\end{align}
\label{lem: probability bound for E<0}
\begin{proof}
See Appendix \ref{sec: lemma 1}.
\end{proof}
\end{lemma}
Note that the estimate is trivially positive when $\kappa=0$ and hence the bound is given conditionally on $\kappa>0$. 
An illustration of how $\text{Pr}(E_1<0|\kappa>0)$ and its corresponding bound evolve as $\Delta$ changes for different choices of $\eta$ is provided in Figure \ref{fig:log prob l3 eta}. Each data point is a Monte Carlo estimate of the conditional probability (conditioned on $\kappa>0$) the random variable \eqref{eq:poisson_estimate}, with $t_1^{\Delta}=\Delta$, is negative. The Monte Carlo estimate of the conditional probability is computed for various  choices of $\eta$ and $\vert x_\Delta-x_0\vert$ using $10^8$ experiments each. \eqref{eq: bound for probability} suggests that choosing $\eta=c\Delta^{\frac{3}{2}}l$ with $\left \vert x_\Delta- x_0\right \vert=d\Delta^{\frac{1}{2}}$ (for some positive constants $c$ and $d$) results in constant bound. This is reflected by the straight line behaviour of the Data 1 in Figure \ref{fig:log prob l3 eta}. For contrast, the bound on the conditional probabilities are also illustrated. 
\begin{figure}[t!]
    \centering
    \includegraphics[width=0.6\linewidth]{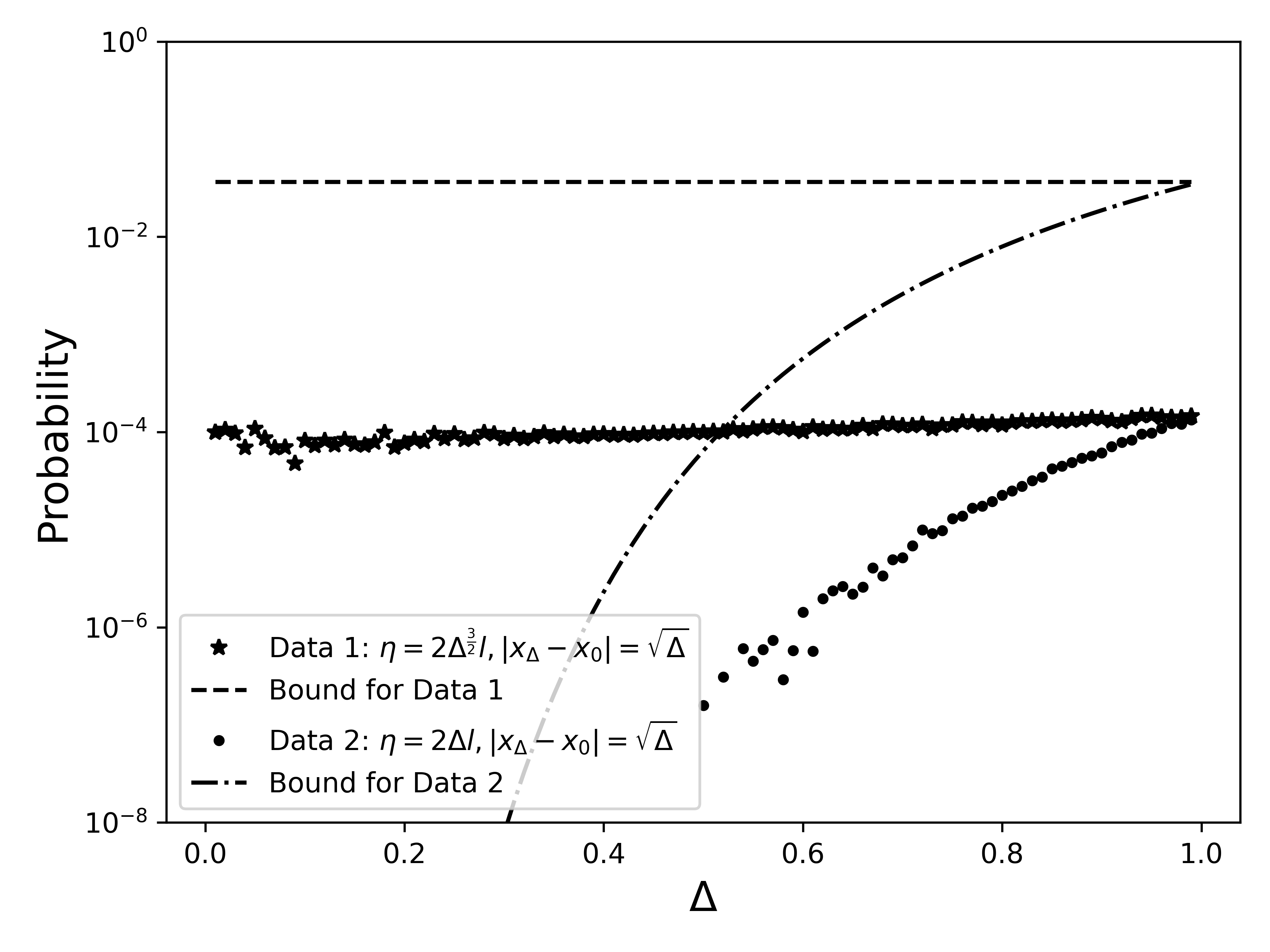}
    \caption{Plot of probability \eqref{eq: bound for probability} 
    vs. $\Delta$ for different choices of $\eta$.}
    \label{fig:log prob l3 eta}
\end{figure}
 We also compute the bound when averaging over $X_\Delta$. Combining the bound \eqref{eq: bound for probability} with the expansion  $\mathbb{I}_{[E_1<0]}\leq \mathbb{I}_{[E_1<0]}\mathbb{I}_{[\eta> \Delta l \left \vert X_\Delta-x_0\right\vert]} + \mathbb{I}_{[\eta\leq \Delta l \left \vert X_\Delta-x_0\right\vert]}$, we can compute the unqualified bound for $\Pr(E_1 <0 \vert \kappa > 0)$ to be 
\begin{equation}
\Pr(E_1 <0 \vert \kappa > 0) \leq
2+4\Phi\left(\frac{2\eta}{\Delta^{\frac{3}{2}}l}\right)-6\Phi\left(\frac{\eta}{\Delta^{\frac{3}{2}}l}\right)\label{eq:avg_bound}
\end{equation}   
where $\Phi$ is the CDF of a standard normal distribution. (The proof is provided in Appendix \ref{sec: expectation of the probability bound}.) In Section \ref{sec:design_choice}, we advocate a design choice of $\eta=\Delta l$ (with the Lipschitz constant estimated in a causal manner with the population of particles) to ensure the simulation cost decreases proportionally with the time discretisation $\Delta$. An estimate of $\mathbb{E}\left\{\exp\left(-\int^T_0 \lambda(X_t) dt \right) \right\}$ or 
\[
 \mathbb{E}\left\{\exp
 \left(-\int^\Delta_0 \lambda(X_t) dt \right)\cdots
 \exp\left(-\int^T_{\lfloor\frac{T}{\Delta} \rfloor \Delta} \lambda(X_t) dt \right)
 \right\}
\]
will entail the product of $T/\Delta$ (conditionally independent) estimates for the individual intervals. Using \eqref{eq:avg_bound}, this estimate is negative with a probability no greater than $\Pr(E_1 <0 \vert \kappa > 0)\times T/\Delta$. Figure \ref{fig: averaged bound} illustrates how the bound in \eqref{eq:avg_bound}, when multiplied with $T/\Delta$, decays with the choice $\eta=\Delta l$.
\begin{figure}[t!]
    \centering
    \includegraphics[width=0.6\linewidth]{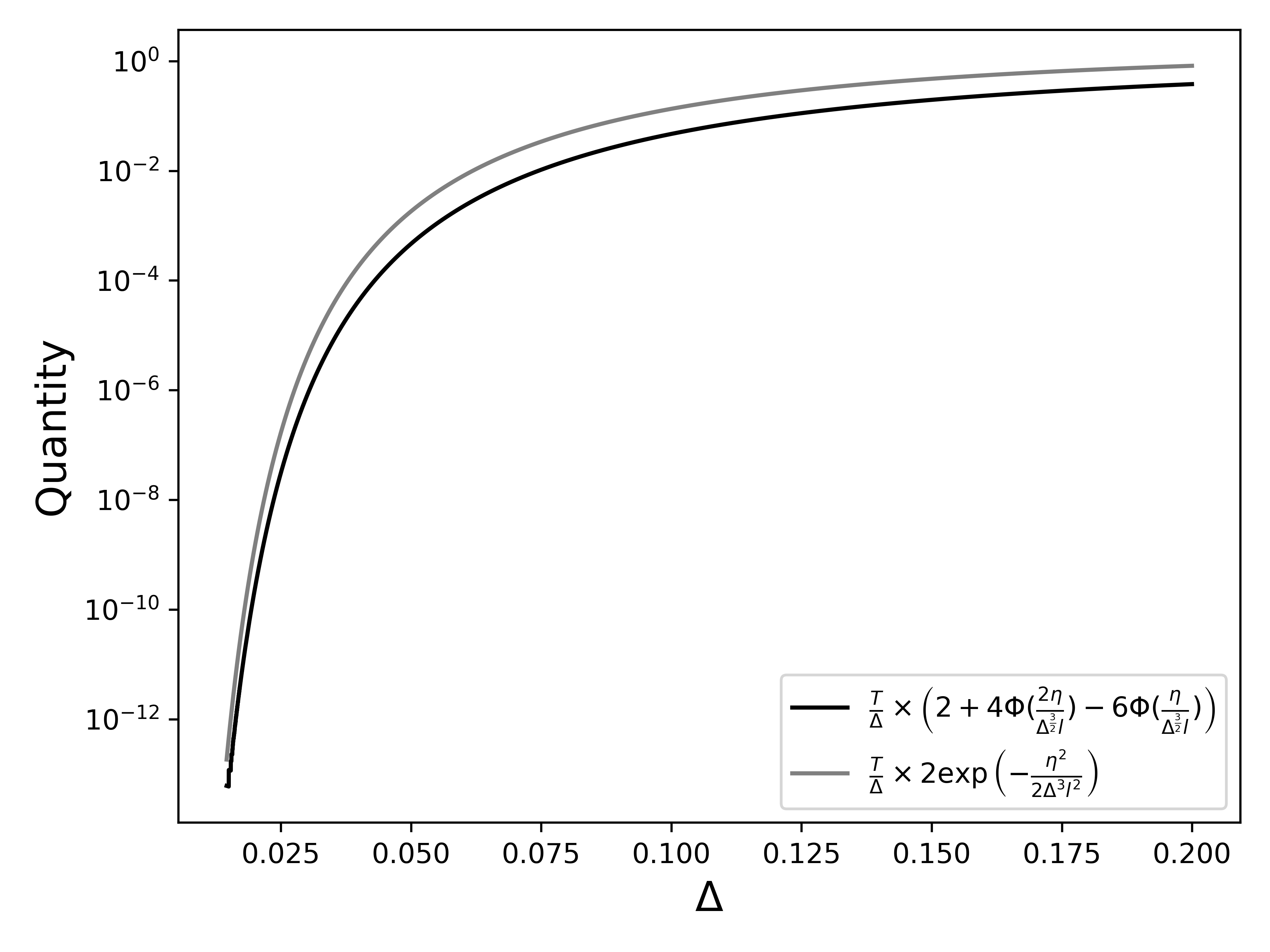}
    \caption{Plot of \eqref{eq:avg_bound} multiplied by $T/\Delta$ vs. $\Delta$ for the design choice of $\eta=\Delta l$ when $T=1$ with fitted relationship.}
    \label{fig: averaged bound}
\end{figure}
\subsubsection{Design choice for $(\eta,\Delta)$}  \label{sec:design_choice}
We employ a reasoning similar to the above to bound the probability  Algorithm \ref{alg: bootstrap particle filter} encounters a negative Poisson estimate. For a step-size $\Delta$, an $N-$particle implementation has $\left\lceil{T/\Delta}\right \rceil$ forward steps and the event of encountering at least one negative Poisson estimate is $\left\{\bigcup^N_{n=1}\bigcup^{\left\lceil{T/\Delta}\right \rceil}_{i=1}\{E_i^{(n)}<0\}\right\}$. Using  Lemma \ref{lem: probability bound for E<0}, its probability may be bounded above by the union bound, 
\begin{align*}
    &\Pr\left(\bigcup^N_{n=1}\bigcup^{\left\lceil{T/\Delta}\right \rceil}_{i=1}\{E_i^{(n)}<0\}\right)\leq\sum^N_{n=1}\sum^{\left\lceil{T/\Delta}\right \rceil}_{i=1}\Pr\left(E_i^{(n)}<0\right)
    < N\left\lceil{\frac{T}{\Delta}}\right\rceil \times 2\exp\left\{-\frac{\frac{2\eta}{\Delta l}(\frac{\eta}{\Delta l}-d \sqrt{\Delta})}{\Delta}\right\}
\end{align*}
where we have assumed that $|x_{ \Delta}-x_{(i-1)\Delta}|\leq d \sqrt{\Delta}$, for all $i \in\left\{1:\left\lceil{T/\Delta}\right \rceil\right\}$, for some constant $d>0$ and  $\eta/(\Delta l)> d \sqrt{\Delta}$. (A similar heuristic could also be found using \eqref{eq:avg_bound}.) We can make a design choice for $\eta$ and $\Delta$ (within the constraints $\eta\geq \Delta^{\frac{3}{2}}l$ and $\eta/(\Delta l)> d \sqrt{\Delta}$) to ensure the probability of encountering a negative estimate is at most $\epsilon$,
\begin{equation}
    \left\lceil{\frac{NT}{\Delta}}\right \rceil\times 2\exp\Big\{-\frac{\frac{2\eta}{\Delta l}(\frac{\eta}{\Delta l}- d \sqrt{\Delta})}{\Delta}\Big\}\leq \epsilon.
    \label{eq: conditional probability equation}
\end{equation}
For example, using $\eta =\Delta l$, the bound will fall below $\epsilon$ once $\Delta$ is small enough, say $\Delta = \bar{\Delta}$, and will continue to hold as $\Delta$ is decreased further since the left hand side decreases as $\Delta$ decreases. A similar heuristic could also be found using \eqref{eq:avg_bound}. In summary, set $\eta=\Delta l$  and 
\begin{equation}
\Delta=\sup\{\Delta >0 : \eqref{eq: conditional probability equation} \textrm{~and~} \left\lceil{{NT}/{\Delta}}\right \rceil\times \eqref{eq:avg_bound}\leq \epsilon
\}
\label{eq:find_Delta}
\end{equation}
One can apply numerical methods such as Newton's method to solve \eqref{eq: conditional probability equation}. Also, $\epsilon$ can be exceptionally small, for example, $\Delta=0.01$, $NT=10^4, d=3$ and $\eta=\Delta l$ yields  $\epsilon\approx 10^{-55}$. 

The design choice ${\eta}_k= \Delta \hat{l}_{k-1}$ can be computed sequentially in Algorithm \ref{alg: bootstrap particle filter}. $\hat{l}_{k-1}$ is the empirical Lipschitz constant updated sequentially as follows \begin{equation}
    \hat{l}_k:=\max\left\{\max_{i \in \{1:N\}}
    \frac{|\lambda(X_{t_k^\Delta}^{(i)})-\lambda(X_{t_{k-1}^\Delta}^{(i)})|}{|X_{t_k^\Delta}^{(i)}-X_{t_{k-1}^\Delta}^{(i)}|}, \hat{l}_{k-1}\right\}
    \label{eq:update_lipschitz}
\end{equation}
where the initial estimate $\hat{l}_{0}$ can be chosen to be the maximum ratio estimate as in  \eqref{eq:update_lipschitz} but computed with the particle set at time $t_0^{\Delta}$ only and the maximum is found over $i \neq j \in\{1:N\}$.  These design choices for $\hat{l}_k$ and $\eta_k$ for Algorithm \ref{alg: bootstrap particle filter} are used in all the numerical experiments presented in Section \ref{sec: numerical experiments}.

\subsubsection{Truncation Bias}
In Algorithm \ref{alg: bootstrap particle filter}, we truncate the negative Poisson estimates to zero which will induce a bias. Hence we wish to study the bias of this truncated estimate for time discretisation, $0<\Delta<\cdots < m\Delta =T$ when $\Delta$ approaches zero, i.e.
\begin{equation*}
    \mathbb{E}\left\{\exp\left(-\int_0^T\lambda(X_s)ds\right)\right\}-\mathbb{E}\left\{E_1^+\cdots E^+_m\right\}.
\end{equation*}
where $E^+_i=E_i\mathbb{I}_{A_i^c}$ is the truncated Poisson estimate and $A_i$ denotes the event $E_i<0$.
To do so, we can bound the omitted term $\mathbb{I}_{A}\prod_{i=1}^m E_i$ where $A=A_1\cup \ldots \cup A_m$ using the following lemma.
\begin{lemma}\label{lem:bias}
Let $\{X_s\}_{0\leq s\leq \Delta}$ be one dimensional Brownian motion which starts at $X_0=x_0$. Let $\lambda(\cdot)$ be a non-negative $l$-Lipschitz function and consider the estimate of the path integral 
\begin{equation*}
\mathbb{E}\left\{\exp\left(-\int_0^\Delta\lambda(X_t)dt\right)\cdots \exp\left(-\int_{(m-1)\Delta}^{m\Delta}\lambda(X_t)dt\right)\right\}=\mathbb{E}\left\{E_1\cdots E_m\right\},
\end{equation*}
where $(E_{i+1},X_{(i+1)\Delta})\leftarrow \mathrm{PE}(\Delta l,i\Delta,(i+1)\Delta,X_{i\Delta})$ (see Algorithm \ref{alg: Random Weight Estimator Algorithm for One Particle}) for $i=0,\ldots,m-1$.

Then the following bound holds,
\begin{align*}
    \left\vert\mathbb{E}\left\{\mathbb{I}_A \prod_{i=1}^m E_i\right\}\right\vert&\leq \exp(\frac{Tl}{2})\times \left(\frac{1+4\Delta^2 l}{1-4\Delta^2 l}\right)^\frac{m}{2} 
    \times m^{\frac{1}{2}}
    \left[2\exp\left(-\frac{1}{2\Delta}\right)\right]^{\frac{1}{2}}.
\end{align*}
\begin{proof}
See Appendix \ref{sec: lemma 2}.
\end{proof}
\end{lemma}
For $m=T/\Delta$, the second (ratio) term in the product recedes quickly to one as $\Delta$ approaches 0, which implies the final term dominates the bias. For $m=T/\Delta$, the final term also tends to 0.
Based on this result, as an indicative trend, the square of the relative bias (which contributes additively in the relative MSE calculation) of Algorithm \ref{alg: bootstrap particle filter} is of the order
\begin{equation*}
    \frac{\left(\mathcal{L}-\mathbb{E}(\hat{\mathcal{L}})\right)^2}{\mathcal{L}^2}\leq \text{const}(T)\times \frac{1}{\Delta} \exp\left(-\frac{1}{2\Delta}\right).
\end{equation*}
where $\hat{\mathcal{L}}$ is \eqref{eq: pf_unbiased likelihood estimates}.  This result is commented on further in Section \ref{sec: benchmark}.
\subsubsection{Further comments} \label{sec:furthercomments}
The following idea, based on \emph{Wald's identity} for sampling, was employed in \cite{fearnhead2010random} to deal with negative weights in particle filtering.
We describe it here in the context of a single step within particle filtering
and discuss its implications for estimating the likelihood. Consider
$X_{0}\sim\nu^{\theta}$ and let $G^{\theta}(x_0)$ be a non-negative function, also assumed
$\theta$ dependent, and the aim is to estimate the likelihood $L(\theta)=\mathbb{E}^{\theta}\left(G^{\theta}(X_0)\right)$. 
Assume there exists an unbiased estimate of $G^{\theta}(x_0)$ for
any $(\theta,x_{0})$ defined as follows. Let $p^{\theta}(e\vert x_{0})$
be a conditional pdf on the real line with mean $\int_{-\infty}^{\infty}ep^{\theta}(e\vert x_{0})\mathrm{d}e=G^{\theta}(x_0)$.
Given $X_{0}$, let $E^{(i)}$, $i=1,2,\ldots$, be independent
samples from $p^{\theta}(e\vert X_{0})$ and let $K=\inf\left\{ k>0:E^{(1)}+\ldots+E^{(k)}>0\right\} $.
Then $\hat{L}=\sum_{i=1}^{K}E^{(i)}$ has mean 
\[
\mathbb{E^{\theta}}(\hat{L})=\mathbb{E}^{\theta}\left(G^{\theta}(X_0) \mathbb{E}^{\theta}(K \vert X_0)\right)\neq L(\theta)\times \mathrm{constant}
\]
where product $G^{\theta}(X_0) \mathbb{E}^{\theta}(K \vert X_0)$ is Wald's identity, $\mathbb{E}^{\theta}(K\vert X_0)$ is the mean of the number of independent
draws needed to ensure positivity and the constant on the right is $\theta$ independent; a $\theta$ independent constant is needed for the method to be used for model calibration. This approach of sampling until the estimate is positive was proposed in \cite{fearnhead2010random} to 
address the event that a negative estimate is returned by $p^{\theta}(e\vert X_{0})$. The constant $\mathbb{E}^{\theta}(K\vert X_0)$
seems to play no role in a particle filtering algorithm, since the
weights are normalised before used as an input to the resampling step.
However $\mathbb{E}^{\theta}(K\vert X_0)$, which is clearly $X_0$ dependent, can be $\theta$ dependent as well, e.g.~as it would for $G^{\theta}(x_0)=\mathbb{E}^{\theta}_{x_0}\{\exp(-\int_0^{\Delta} \lambda (X_s)\mathrm{d}s) \}$ and its estimate $E^{(i)}$ returned by Algorithm \ref{alg: Random Weight Estimator Algorithm for One Particle} (for $\mathrm{PE}(\eta,0,\Delta,x_0)$) since $K$ would depend on the law of $\{X_t\}_{t}$. Also, the function $G^{\theta}(x_{0})$ can be $\theta$ dependent. As there is no easy way to
compute or remove this factor $\mathbb{E}^{\theta}(K\vert X_0)$, this precludes
its use within e.g.~a PMMH sampler which require a
(positive) unbiased estimator of $L(\theta)$ to generate MCMC samples
from the posterior density of the model parameters $\theta$. We provide some experiments in Appendix \ref{sec: wald experiments} to show that the idea of Wald's identity for sampling returns biased estimates.

We note finally that it is possible to adapt our approach slightly to return
(perfectly) unbiased estimates.  Recall that a particle filter such as
Algorithm \ref{alg: bootstrap particle filter} may return an unbiased estimate
of not only the normalising constant, but more generally of any
unnormalised path expectation \citep{del2004feynman}; that is, the quantity 
\begin{equation}
    \label{eq:unbiased_gamma}
    \hat{\mathcal{L}} \times \frac{\sum_{i=1}^N W_{m}^{(i)}
\varphi(X_{t_0^\Delta}^{(i)}, \ldots, X_{t_{m}^\Delta}^{(i)})}{\sum_{i=1}^N W_{m}^{(i)}}
\end{equation}
is an unbiased estimate of 
\begin{equation*}
\mathbb{E}\left\{ \left(\prod_{i=1}^{n_p}\lambda(X_{t_{i}})g^\theta(y_{t_{i}}\vert X_{t_{i}})\right)
  \Psi(X_{t_0^\Delta}, \ldots, X_{t_m^\Delta}) \times
  \varphi(X_{t_0^\Delta}, \ldots, X_{t_m^\Delta}) \right\}
\end{equation*}
where $\Psi(\cdot)$ is the expectation of a product of Poisson estimates of the form 
$\prod_{i=1}^m \max(0, E_i)$. This quantity would be equal to 
$\exp\left( - \int_0^T \lambda(X_s) ds\right) $ if we could replace each
truncated estimate $\max(0, E_i)$ by the estimate $E_i$ itself.

We may use this to estimate unbiasedly the marginal likelihood of an
alternative model, based on a different likelihood for the data (given the
states). In particular, consider a variant of Algorithm \ref{alg: bootstrap
particle filter} where $\max\{E_k^{(i)}, 0\}$ is replaced by $|E_k^{(i)}|$ in
line 11. (Adapt the definition of $\Psi$ accordingly.) The weights remain
non-negative, and the output remains biased (for estimating the true likelihood
$\mathcal{L}$). In \eqref{eq:unbiased_gamma}, replace $\varphi(\cdot)$ by 
$(-1)^n$,  where $n$ is the number of negative Poisson estimates $E_k$ that have
occurred while constructing the considered trajectory (the argument of
$\varphi(\cdot)$).  It is easy to see that this is an unbiased estimate of the
true likelihood $\mathcal{L}$. (Formally, $\varphi$ is then a function of both
the  state trajectory and the $E_k$ variables that have been generated while
constructing that trajectory in this case). 

In our numerical experiments, it was easy to set up the tuning parameters to
make the number of occurrences of negative weights equal to virtually zero, so
we could not observe any practical benefit in removing (entirely) the bias.
However, this approach may be kept in mind for more complicated scenarios. 

\section{Numerical Experiments}
\label{sec: numerical experiments}

In this section, we present numerical examples to compare Algorithm \ref{alg: bootstrap particle filter_discrete_time_version} and Algorithm \ref{alg: bootstrap particle filter} for likelihood estimation, smoothing and model calibration using Particle MCMC. 

\subsection{1D example with exact calculation}
\label{sec: benchmark}
We first consider a simple example in which the state $X_t$ is a one
dimensional Brownian motion and $X_t$ is observed in zero mean unit variance
Gaussian noise.
The intensity function of the Cox process is $\lambda(x) = x+10$.
The state starts at $x_0=0$ at time $t=0$ and the record of observations stops
at time $T=2$. 

The integration that defines this likelihood can be computed exactly and thus
can serve as ground truth; see Appendix \ref{sec: exact computation of likelihood function}.
We assume we observe $n_p=2$ data-points, to make it possible to
do a large number of runs; see Appendix \ref{sec: empirical relationship between rel var and delta} for extra results
with $n_p>2$. 

For the analysis below,  we use the relative mean squared error (rMSE) as the metric to measure the quality of likelihood estimates.
\begin{figure}[t!]
     \centering
     \begin{subfigure}[b]{0.48\textwidth}
         \centering
         \includegraphics[width=\textwidth]{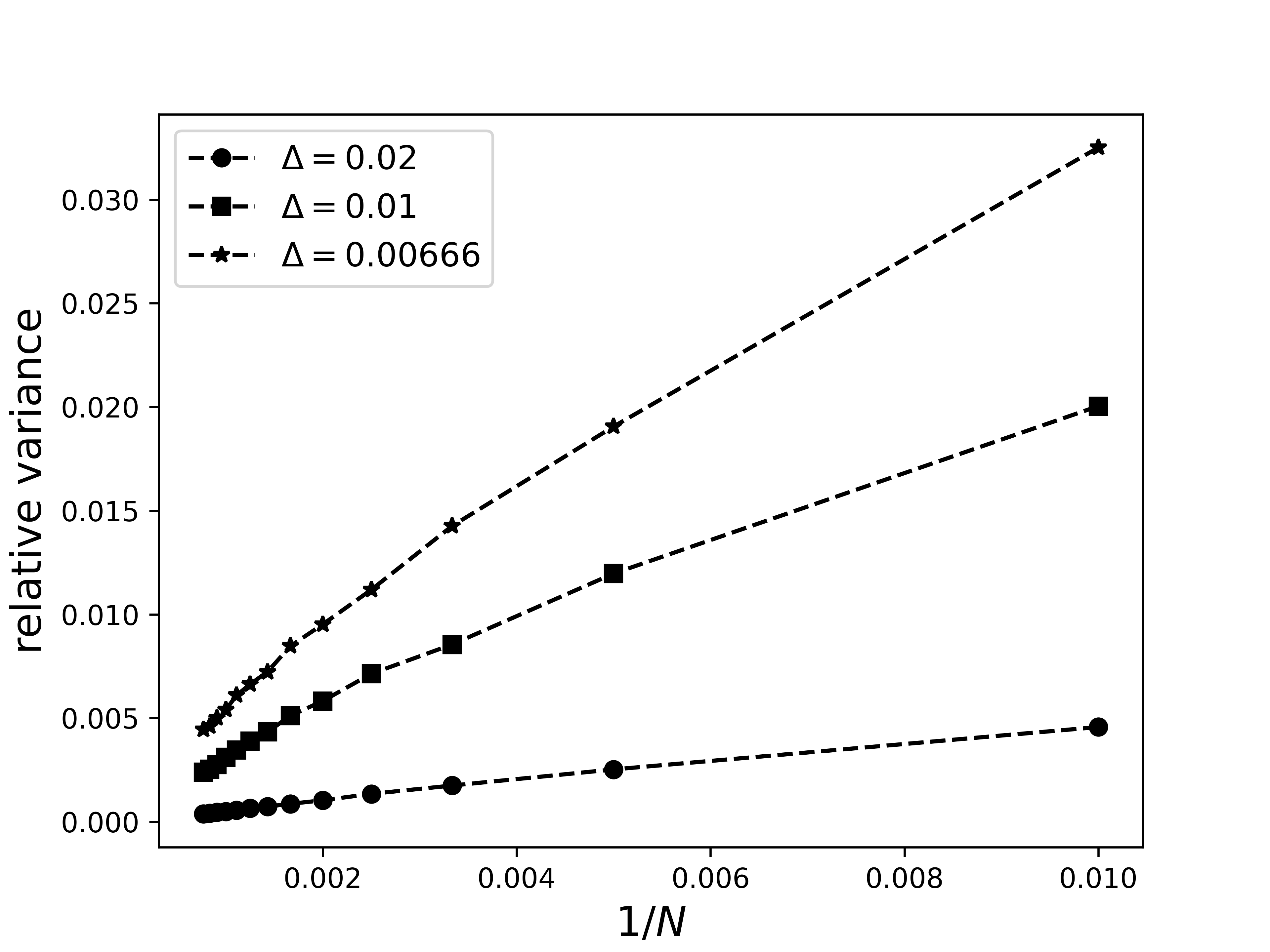}
         \caption{}
         \label{fig:relative variance vs Delta}
     \end{subfigure}
     \hfill
     \begin{subfigure}[b]{0.48\textwidth}
         \centering
         \includegraphics[width=\textwidth]{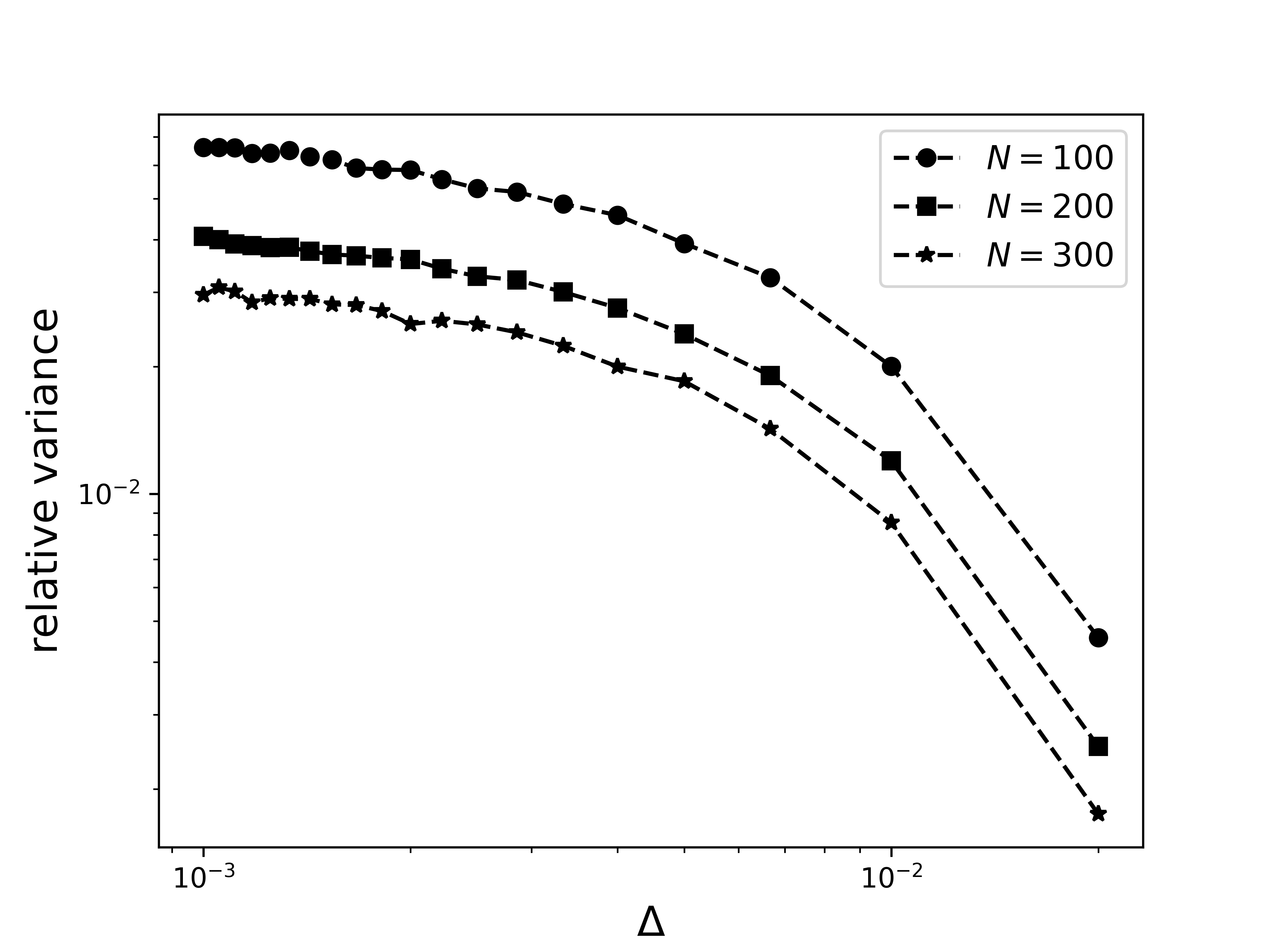}
         \caption{}
         \label{fig:relative variance vs N}
     \end{subfigure}
     \caption{Plot of relative variance, defined to be $\mathbb{E}\{(\hat{\mathcal{L}}_\Delta)^2\}/\mathcal{L}_\Delta^2-1$, for $\hat{\mathcal{L}}_\Delta$ given by Algortihm \ref{alg: bootstrap particle filter_discrete_time_version}. In (a) versus $1/N$ and in (b) as $\Delta$ varies on a log scale. }
\end{figure}
Numerical results displayed in Figure \ref{fig:relative variance vs Delta}
shows that the estimate of $\mathcal{L}_\Delta$ given by Algorithm \ref{alg:
bootstrap particle filter_discrete_time_version}, for any $\Delta$, has a
relative variance which is inversely proportional to the number of particles
$N$ used in the particle filter, where the relative variance defined to be 
$\mathbb{E}\{(\hat{\mathcal{L}}_\Delta)^2\}/\mathcal{L}_\Delta^2-1$. 
(We note though that the slope varies very slightly with $1/N$.) In Figure \ref{fig:relative variance vs N}, as expected, the relative variance for a fixed $N$ stabilises as $\Delta$ decreases. (This reason is that a time-discretised particle system with systematic resampling converges to a continuous time limit as $\Delta$ approaches zero, as recently shown in \cite{chopin2022resampling}.) 
For any sufficiently smooth function $\lambda(\cdot)$, the weak error of Euler scheme (i.e. relative bias  $(\mathcal{L}_\Delta/\mathcal{L})-1$ in our case) is at most of order $\Delta$ \citep[Chapter~17]{kloeden2011numerical}. Overall, this implies the following empirical relationship for all values of $n_p$ when $\Delta$ is small:
\begin{alignat}{1}
    \textrm{rMSE} &= \frac{1}{\mathcal{L}^2} \mathbb{E}\left\{\left(\hat{\mathcal{L}}_{\Delta} - \mathcal{L} \right)^2 \right\} = \frac{c_1}{N} + c_2\Delta^2= \frac{c_1 }{\mathcal{C}\Delta} + c_2\Delta^2
    \label{eq: empirical relationship}
\end{alignat}
where $\mathcal{C}$ denotes the CPU time spent to run the particle filter (Algorithm \ref{alg: bootstrap particle filter_discrete_time_version}) to completion. 
In the last equality, we use the relationship that $\mathcal{C}$ increases linearly with $NT/\Delta$ which corresponds to $T/\Delta$ propagation steps for $N$ particles. (Figure \ref{fig:Ushape} confirms \eqref{eq: empirical relationship}.) For fixed CPU time of $\mathcal{C}$, the value of $\Delta$ that minimises the relative MSE is $\Delta^*=\left(\frac{c_1}{c_2\mathcal{C}}\right)^{\frac{1}{3}}$. Substituting this $\Delta^*$ into \eqref{eq: empirical relationship} gives the best relative MSE value for each $\mathcal{C}$, which is of order $\mathcal{O}(\mathcal{C}^{-\frac{2}{3}})$ and confirmed in Figure \ref{fig:Comparison with random weight method}. Similarly, we can apply the same idea to determine $\Delta$ that minimises the relative MSE for Algorithm \ref{alg: bootstrap particle filter},
\begin{align}
    \text{rMSE}&\leq \frac{c_1}{N}+\frac{c_2}{\Delta} \exp\left(-\frac{1}{2\Delta}\right)=\frac{c_1}{\mathcal{C}\Delta}+\frac{c_2}{\Delta}\exp\left(-\frac{1}{2\Delta}\right)
    \label{eq: rMSE of unbiased method}
\end{align}
Since the minimisation problem above cannot be solved exactly, one can pursue a surrogate for $\Delta^*$, in its vicinity, by minimising
\begin{equation*}
    f(\Delta) = \frac{c_1}{\mathcal{C}\Delta}+c_2\exp\left(-\frac{1}{2\Delta}\right).
\end{equation*}
Note that $\text{rMSE}(\Delta)>f(\Delta),\quad \forall 0<\Delta<1$. Minimising this equation gives $\Delta^*=\left(2\log\left(\frac{c_2\mathcal{C}}{2c_1}\right)\right)^{-1}$.
 Hence $\Delta^*$, not being the true minimiser of \eqref{eq: rMSE of unbiased method}, is a more conservative solution. Substituting this $\Delta^*$ into \eqref{eq: rMSE of unbiased method} gives an indication of the best relative MSE value for each $\mathcal{C}$, which is of order of $\mathcal{O}\left({\mathcal{C}}^{-1}{\log(\mathcal{C})}\right)$. In practice, we do not recommend this optimisation but rather choose $(\Delta,\eta)$ as  detailed in Section \ref{sec:design_choice} and then stick to this choice even if more CPU time $\mathcal{C}$ has become available.
\begin{figure}[t!]
     \centering
     \begin{subfigure}[b]{0.48\textwidth}
         \centering
         \includegraphics[width=\textwidth]{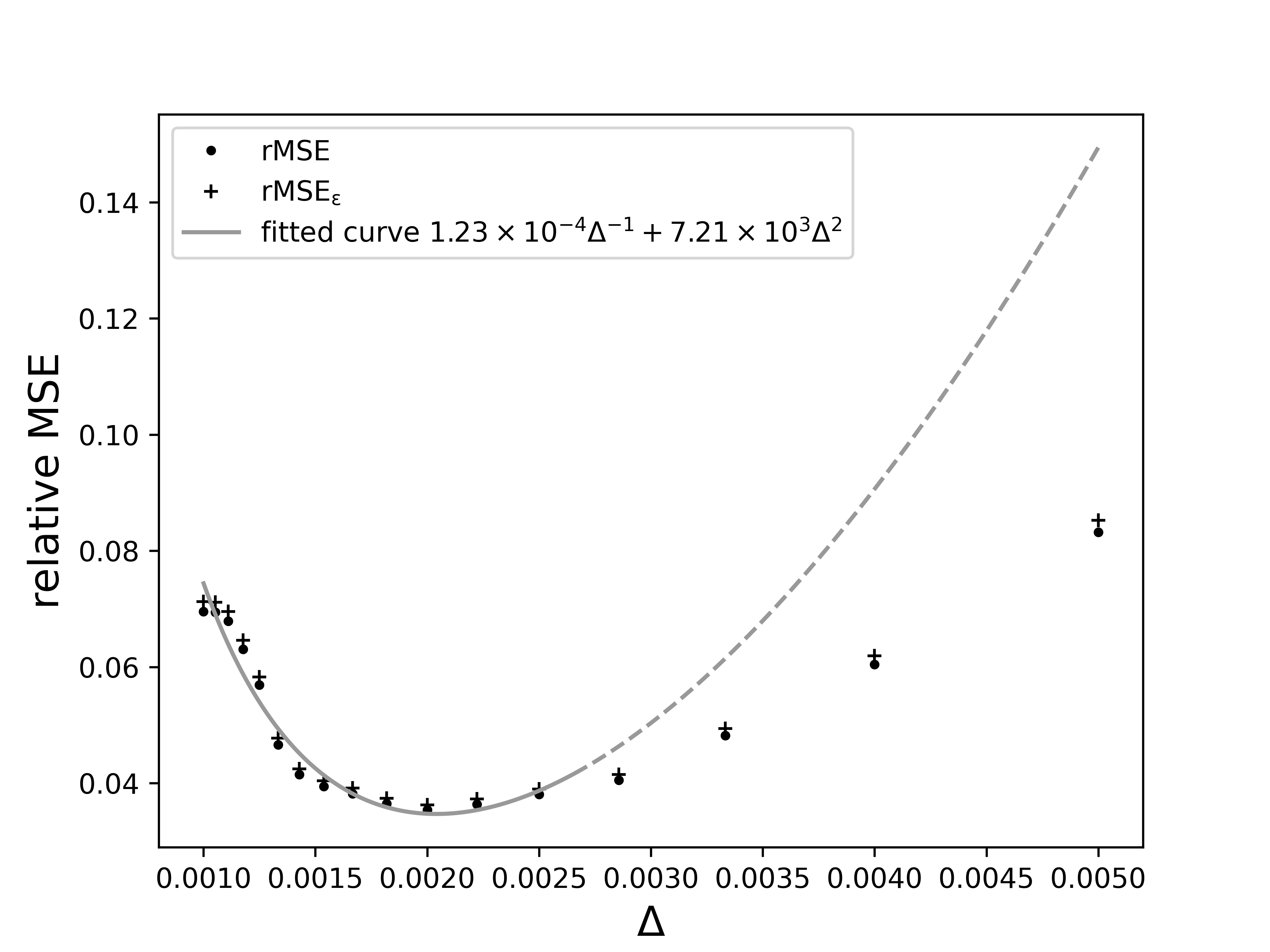}
         \caption{}
         \label{fig:Ushape np=2}
     \end{subfigure}
     \hfill
     \begin{subfigure}[b]{0.48\textwidth}
         \centering
         \includegraphics[width=\textwidth]{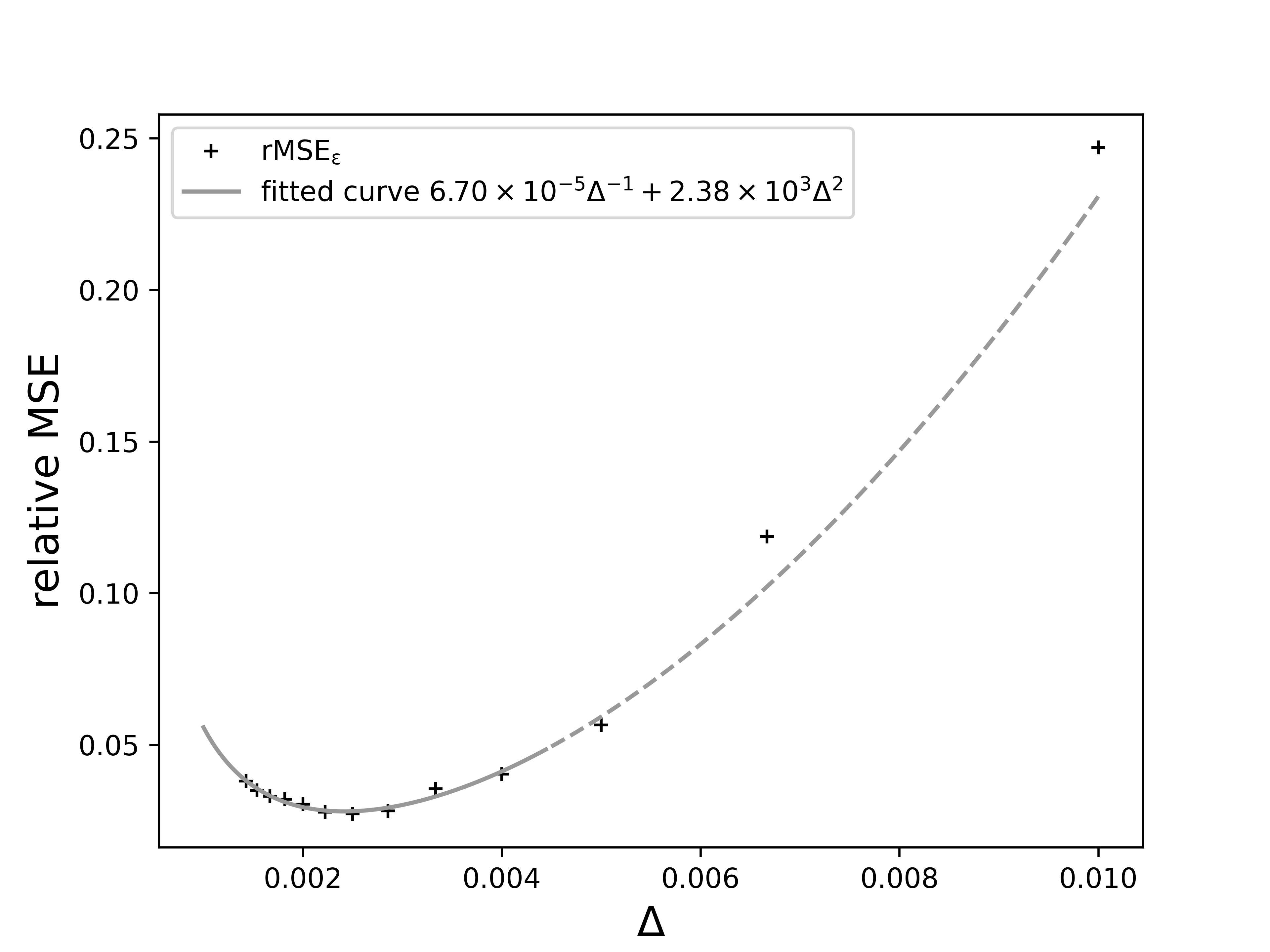}
         \caption{}
         \label{fig:Ushape np=13}
     \end{subfigure}
     \caption{Plot of relative MSE versus $\Delta$ for fixed $1.5s$ CPU time for likelihood estimates computed by Algorithm \ref{alg: bootstrap particle filter_discrete_time_version} (using \eqref{eq: pf_biased likelihood estimates}) for (a) $n_p=2$ and (b) $n_p=13$.
     Overlaid is the fitted relationship in \eqref{eq: empirical relationship} to a range of $\Delta$ values around the minimum of rMSE, illustrated by the solid segment of the line. The dashed segment cover data points there were not used in fitting. The coefficients of the fitted curve to $\text{rMSE}_\epsilon$ are only 2.4\% different to that of the rMSE data points, hence it would be indistinguishable graphically.}
     \label{fig:Ushape}
\end{figure}
We define $\text{rMSE}_\epsilon$ to be \eqref{eq: empirical relationship} with $\mathcal{L}$ replaced with $\mathcal{L}_{\textrm{MC}}=\mathcal{L} + \epsilon$.
Recall we denote by $\mathcal{L}_{\text{MC}}$ the Monte Carlo estimate returned by the modified Algorithm 3 which uses the true path integral given by \eqref{eq:exact_E_k} in Appendix rather than Poisson estimate. We ensure the Monte Carlo error $\epsilon$ is small enough so that our conclusions in comparing the accuracy of Algorithms \ref{alg: bootstrap particle filter_discrete_time_version} and \ref{alg: bootstrap particle filter} are not rendered inaccurate for the case $n_p>2$ studied below. We use the $n_p=2$ case to choose a value of $\epsilon$ that ensures the  best $\Delta$ found using $\text{rMSE}_\epsilon$ is close enough to the desired (best) $\Delta$ for $\text{rMSE}$.

Continuing with for $n_p=2$, Figure \ref{fig:Ushape np=2} reports the rMSE and $\text{rMSE}_\epsilon$ of Algorithm \ref{alg: bootstrap particle filter_discrete_time_version} for a fix CPU budget and different $\Delta$ values with the expected relationship in \eqref{eq: empirical relationship} fitted to a range of $\Delta$ values around the minimum. $\text{rMSE}_{\epsilon}$ uses $\mathcal{L}_{\textrm{MC}}$ which is the average estimate of $\mathcal{L}$ given by $10^6$ runs of modified Algorithm \ref{alg: bootstrap particle filter} with each run using $N=10^6$ particles. We calculate the relative error between $\Delta_\epsilon^*$ and $\Delta^*$, and between $\text{rMSE}_\epsilon(\Delta^*_\epsilon)$ and $\text{rMSE}(\Delta^*)$, using their fitted $c_1$'s and $c_2$'s values,
\begin{align*}
    \left\vert\frac{\Delta_\epsilon^*-\Delta^*}{\Delta^*}\right\vert=1.1\times 10^{-9},\quad\left\vert \frac{\text{rMSE}_\epsilon(\Delta^*_\epsilon)-\text{rMSE}(\Delta^*)}{\text{rMSE}(\Delta^*)}\right\vert=0.024.
\end{align*}
This shows that $10^6$-averaged runs of modified Algorithm \ref{alg: bootstrap particle filter} with $N=10^6$ particles is more than sufficient to produce accurate estimate $\mathcal{L}_{\text{MC}}$ as the substitute of $\mathcal{L}$. We used the same number of Monte Carlo repetitions and $N$ for values of $n_p>2$ up to $n_p=13$, which are reported in  Figure \ref{fig:Ushape np=13}. Both Figure \ref{fig:Ushape np=2} and \ref{fig:Ushape np=13} validate the expression for the rMSE \eqref{eq: empirical relationship} in the locality of the minimum $\Delta$.
\begin{figure}[t!]
    \centering
    \includegraphics[width=0.6\linewidth]{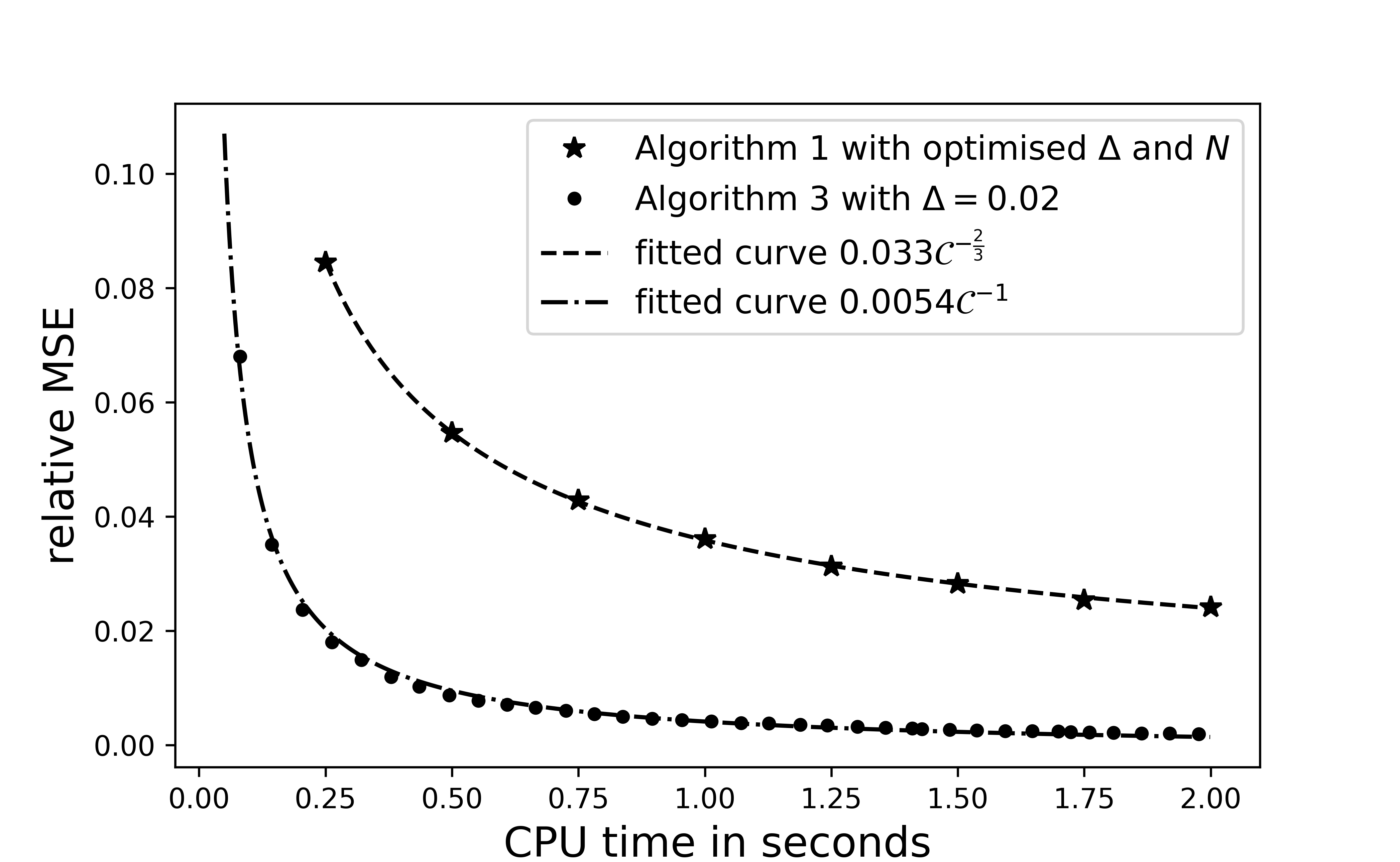}
    \caption{Comparison between likelihood estimates computed by Algorithm \ref{alg: bootstrap particle filter_discrete_time_version} (using \eqref{eq: pf_biased likelihood estimates}) and by Algorithm \ref{alg: bootstrap particle filter} (using \eqref{eq: pf_unbiased likelihood estimates}). The true likelihood is approximated with Algorithm \ref{alg: bootstrap particle filter} using $N=10^4$ particles and $\Delta=0.02$. Note that Algorithm 1 uses optimised $\Delta$ and $N$ to obtain the best MSE for a given CPU budget. Algorithm \ref{alg: bootstrap particle filter} uses fixed $\Delta$ value $\Delta=0.02$ and the design choice for $\eta$ described in Section \ref{sec:design_choice}.}
    \label{fig:Comparison with random weight method}
\end{figure}
We continue to use $\mathcal{L}_{\text{MC}}$ to compare Algorithms \ref{alg: bootstrap particle filter_discrete_time_version} and \ref{alg: bootstrap particle filter}. We use $\mathcal{L}_{\text{MC}}$ to find the smallest relative MSE  Algorithm \ref{alg: bootstrap particle filter_discrete_time_version} can achieve for a given CPU budget, while we use $\mathcal{L}_{\text{MC}}$ to compute the relative MSE of  Algorithm \ref{alg: bootstrap particle filter} for the same CPU budget. For Algorithm \ref{alg: bootstrap particle filter_discrete_time_version}, for each value of $\mathcal{C}$, we repeat the procedure illustrated in Figure \ref{fig:Ushape} to find the $\Delta$ that yields the smallest $\text{rMSE}_{\epsilon}$ -- this $\Delta$ is the minimiser of the fitted line as illustrated in Figure \ref{fig:Ushape np=13}. For Algorithm \ref{alg: bootstrap particle filter}, we spend the budget on increasing the number of particles $N$ while using a fixed $\Delta$ value of $\Delta=0.02$. The results of this comparison are shown in Figure \ref{fig:Comparison with random weight method}. It appears that Algorithm \ref{alg: bootstrap particle filter} achieves the best decay rate of rMSE with CPU budget, which is the inverse relationship, whereas Algorithm \ref{alg: bootstrap particle filter_discrete_time_version} can only achieve a rate of $\mathcal{C}^{-2/3}$.
\subsection{3D Single Molecule Model}
\label{sec: 3D single molecule model}
In this section we apply our methodology to track a moving biological molecule
(biomolecule) in a live cell, in three dimensions, arising from single molecule
fluorescence microscopy.  An illustration of how the data is generated is given
Figure \ref{fig:setup}. Single molecule fluorescence microscopy is a live cell
imaging technique where biomolecules of interest are tagged with a fluorophore,
which are then excited with light at a particular frequency. These molecules
fluoresce under excitation and emit light a different frequency, which is then
captured by a CCD camera after optical magnification. The recorded images are
used to uncover their motion. A mathematical abstraction of the problem is
precisely the model in Section \ref{sec: discrete time HMM}, see also
\cite{ober2020quantitative} and \cite{d2022limits}. In particular, the moving
molecule follows a diffusion model and its observations are the (random)
arrival times and locations of individual photons. The data are both the
arrival times and locations of the photons. The photon arrival times are
governed by the depth of the molecule (see Figure \ref{fig:schematics for total
internal reflection}) as the excitation of the molecule varies inversely with
the molecule's depth due to the attenuation of the excitation light.  Photon
arrival locations are imprecise (noise corrupted) observations of the
molecule's location in the other two dimensions as governed by diffraction
theory. The relevant photon location model is the Born and Wolf model for the
point spread function, which describes how a point light source appears in an
image as it moves in and out of focus \citep{ober2020quantitative}.
\begin{figure}[t!]
    \centering
    \includegraphics[width=0.7\linewidth]{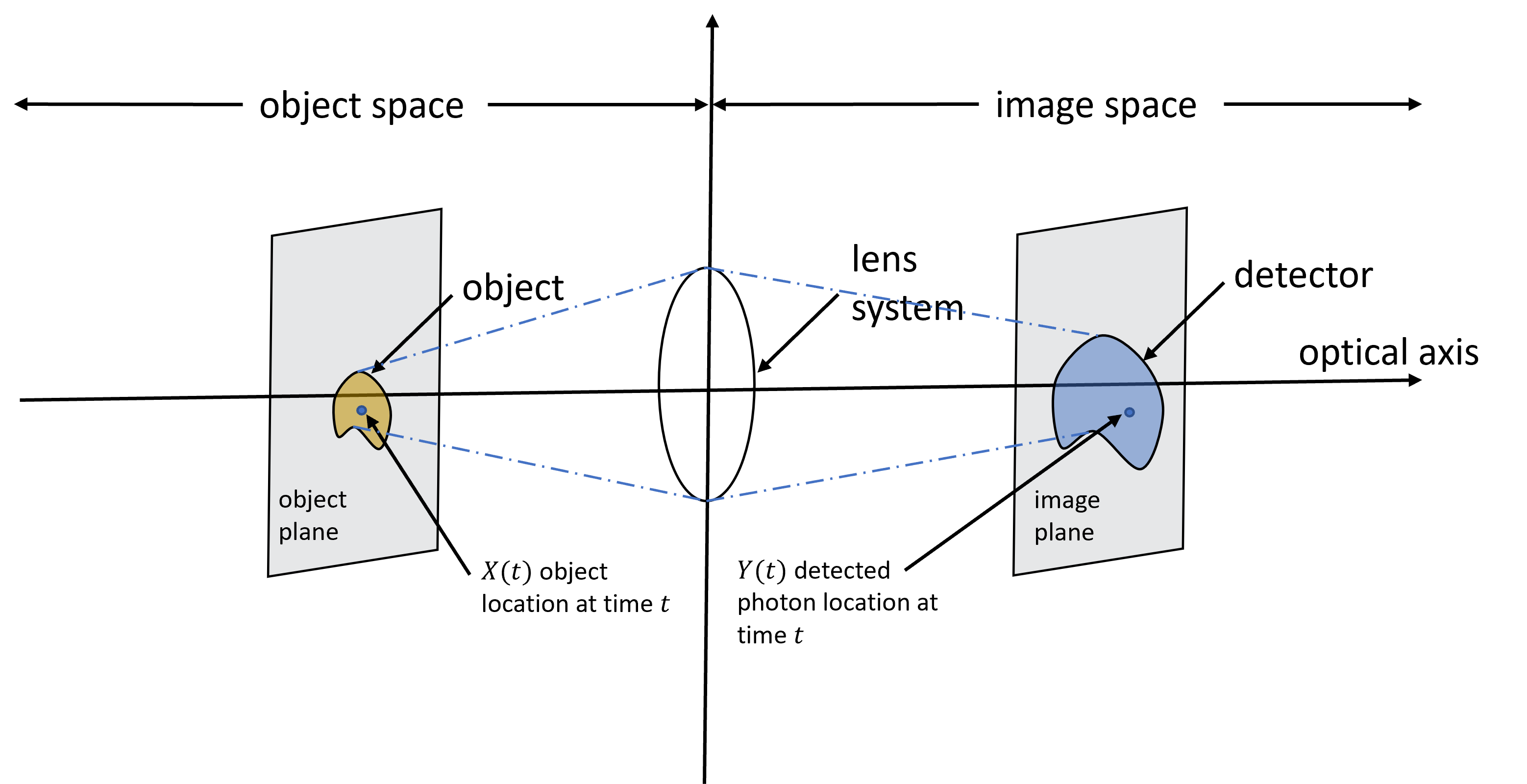}
    \caption{Illustration of the setup about how the image is acquired by a microscope. $X(t)$ denotes the object location at time $t$ and $Y(t)$ denotes the detected photon location at time $t$. }
    \label{fig:setup}
\end{figure}
We let $(X_t)_{0\leq t\leq T}:=(X_{1,t},X_{2,t}, X_{3,t})^\top_{0\leq t\leq T}$ denote the true, 3-dimensional location of the molecule at time $t$. The three components of the molecule state, i.e.$(X_{1,t}, X_{2,t}, X_{3,t})^\mathrm{T}_{0\leq t\leq T}$, are its $(x,y,z)^\top$ location and assumed to follow the \textit{Ornstein-Uhlenbeck} (O-U) model,
\begin{align*}
    dX_{i,t} =-\phi_i(X_{i,t}-\mu_i)dt+dW_{i,t},\qquad \text{for }i=1,2,3
\end{align*}
where $\phi_i>0$ and $(W_{i,t})_{0\leq t\leq T}$, $i=1,2,3$, are independent
Brownian motions. We assume that the initial distribution which generates $X_0$
is $\mathcal{N}(\mu, \Sigma_0)$,
where the covariance matrix $\Sigma_0=p_0\times \mathbb{I}_{3\times 3}$.
The transition density $f^\theta_{\delta}(x'|x)$ of the process can be expressed as follows,
\begin{align}
    X_{i,t+\delta}|&\big(X_{i,t}=x_{i}\big)\sim \mathcal{N}
    \big(\mu_{i}+e^{-\delta \phi_i}(x_{i}-\mu_{i}),
    \frac{1}{2\phi_i}(1-e^{-2\delta\phi_i})\big), \qquad i=1,2,3.
    \label{eq: f12 density}
\end{align}
For an object located at $(x_1,x_2,x_3)^\mathrm{T}\in \mathbb{R}^3$ in the object space (prior to magnification), the location (on the detector) at which a photon is detected is specified probabilistically with $2D$ probability density function,
\begin{equation}
    g^\theta(y|x):=\frac{1}{|M|}q_{x_3}\big(M^{-1}y-(x_{1},x_{2})^\intercal\big), \qquad y\in\mathbb{R}^2
    \label{eq: photon distribution profile}
\end{equation}
where $M\in \mathbb{R}^{2\times 2}$ is an invertible lateral magnification
matrix and the image function $q_{x_3}:\mathbb{R}^2\rightarrow \mathbb{R}$
describes the image of an object in the detector space when that object is
located at $(0,0,x_3)$ in the object space, where $x_3\in \mathbb{R}$ is the
location of the object on the optical axis. This 3D Born and Wolf model is the
resulting image function, derived from diffraction theory, for a point source
that can also be out of focus \citep{born2013principles}. For $(x_1,x_2)\in \mathbb{R}^2$, 
\begin{equation}
    q_{x_3}(x_1,x_2)=\frac{4\pi n_\alpha^2}{\lambda_e^2}\left \vert \int^1_0J_0(\frac{2\pi n_\alpha}{\lambda_e}\sqrt{x_1^2+x_2^2} \rho)\exp(\frac{j\pi n_\alpha^2 x_3}{n_0\lambda_e}\rho^2)\rho d\rho \right\vert^2,
    \label{eq: Born and WOlf profile}
\end{equation}
where  $n_0$ is the refractive index of the objective lens immersion medium and $n_\alpha$ is the numerical aperture of the objective lens. $\lambda_e$ is the emission wavelength of the molecule. $J_0(\cdot)$ and $J_1(\cdot)$ represents the zero-th order and the first order Bessel function of the first kind, respectively. The probability density functions of Born and Wolf model at different defocus levels are plotted in Figure \ref{fig: BW_pdf_different defocused levels}. A large defocus tends to produce images of poor quality and this will further pose difficulty in estimating the molecule's position.

The time instances at which the photons arrive on the detector are random as
well, and the photon emission process is modelled as a Poisson process \citep{ober2020quantitative}. The photon rate, denoted as $\lambda(t)$, is the rate at which photons are emitted by the object at time $t$ is often assumed to be constant. For example, \cite{d2022limits} and \cite{vahid2020fisher} applied particle filtering to jointly calibrate the model and localise the single molecule under the assumption that the molecule is static on the optical axis, or at least their movement on the optical axis has an insignificant effect on the photon rate. In contrast to their modelling assumption, we follow the approach of \cite{szalai2021three} to incorporate movement in all three coordinates. The molecule's depth effects the photon arrival rate and arrival locations, the former through a state (depth) dependent photon detection rate $\lambda(X_t)$ and the latter through the 3D Born and Wolf model. 
\begin{figure}[t!]
     \centering
     \begin{subfigure}[b]{0.48\textwidth}
         \centering
         \includegraphics[width=\textwidth]{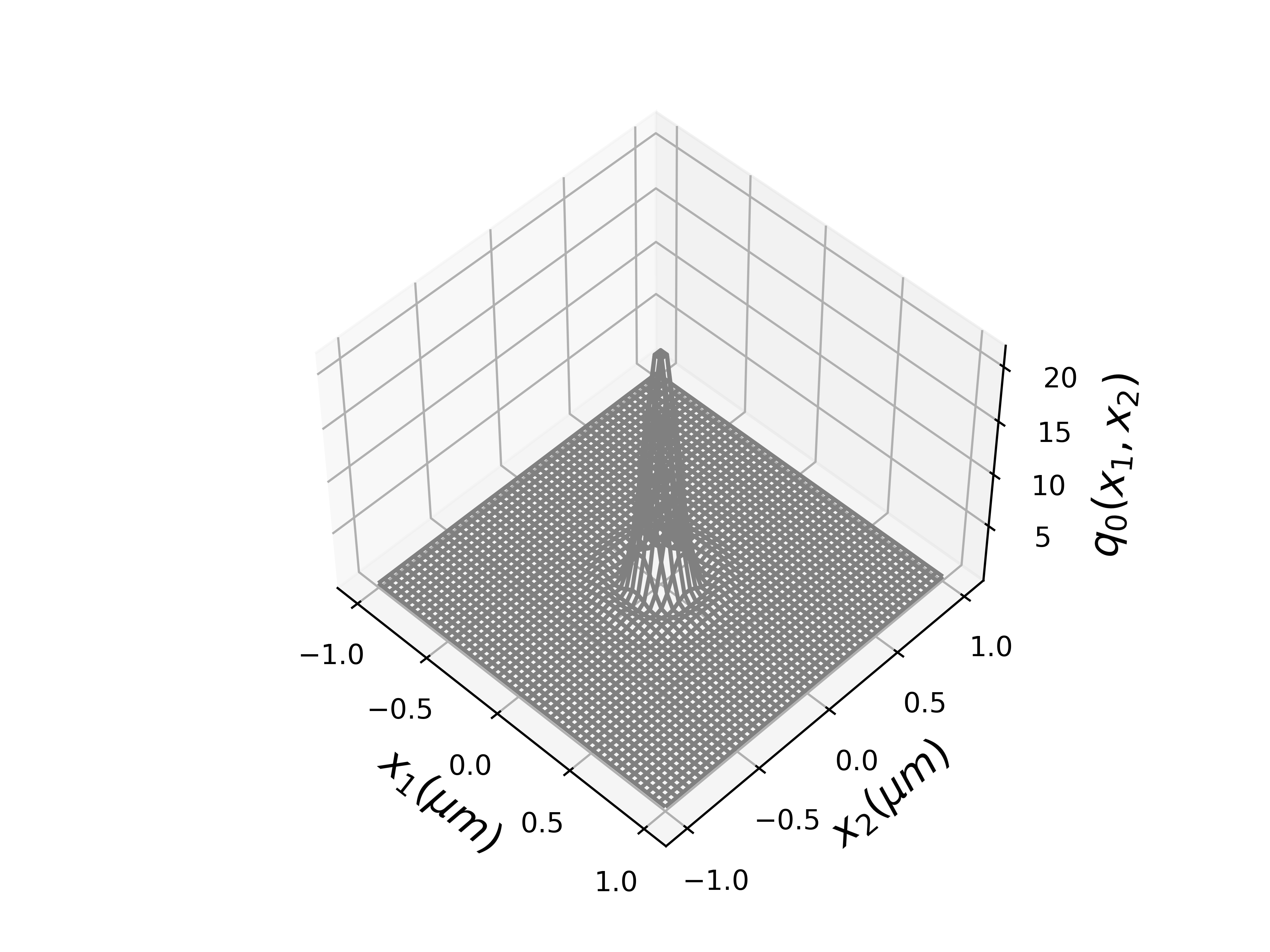}
         \caption{$x_3=0\mu m$}
         \label{fig:BW_pdf_z0}
     \end{subfigure}
     \hfill
     \begin{subfigure}[b]{0.48\textwidth}
         \centering
         \includegraphics[width=\textwidth]{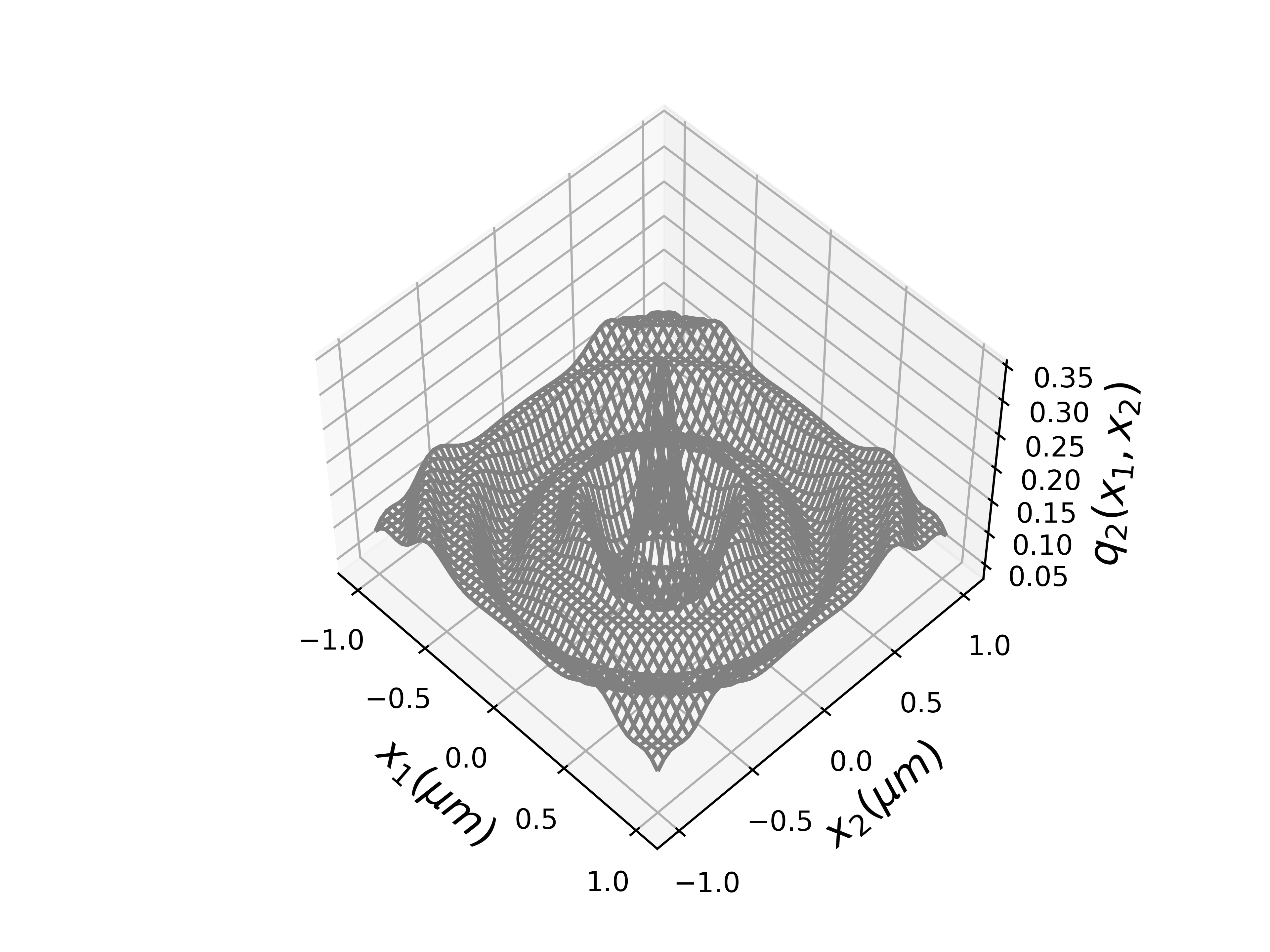}
         \caption{$x_3=2\mu m$}
         \label{fig:BW_pdf_z2}
     \end{subfigure}
     \begin{subfigure}[b]{0.48\textwidth}
         \centering
         \includegraphics[width=\textwidth]{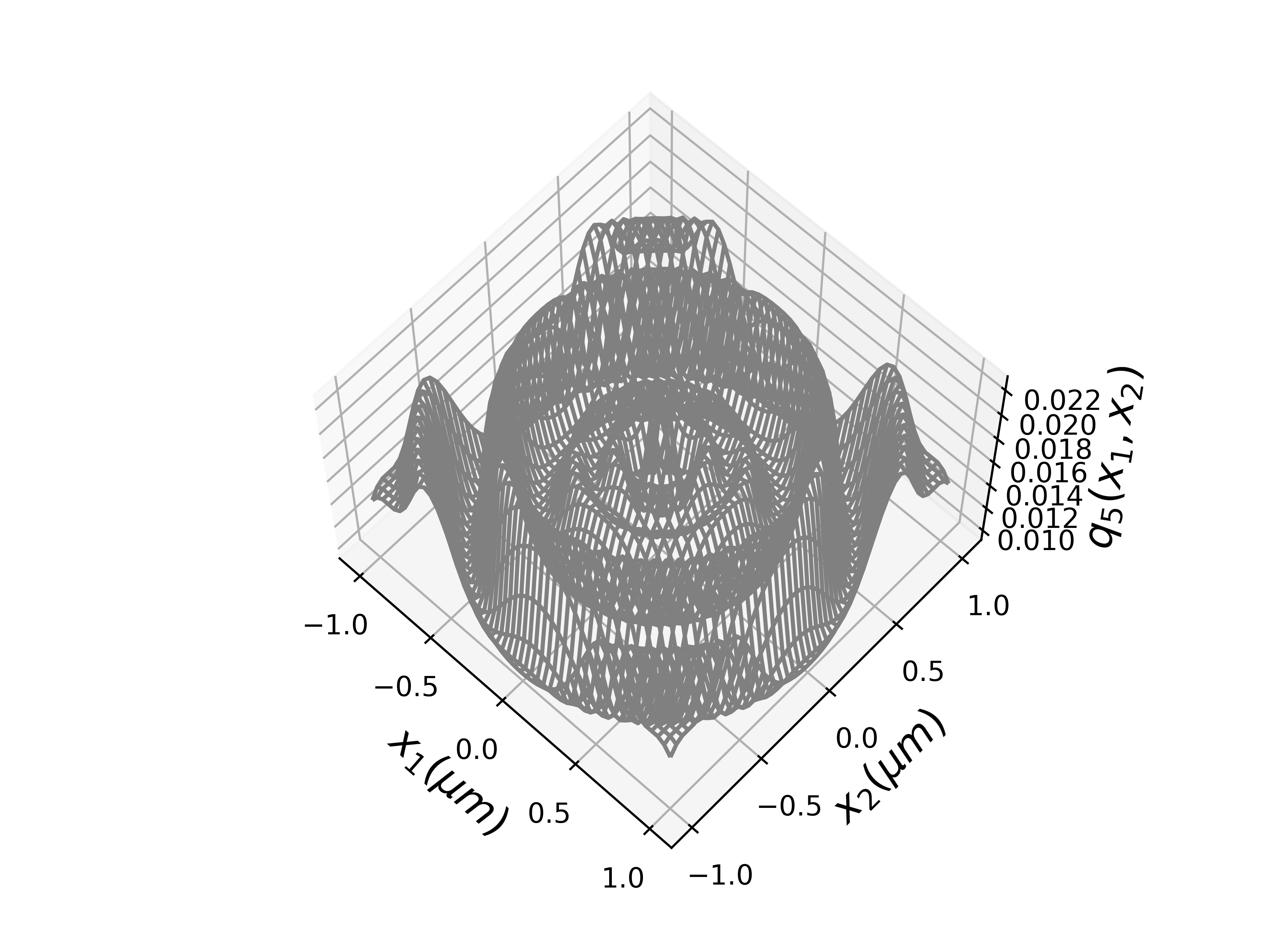}
         \caption{$x_3=5\mu m$}
         \label{fig:BW_pdf_z5}
     \end{subfigure}
     \begin{subfigure}[b]{0.48\textwidth}
         \centering
         \includegraphics[width=\textwidth]{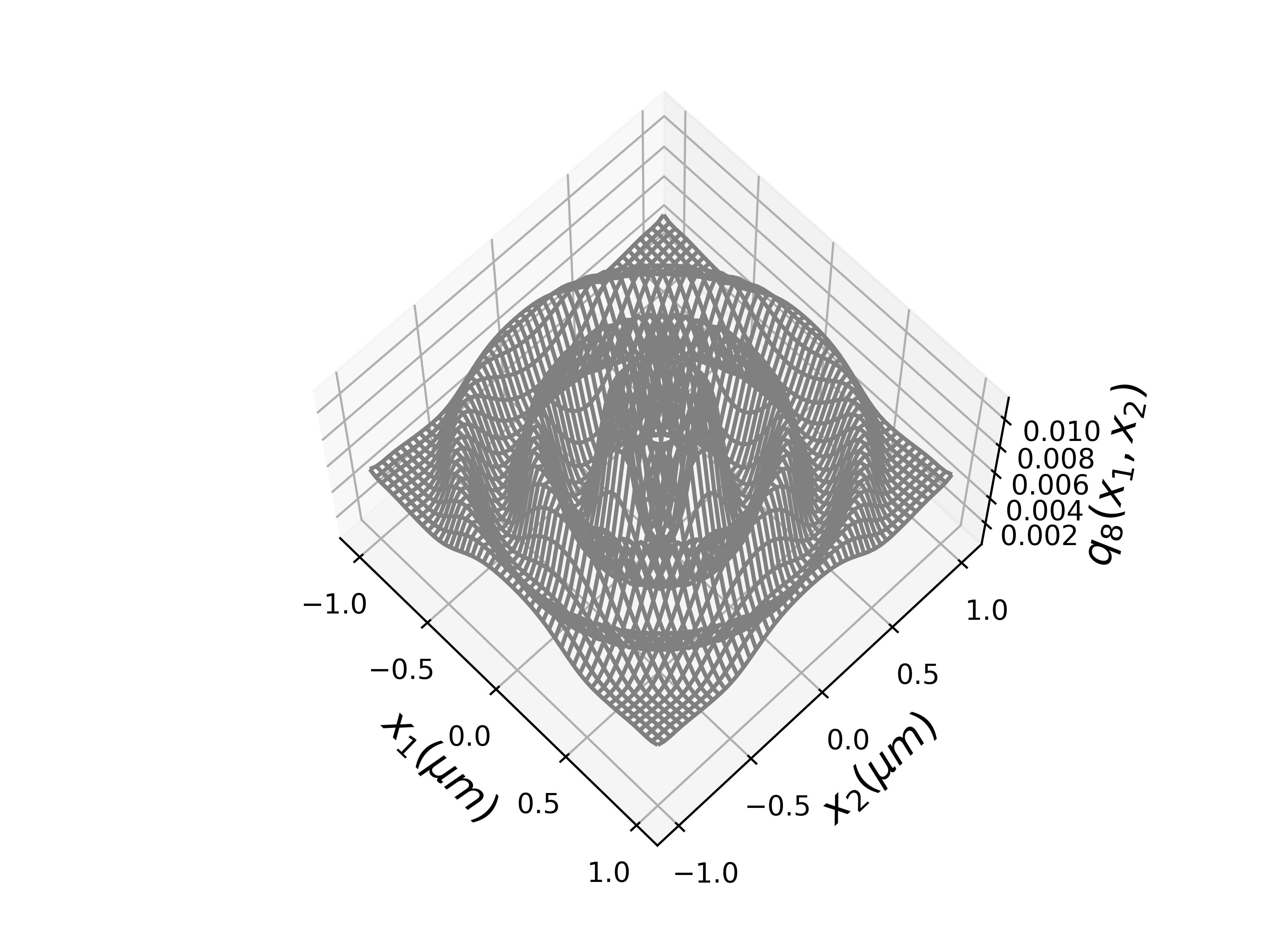}
         \caption{$x_3=8\mu m$}
         \label{fig:BW_pdf_z8}
     \end{subfigure}
     \caption{Born and Wolf point spread function at different defocus levels. Mesh representations are shown for \eqref{eq: Born and WOlf profile} at different defocuses $x_3$, computed with wavelength $\lambda_e=0.52\mu m$, numerical aperture $n_\alpha = 1.4$, refractive index of the objective lens immersion medium $n_0=1.515$. The $x_3$ values shown correspond to point source positions (a) $x_3=0\mu m$ (in focus), (b) $2\mu m$, (c) $5\mu m$ and (d) $8\mu m$.}
     \label{fig: BW_pdf_different defocused levels}
\end{figure}
\begin{figure}[t!]
    \centering
    \includegraphics[width=0.7\linewidth]{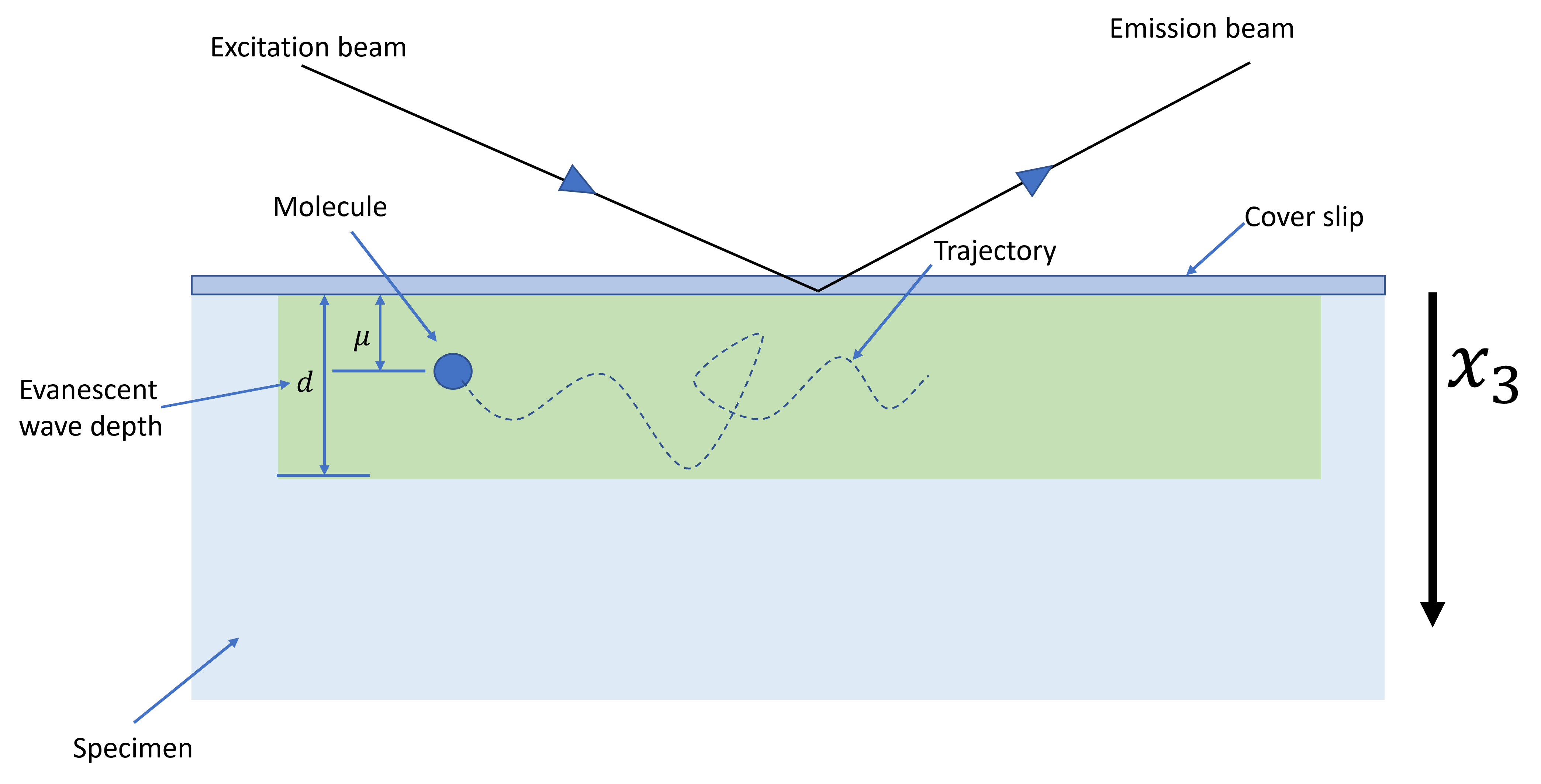}
    \caption{Total internal reflection fluorescence. An excitation beam with wavelength $\lambda_v$ traveling from a high refractive index ($n_1$) medium into a lower refractive index ($n_2$) medium is totally internal reflected at a planar interface. The reflection generates a thin layer of light in the lower refractive index medium, whose intensity decays exponentially along the $x_3$-axis with a characteristic constant $d$. While the molecule moves in the field, it is illuminated/excited and thus fluoresces. Parameter $\mu$ along $x_3$-direction is the mean of $X_{3,t}$ which the molecule diffuses about.}
    \label{fig:schematics for total internal reflection}
\end{figure}
Figure \ref{fig:schematics for total internal reflection} illustrates the total internal reflection phenomenon. For comprehensive review, see \cite{axelrod2001total}. This phenomenon is caused by the surface-associated evanescent electromagnetic field that is generated when an excitation beam is internally reflected at a planar interface between two transparent mediums with different refractive indices, $n_1$ and $n_2$. 
\cite{szalai2021three} have shown that the photon detection rate $\lambda(\cdot)$ decays exponentially along $x_3$-axis,  $\lambda(x_3) = \lambda_0 \exp(-\frac{x_3}{d}),
    \label{eq: intensity function}$
where $\lambda_0$ denotes the rate of photons emitted by a fluorophore at $x_3=0$. 
Under this set-up, one could consider $x_3$-movement of a single molecule as a \textit{reflected diffusion process} since the molecule should reflect when it encounters the cover slip. However, there is no practical approach available for exact simulation of this reflected process. For instance, \cite{blanchet2018exact}'s approach requires infinite expected running time, hence impractical when adapting into the particle filtering algorithms. 
We adopt a simpler approach by assuming standard O-U process which would be suitable if the molecule does not encounter a boundary (i.e. either cover slip $x_3=0$ or its maximum  depth $d$) over its observation period, e.g. if the observation periods are short and/or the molecule is diffusing about a mean depth $\mu$ in the middle of the cell with large $\phi_3$ (i.e. stronger attraction to $\mu$), see Figure \ref{fig:schematics for total internal reflection}.
\begin{figure}[t!]
     \centering
     \begin{subfigure}[b]{0.48\textwidth}
         \centering
         \includegraphics[width=\textwidth]{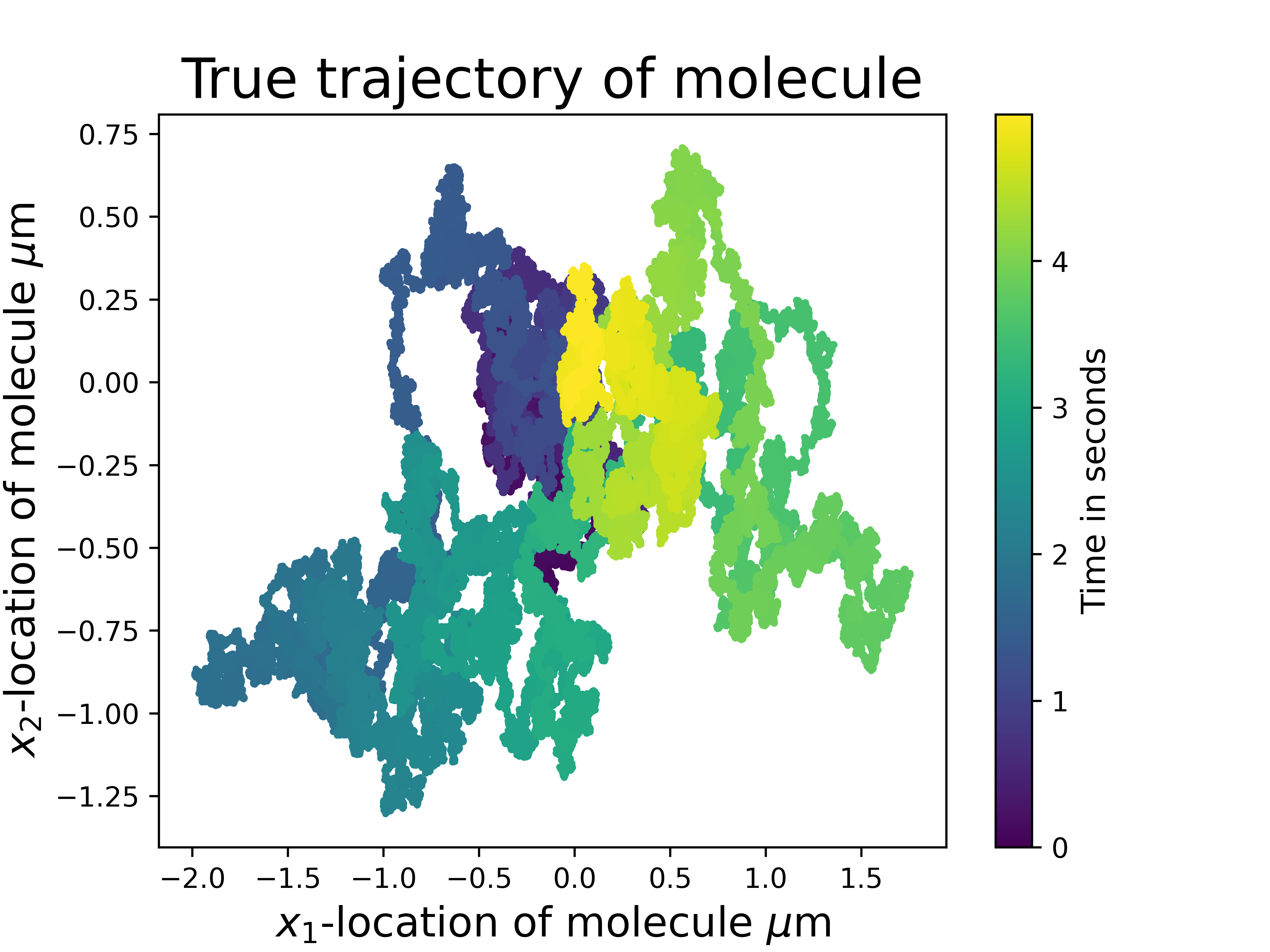}
         \caption{}
         \label{fig:true trajectory}
     \end{subfigure}
     \hfill
     \begin{subfigure}[b]{0.48\textwidth}
         \centering
         \includegraphics[width=\textwidth]{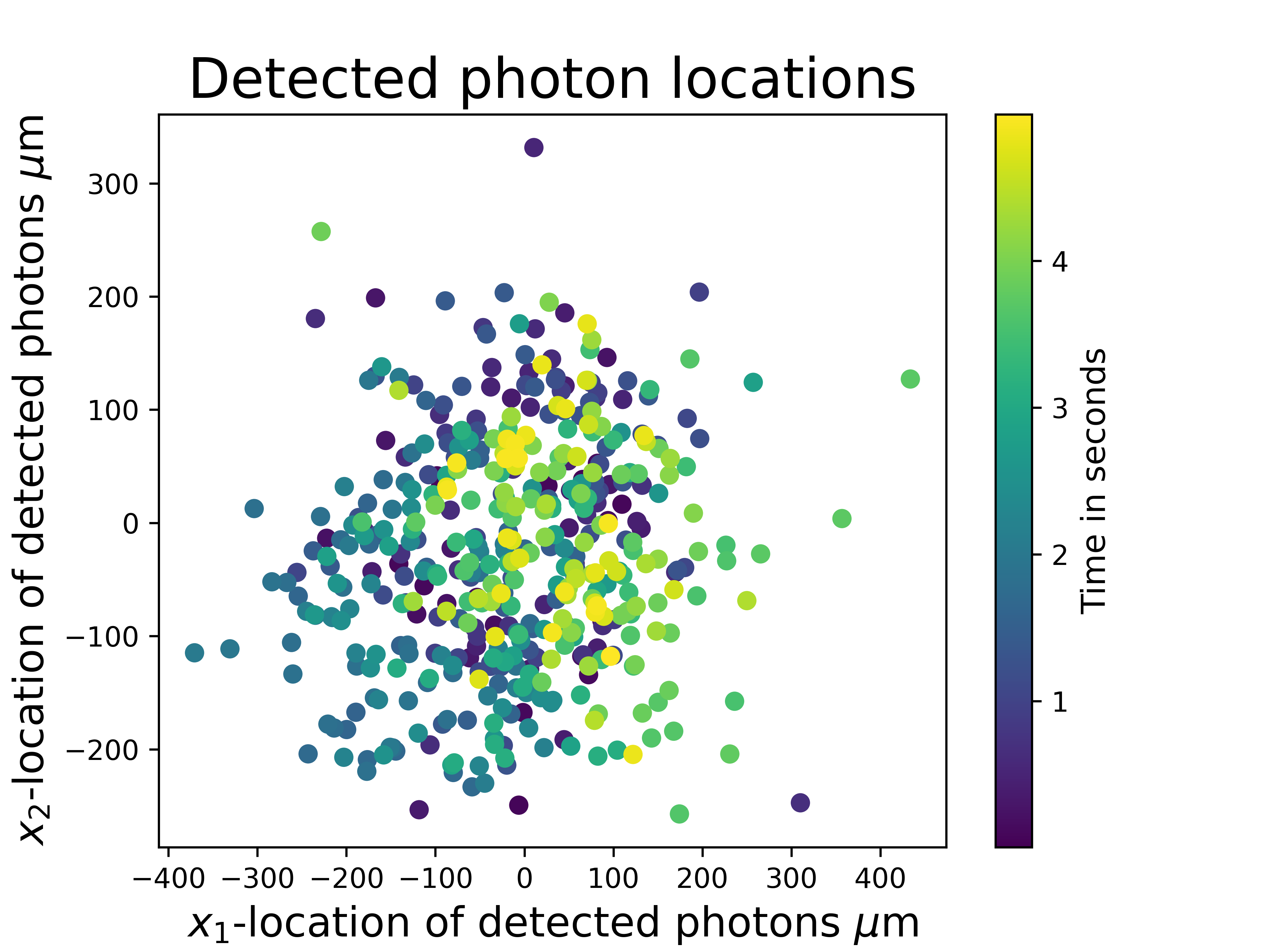}
         \caption{}
         \label{fig:BW model chapter 3}
     \end{subfigure}
     \begin{subfigure}[b]{0.48\textwidth}
         \centering
         \includegraphics[width=\textwidth]{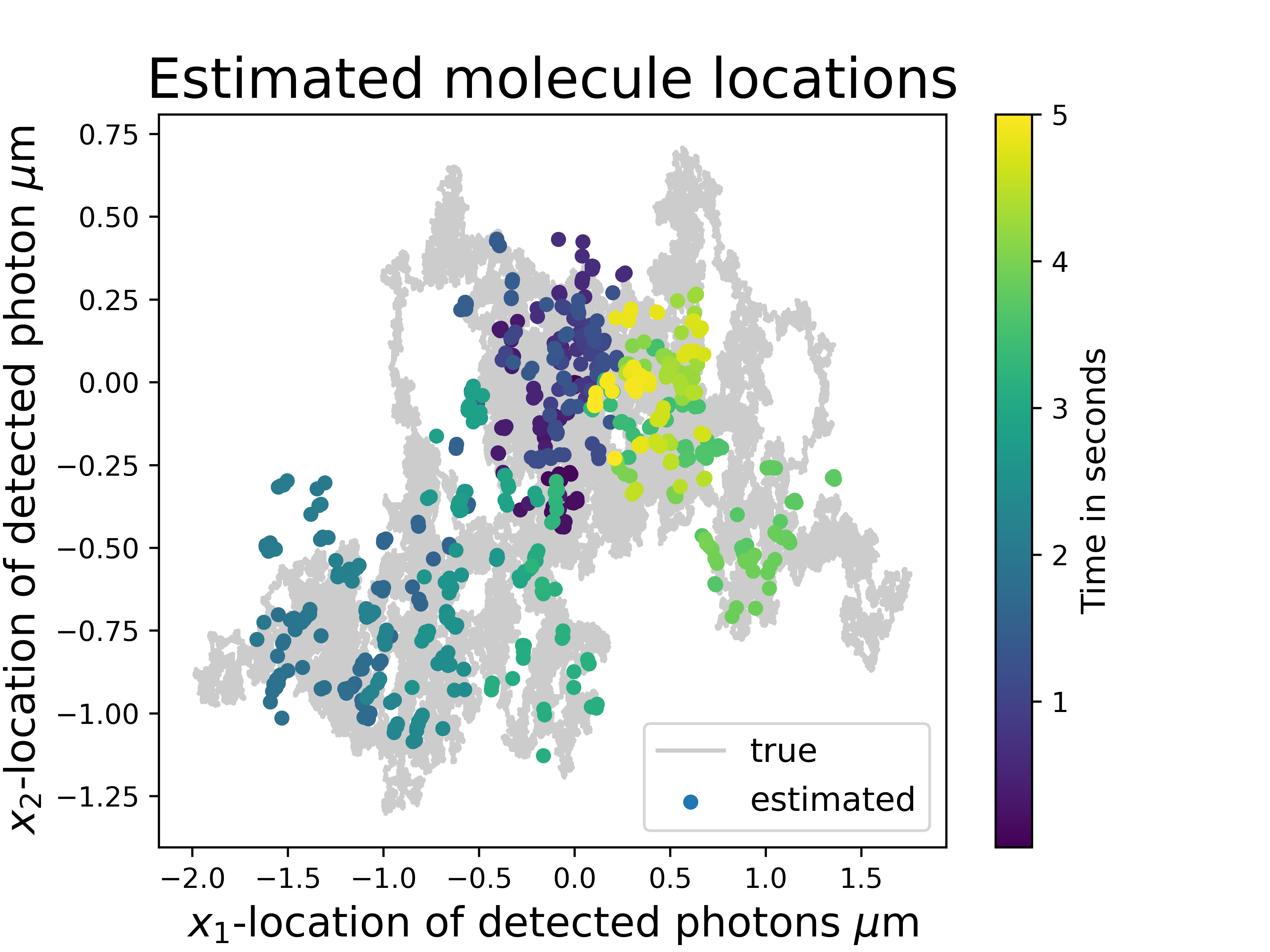}
         \caption{}
         \label{fig:BW_estimated}
     \end{subfigure}
     \begin{subfigure}[b]{0.48\textwidth}
         \centering
         \includegraphics[width=\textwidth]{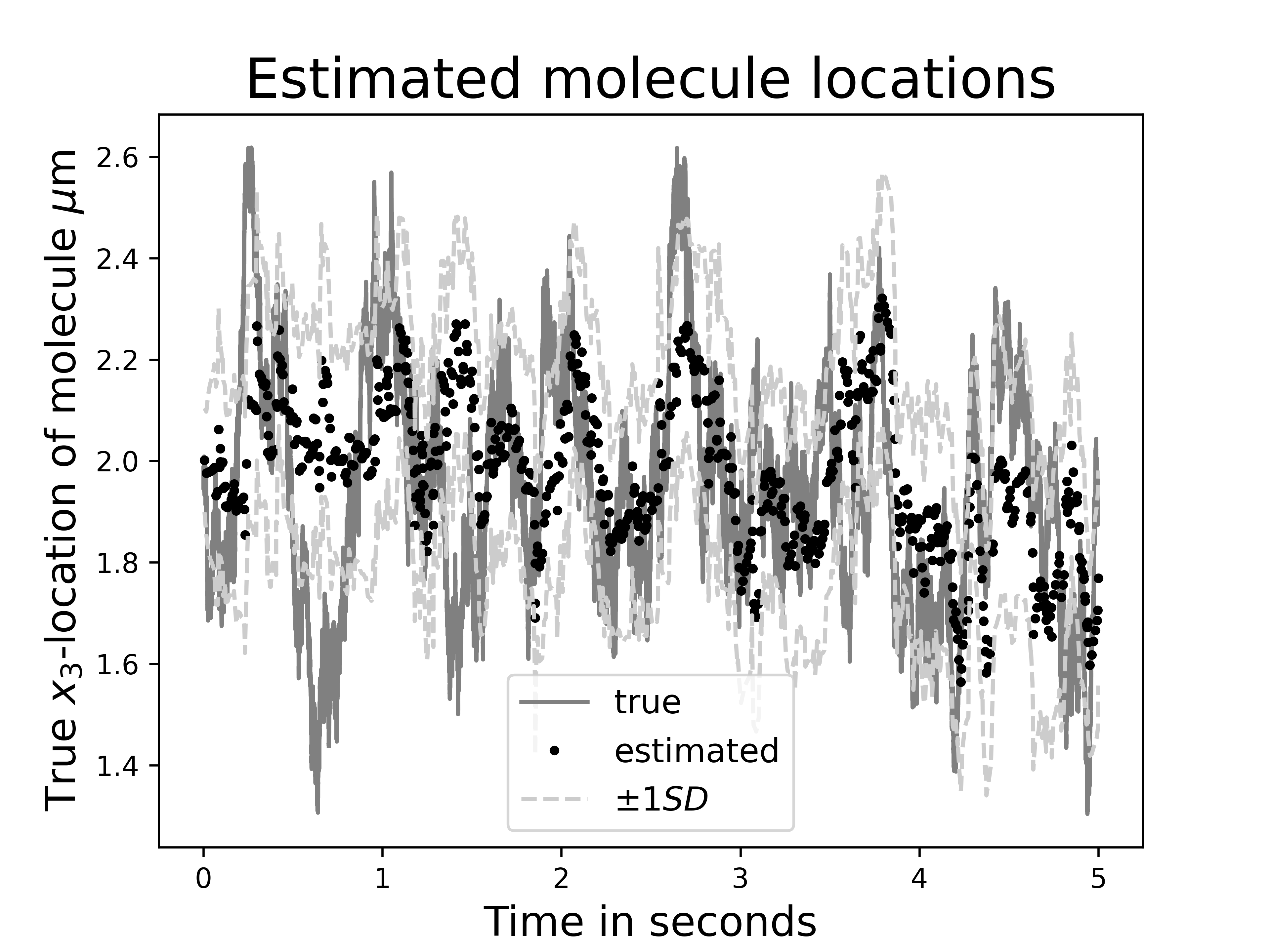}
         \caption{}
         \label{fig:x_3_estimated_BW}
     \end{subfigure}
     \caption{(a) True trajectory of a molecule; (b) observed photon locations; (c) estimated $(x_1,x_2)$ molecule locations and (d) true $x_3$ molecule locations and estimated location.}
\end{figure}
In Figure \ref{fig:true trajectory}, we plot the true trajectory of the molecule for the numerical studies, which is generated by the SDE in \eqref{eq: f12 density} with parameters $\{\phi=(\phi_1, \phi_2, \phi_3)^\top=(1, 1, 4)^\top, \mu=(\mu_1, \mu_2, \mu_3)^\top=(0, 0, 2)^\top, p_0=1/(2\phi)\}$ for the time interval [0, 5.0]. The initial variance $p_0$ is set to be the stationary variance. We used a thinning algorithm (Algorithm \ref{alg: thinning algorithm} detailed in Appendix \ref{sec: thinning algorithm}) to generate the observation times by using the intensity function (see \eqref{eq: intensity function}) with parameters $\{\lambda_0=100, d=20 \mu m\}$. Given these observation times, we generate the observed photon locations with the photon distribution profile given by \eqref{eq: photon distribution profile} and \eqref{eq: Born and WOlf profile}. The associated parameters are $\{M=m\mathbb{I}_{2\times 2}, m=100, n_\alpha=1.4, \lambda_e=0.52 \mu\text{m},n_0=1.515, \sigma_a^2=49\times 10^{-4} \mu\text{m}^2\}$ and corresponding data set of photon locations is shown in Figure \ref{fig:BW model chapter 3}.
The colours in Figure \ref{fig:true trajectory} indicate time, and as described by the legend, the colours lighten with the progress of time. Figure \ref{fig:BW_estimated} shows the mean of the estimated $(X_1,X_2)$ locations of the molecule found using Algorithm \ref{alg: bootstrap particle filter}, which does track the true trajectory. Figure \ref{fig:x_3_estimated_BW} shows true $X_3$ position of the molecule, the mean of the estimated $X_3$ position (also obtained using Algorithm \ref{alg: bootstrap particle filter}), along with the standard deviations. During periods when there are no observations, the estimated $X_3$ value is larger, as we would expect since this corresponds to a smaller photon arrival intensity function. Figure \ref{fig:x_3_estimated_BW} also shows that large $X_3$ state values degrade the estimation quality (which is more clearly seen for the $X_3$ values). This is due to the Born and Wolf observation model for an out-of-focus molecule, see \eqref{eq: Born and WOlf profile}. Additional results of this phenomenon are reported in Appendix \ref{sec: additional experiments}. 

\begin{figure}[t!]
    \centering
    \includegraphics[width=0.6\linewidth]{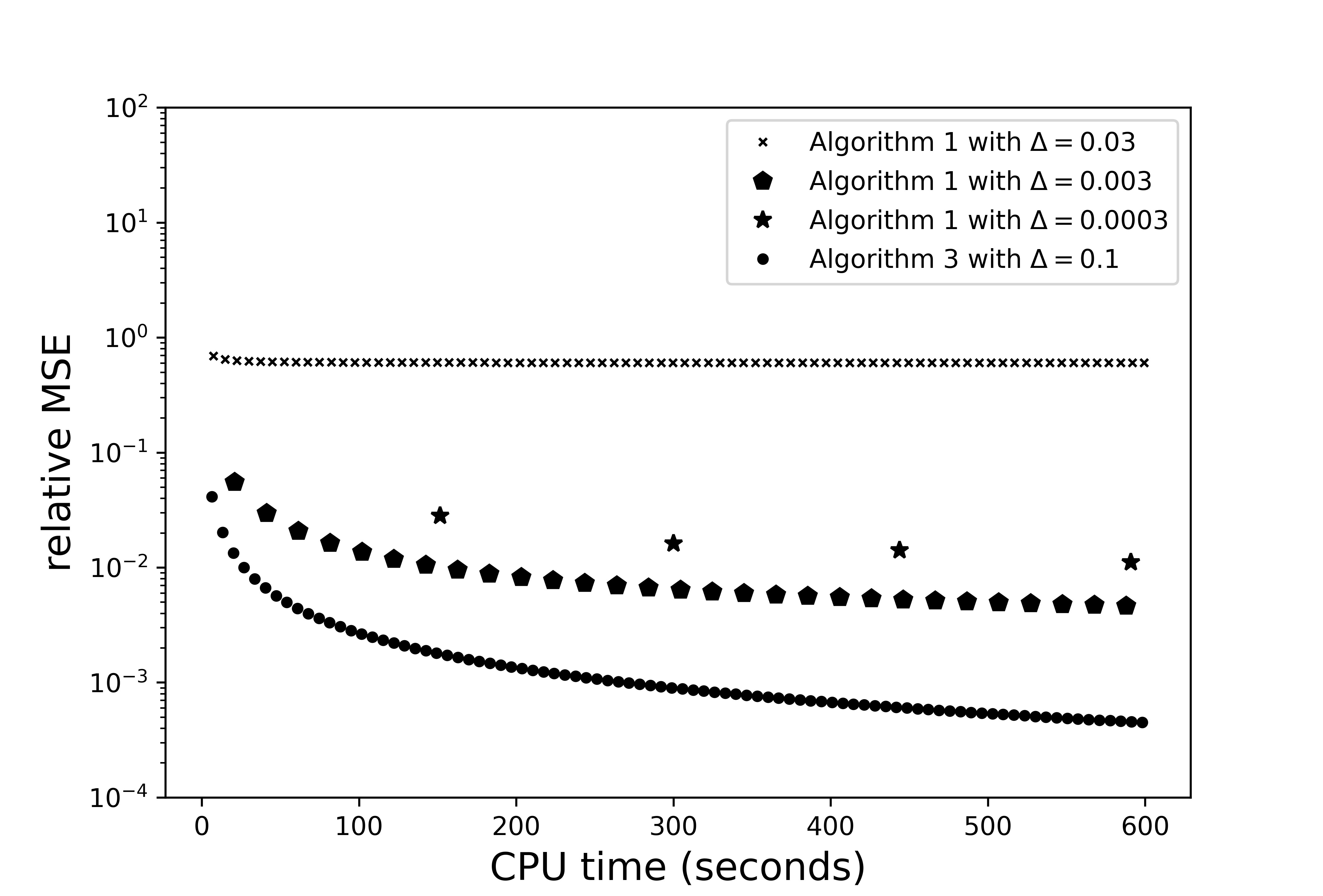}
    \caption{Plot of $\text{rMSE}_{\epsilon}$ of likelihood estimates versus CPU time for Algorithm \ref{alg: bootstrap particle filter} with $\Delta=0.1$ and Algorithm \ref{alg: bootstrap particle filter_discrete_time_version} with $\Delta=0.03, 0.003$ and 0.0003. rMSE of likelihood estimates based on the average estimate of $\mathcal{L}$ given by $10^6$ runs of Algorithm \ref{alg: bootstrap particle filter} with each run using $N=10^6$ particles, and the design choice for $l$ and $\eta$ given in Section \ref{sec:design_choice}. Here 600 CPU seconds corresponds to  $2\times 10^5$ particles for Algorithm \ref{alg: bootstrap particle filter}.}
    \label{fig:comparison dt}
\end{figure}

Figure \ref{fig:comparison dt} shows a comparison between the estimation quality of Algorithms \ref{alg: bootstrap particle filter_discrete_time_version} (with $\Delta=0.03, 0.003$ and $0.0003$) and \ref{alg: bootstrap particle filter} (with $\Delta=0.1$) for this single molecule example. For both methods, the CPU time is increased by increasing the number of particles used in the algorithms. The superiority of Algorithm \ref{alg: bootstrap particle filter} is apparent as measured using relative MSE of the likelihood estimate. As can be seen, the best $\Delta$ for 
Algorithm \ref{alg: bootstrap particle filter_discrete_time_version} is not necessarily the smallest one for a fix CPU budget. This also has practical consequences. For high frequency data, there will be potentially more time intervals between observation arrivals of which are much smaller than $\Delta$. This would lead to a small bias for Algorithm \ref{alg: bootstrap particle filter_discrete_time_version} although at a higher computational cost. Further reducing the bias, the relative MSE would be dominated by the variance if the CPU budget only permits a smaller number of particles.
\subsubsection{Model calibration using PMCMC}
Estimating the parameters of the molecular dynamics is also important in single molecule studies. \cite{d2022limits} calibrate the model
using using maximum likelihood estimation after discretising the path integral. In contrast, we use the particle marginal Metropolis-Hastings (PMMH) algorithm \cite{andrieu2010particle} to sample from the posterior density $p(\theta,x_{0:m}|y_{0:m})$ where $\theta=(\phi, \mu)$. Data is simulated from from the model $\{\phi=(1,1,4)^\top, \mu=(0,0,2)^\top, p_0=(0.5,0.5,0.125)^\top\}^\top$ with intensity function parametrised by $\{\lambda_0=25,d=20\}$. The parameters of the Born and Wolf model remain the same as before. Precise estimation of $\phi_3$ requires a longer time series as it is weakly identified and we use 350 observations collected in the time interval $[0,15s]$. The following independent priors are used: $\phi_3\sim\mathcal{U}(0,10)$ and $\mu_3\sim\mathcal{U}(0,10)$. ($\mathcal{U}$ denotes the continuous uniform distribution.) We used a normal random-walk Metropolis Hastings proposal with initial covariance $0.1\times \mathbb{I}_{2\times2}$ to update the parameters jointly. The continuous covariance adaptation scheme of \cite{haario2001adaptive} is adopted in the PMMH algorithm. We chose the following three experimental settings:
\begin{itemize}
    \item Experiment 1 (low CPU budget): $\mathcal{C}=1.5s$, which only allows a coarse time discretisation, which coincides with the time of arrivals of data for both Algorithms \ref{alg: bootstrap particle filter_discrete_time_version} and \ref{alg: bootstrap particle filter}. This forces large $\Delta$ in $\mathcal{L}_\Delta$ for Algorithm \ref{alg: bootstrap particle filter_discrete_time_version}. $N$ is adjusted accordingly so that CPU budget is the same for both algorithms.
    \item Experiment 2 (larger CPU budget): $\mathcal{C}=2.5s$ permits a finer time discretisation than the observation arrival times. The best $\Delta$ and $N$ (within the CPU budget) are chosen for Algorithm \ref{alg: bootstrap particle filter_discrete_time_version} using the procedure outlined in Section \ref{sec: benchmark}. (Employing a larger CPU budget allows a smaller $\Delta$ than Experiment 1 in $\mathcal{L}_\Delta$.) For Algorithm \ref{alg: bootstrap particle filter}, we used $\Delta=0.01$ and its CPU cost adjusted $N$.
    \item Experiment 3 (effective sample size based comparison): $\mathcal{C}=2.0s$, the effective sample size (ESS) for PMMH using  Algorithm \ref{alg: bootstrap particle filter} with $\Delta=0.01$ is found, and then the best $\Delta$ and $N$ are chosen for Algorithm \ref{alg: bootstrap particle filter_discrete_time_version} while ensuring its ESS matches that of Algorithm \ref{alg: bootstrap particle filter}. The ESS, which measures the number of `independent samples,' is
    \begin{equation*}
        \text{ESS}=\frac{M}{-1+2\sum_{t=0}^K\left(\rho_{2t}+\rho_{2t+1}\right)}
    \end{equation*}
    where $\rho_t$ is estimated autocorrelation at lag $n$ and $K$ is the last integer for which the sum in the sum bracket is still positive.  The general trend is that the ESS of PMMH with Algorithm \ref{alg: bootstrap particle filter_discrete_time_version} increases when $\Delta$ is increased (i.e. larger $N$ for the fixed CPU budget), although the estimation will be more biased for Algorithm \ref{alg: bootstrap particle filter_discrete_time_version}. 
\end{itemize}
We ran the algorithms for $10^5$ with a $10^4$ burn-in iterations. Figure \ref{fig:pmmh} displays the estimates of the marginal posterior densities for $\mu_3$ and $\phi_3$ for all three experiments.
\begin{figure}[t!]
    \centering
    \includegraphics[width=0.4\linewidth]{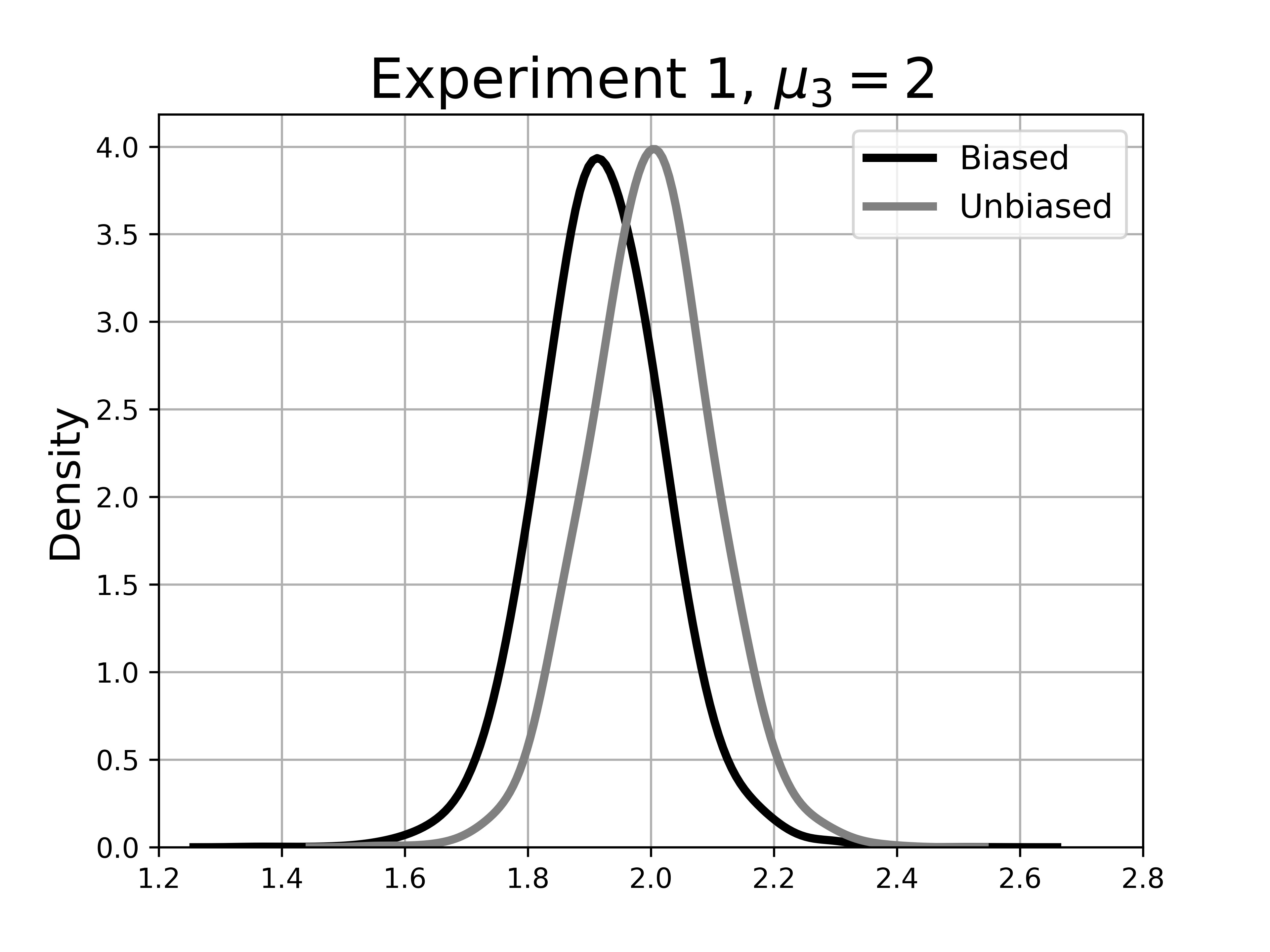}
    \includegraphics[width=0.4\linewidth]{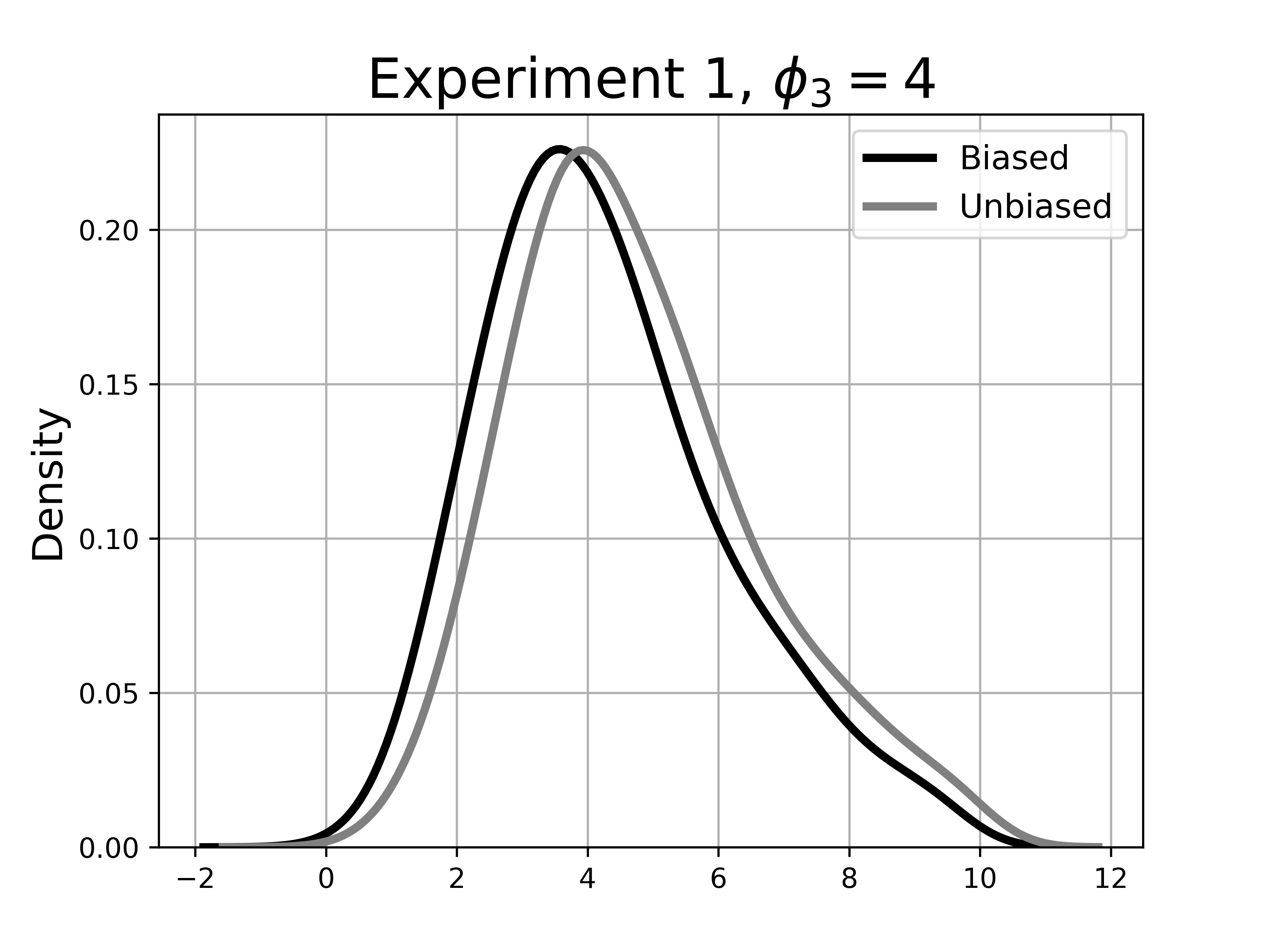}
    \includegraphics[width=0.4\linewidth]{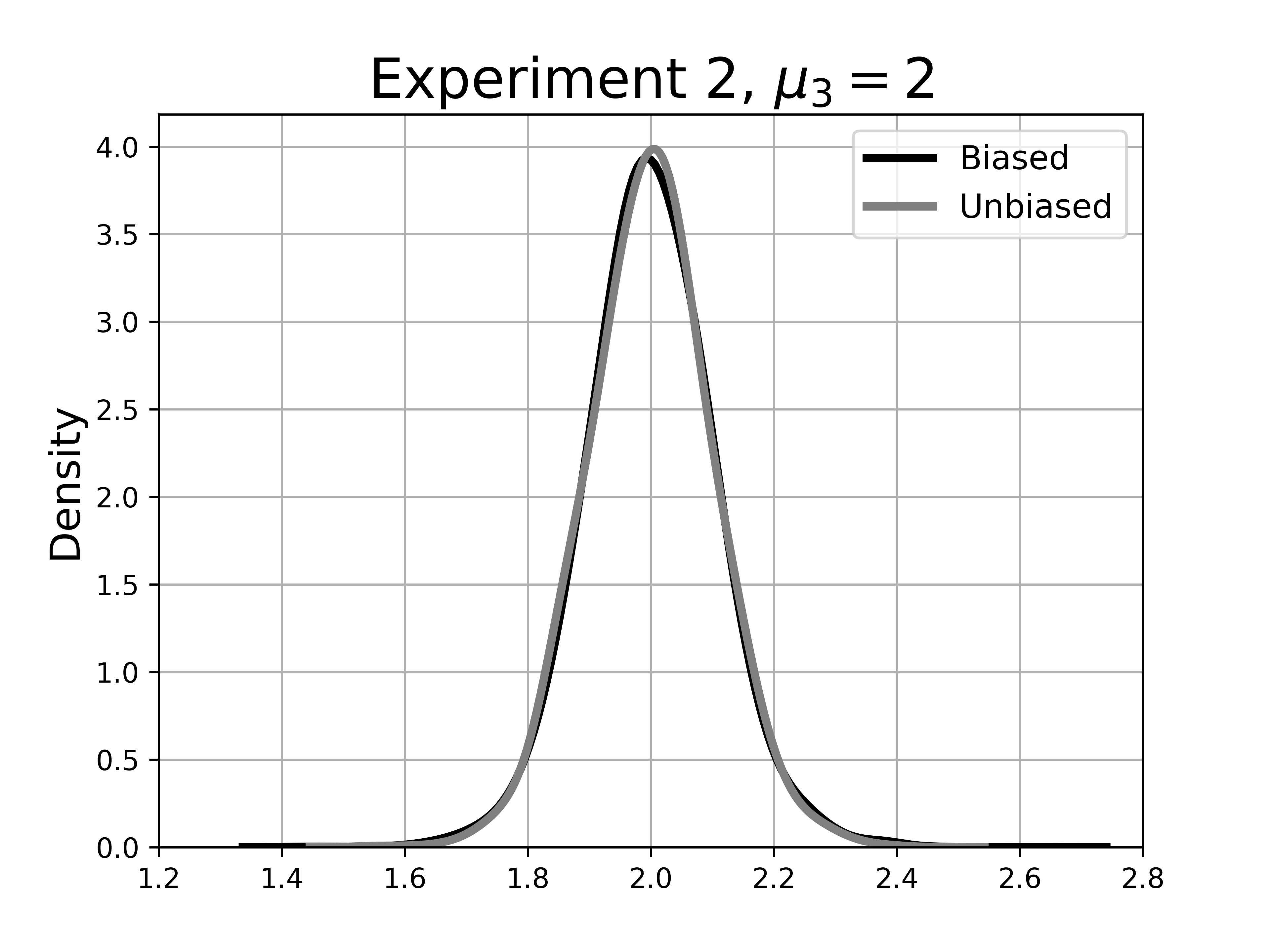}
    \includegraphics[width=0.4\linewidth]{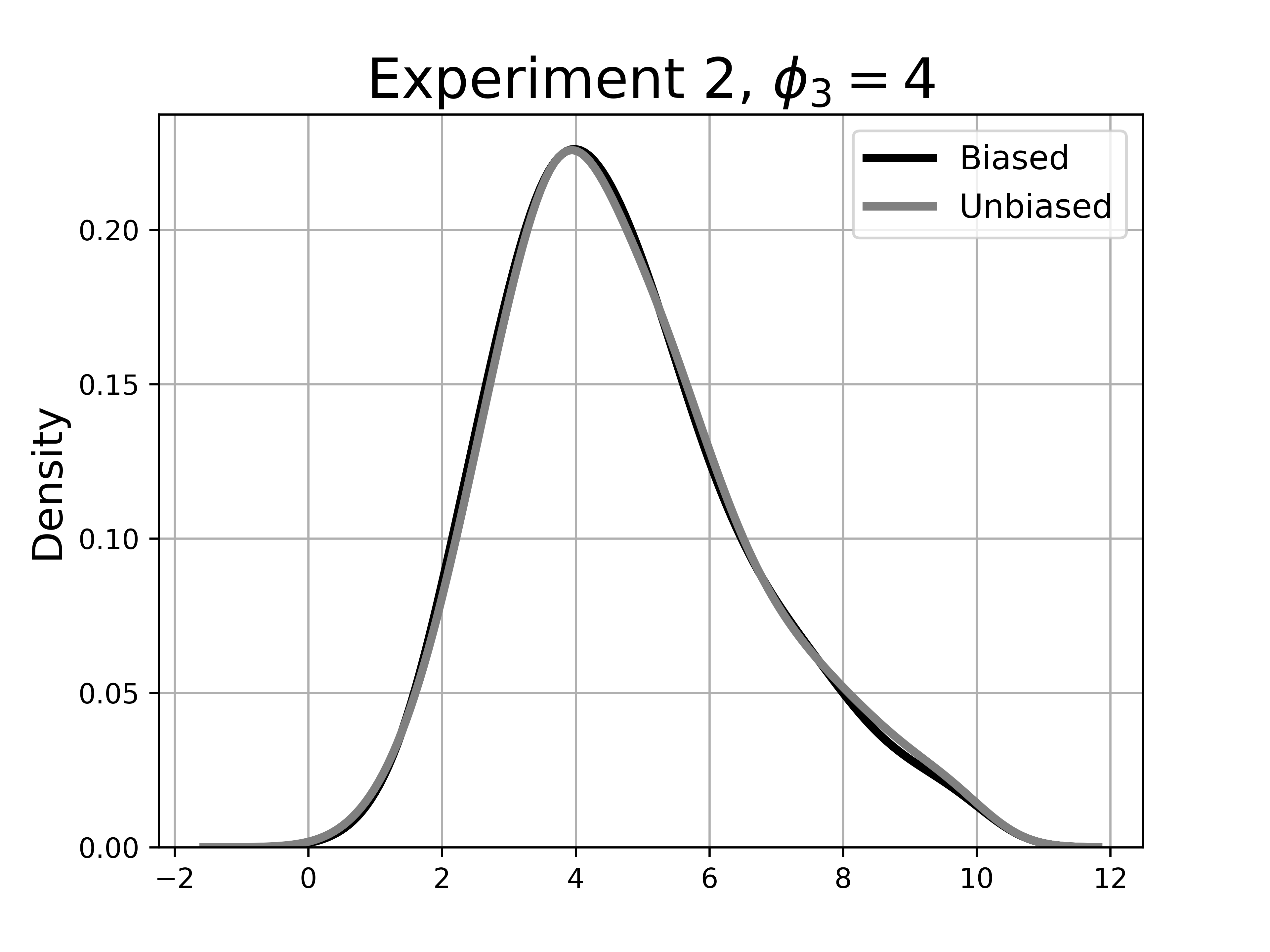}
    \includegraphics[width=0.4\linewidth]{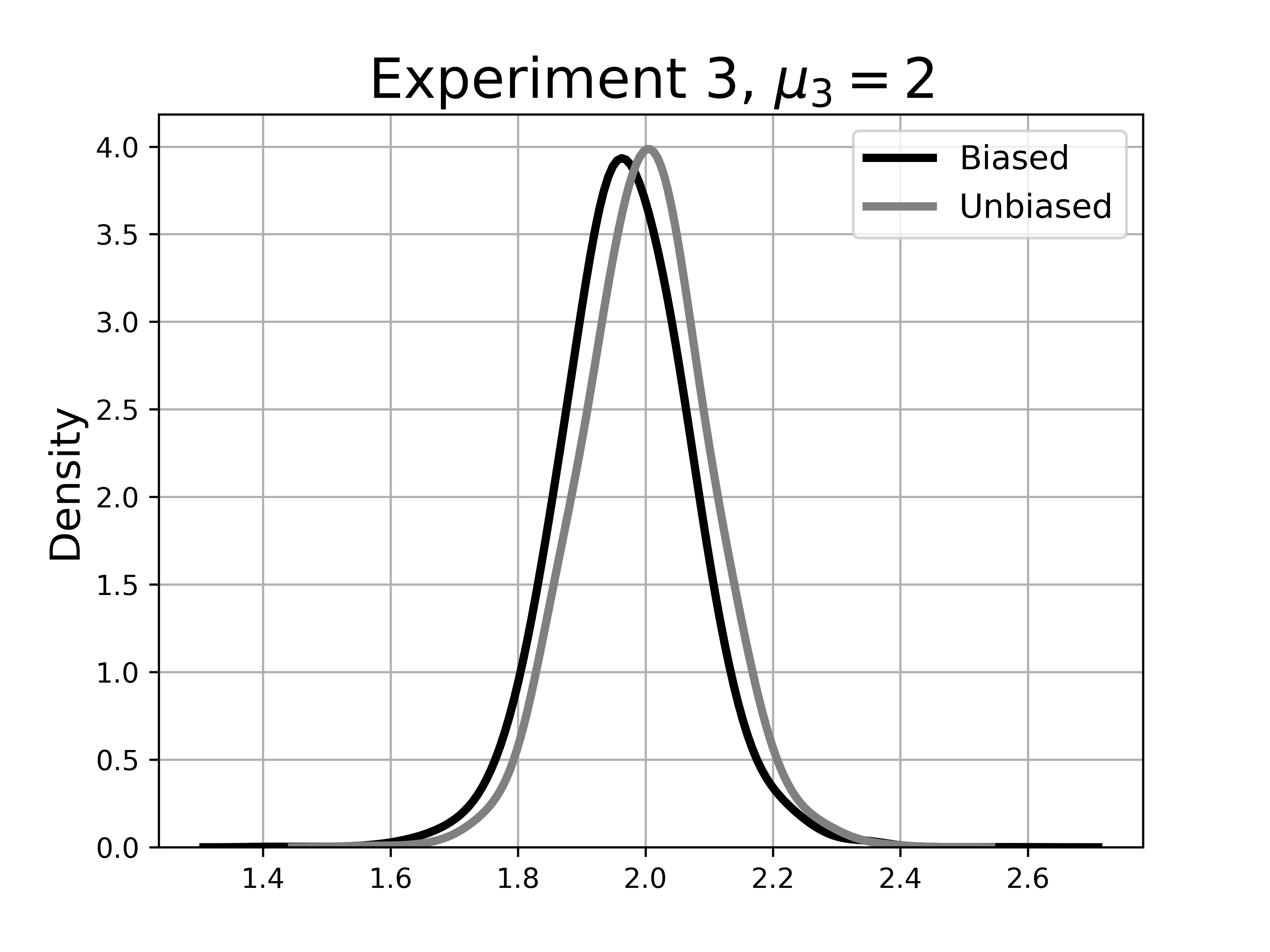}
    \includegraphics[width=0.4\linewidth]{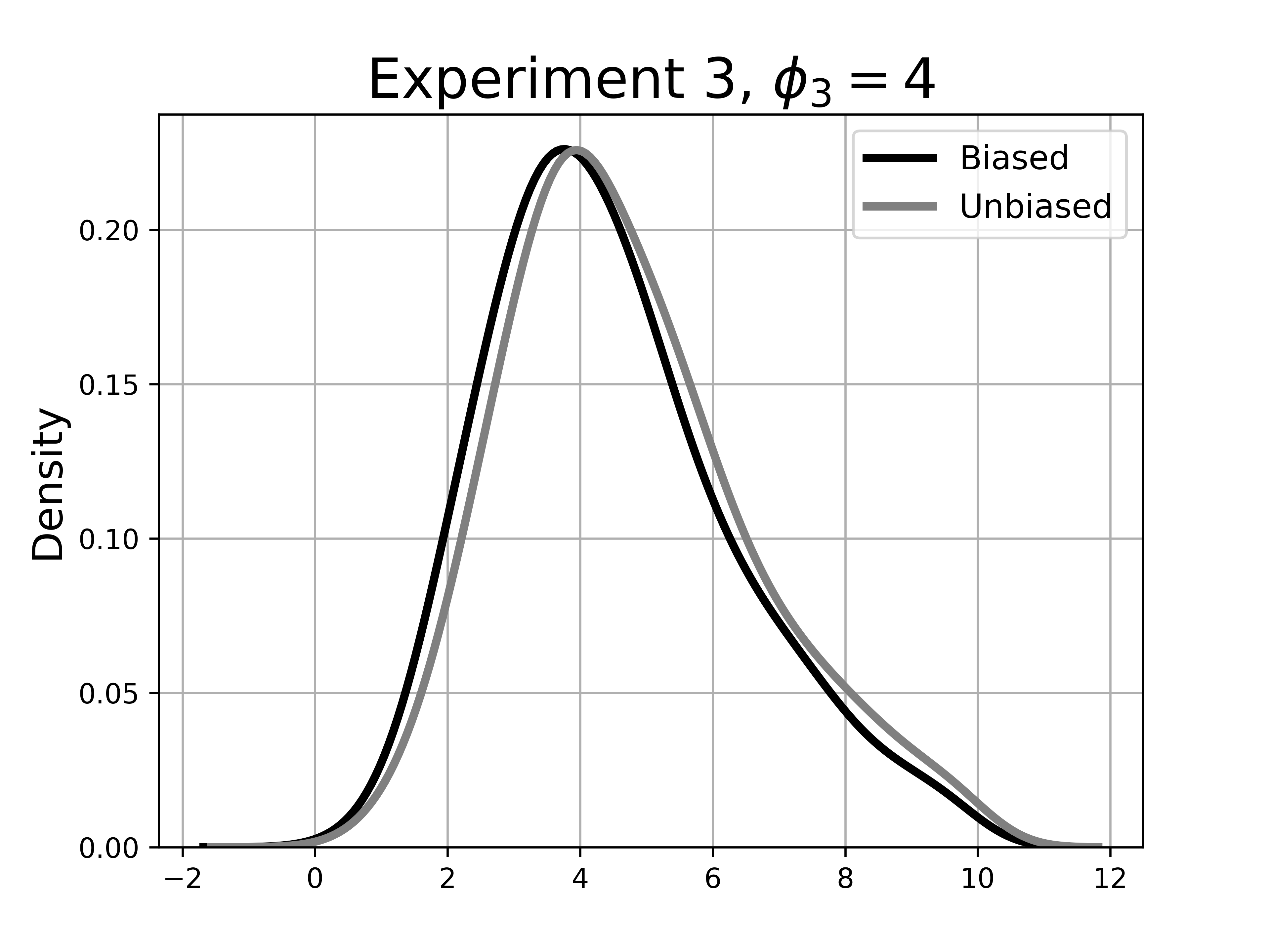}
    \caption{PMMH experiments: posterior densities $p(\mu_3|y_{t_1:t_{n_p}})$ and $p(\phi_3|y_{t_1:t_{n_p}})$ for experiment 1 in top panel, experiment 2 in middle panel and experiment 3 in bottom panel.}
    \label{fig:pmmh}
\end{figure}
Experiment 1 shows that Algorithm \ref{alg: bootstrap particle filter} effectively removes bias in the estimation of parameters while the performance of Algorithm \ref{alg: bootstrap particle filter_discrete_time_version} is compromised by the limited CPU budget. Experiment 2 shows that given sufficient budget Algorithm \ref{alg: bootstrap particle filter_discrete_time_version} is able to achieve almost the same performance of Algorithm \ref{alg: bootstrap particle filter}. Comparison between experiments 1 and 2 show that the unbiased posterior produced by Algorithm \ref{alg: bootstrap particle filter} remain unaffected by the amount of available CPU budget. Experiment 3 is carried out under the setting that both algorithms produce equally ``efficient'' MCMC samples for a fix CPU budget, the result shows that Algorithm \ref{alg: bootstrap particle filter_discrete_time_version} still yields some amount of bias to the estimation.

\section{Conclusion}

In this paper we addressed the problem of smoothing and model calibration for a
partially observed diffusion with a Cox process observation model. The
intractable likelihood was estimated using the positive part of a Poisson
estimate, for the path integrals within, embedded within particle filtering.
The probability of encountering a negative Poisson estimate in one complete
particle filtering pass through the data was strongly controlled by adjusting
$\eta=\mathcal{O}(\Delta l)$. As such, due to the rarity of the occurrence of a
negative estimate -- which triggers the particle weight truncation -- the time
discretisation error which biases conventional particle implementations such as
in \cite{d2022limits} is effectively removed in this work. 
The numerical results showed that our proposed particle method (Algorithm
\ref{alg: bootstrap particle filter}) outperforms the conventional
(discretisation-based) particle filter in terms of relative MSE, ours decaying
with order $\mathcal{O}(\mathcal{C}^{-1})$ compared to
$\mathcal{O}(\mathcal{C}^{-\frac{2}{3}})$ where $\mathcal{C}$ is the
computational budget. Our particle filter was then applied to a challenging
three-dimensional single molecule microscopy example to estimate both the
trajectory of the moving molecule and to calibrate the model. The bias in the
posterior distribution for the model parameters computed using a conventional
implementation like \cite{d2022limits} was clearly illustrated, whereas in ours
it  was not discernible. Although the bias in the conventional method can be
reduced by employing a smaller $\Delta$, the time discretisation interval, this
not only requires significant additional CPU time, it will also prohibit the
application of backward sampling steps in particle filtering. This is a
direction in which this work could taken further, which is to define a forward
filtering backward sampling implementing of our method. In the context of
diffusions, this is a challenging problem, see \cite{yonekura2022online} for a
recent study.
\appendix
\section{Bridge density for linear Gaussian diffusion}
Consider the following stochastic differential equation (SDE),
\label{sec: bridge density}
\begin{equation}
    dX_t=b(t,X_t)dt+\sigma(t,X_t)dW_t
    \label{eq: linear SDE}
\end{equation}
where $X_t$ is an $n$-dimensional diffusion process, $W_t$ is an $m$-dimensional standard Brownian motion for $m\leq n$. For $b(t,X_t):=b_0+b_1(t)X_t$ and $\sigma(t,X_t):=\sigma(t)\in\mathbb{R}^{n\times m}$, the solution to \eqref{eq: linear SDE} at discrete time points $t_0<t_1<\ldots$ is given by \cite{jazwinski2007stochastic, evans2012introduction}
\begin{equation}
    X_{t_{i+1}}=\Phi(t_i,t_{i+1})X_{t_i}+a(t_i,t_{i+1})+\int_{t_i}^{t_{i+1}}\Phi(t_i,t)\sigma(t)dW_t
\end{equation}
where the fundamental matrix function $\Phi\in\mathbb{R}^{n\times n}$ satisfies the following for all $s,t,u\geq t_0$
\begin{align*}
    \frac{d\Phi(s,t)}{dt}=b_1(t)\Phi(s,t),\quad
    \Phi(t,t)=\mathbb{I}_{n\times n}, \quad\Phi(s,t)\Phi(t,u)=\Phi(s,u),
\end{align*}
the vector $a(t_i, t_{i+1})\in \mathbb{R}^n$ is given by
$
    a(t_i, t_{i+1})=\int_{t_i}^{t_{i+1}} b_0\Phi(t_i,t)dt.
$
Therefore the transition density $f_{t_i,t_{i+1}}(x'|x)$ can be expressed as a Gaussian as follows,
\begin{align*}
    f_{t_i,t_{i+1}}(x'|x) &:= \frac{\exp\left(-\frac{1}{2}\left(x'-\mu(x,t_i,t_{i+1})\right)^\top R^{-1}(t_i,t_{i+1})\left(x'-\mu(x, t_i,t_{i+1})\right)\right)}{\sqrt{|2\pi R^{-1}(t_i,t_{i+1})|}},\\
    &\propto \frac{\exp\left(-\frac{1}{2}\left(x-\hat{\mu}(x', t_i,t_{i+1})\right)^\top \hat{R}^{-1}(t_i,t_{i+1})\left(x-\hat{\mu}(x', t_i,t_{i+1})\right)\right)}{\sqrt{|2\pi\hat{R}^{-1}(t_i,t_{i+1})|}}\\&:=\hat{f}_{t_i,t_{i+1}}(x|x'),
\end{align*}
where
\begin{align*}
    \mu(x,t_i,t_{i+1})&:=\Phi(t_i,t_{i+1})x+a(t_i,t_{i+1}),\\
    \hat{\mu}(x',t_i,t_{i+1})&:=\Phi^{-1}(t_i,t_{i+1})x'-\Phi^{-1}(t_i,t_{i+1})a,\\
    R(t_i,t_{i+1})&:=\int_{t_i}^{t_{i+1}} \Phi(t_i,t)\sigma(t)\sigma^\top (t)\Phi^\top(t_i,t)dt,\\
    \hat{R}(t_i, t_{i+1})&:=\int_{t_i}^{t_{i+1}}\Phi^{-1}(t,t_{i+1})\sigma(t)\sigma^\top(t)\Phi^{-\top}(t,t_{i+1})dt.
\end{align*}
Assume $s<\tau<t$, then in addition to sampling $p(x_\tau|x_s)$ exactly, one can also sample $X_\tau\sim p(x_\tau|x_s,x_t)$ exactly where
\begin{align*}
    &p(x_\tau|x_s,x_t)\propto f_{s,\tau}(x_\tau |x_s)f_{\tau, t}(x_t|x_\tau)\propto f_{s, \tau}(x_\tau|x_s) \hat{f}_{\tau,t}(x_\tau|x_t)\\
    &\propto\text{N}\left(x_\tau;\left(R^{-1}(s,\tau)+\hat{R}^{-1}(\tau,t)\right)^{-1}\left(R^{-1}(s,\tau)\mu(x_s,s,\tau)+\hat{R}^{-1}(\tau,t)\hat{\mu}(x_t,\tau,t)\right),\right.\\&\left.\quad\left(R^{-1}(s,\tau)+\hat{R}^{-1}(\tau,t)\right)^{-1}\right).
\end{align*}
\section{Proof of Lemma 1}
\label{sec: lemma 1}
The following propositions will be used in the final proof.
\begin{proposition}
\label{prop: supremum of one side brownian bridge}
Let $X$ be the Brownian motion which starts at $X_0=x_0$, then the following equality holds for any $a>0$:
\begin{align*}
    &\Pr\Big(\sup_{0\leq s\leq \Delta}X_s-X_0\geq a| X_\Delta=x_\Delta\Big)\\=&\begin{cases}\exp\Big\{-\frac{2a}{\Delta}[a-(x_\Delta-x_0)]\Big\}, &a>x_\Delta-x_0\\
    1, &a\leq x_\Delta-x_0
    \end{cases}
\end{align*}
\end{proposition}
\begin{proof}
\label{pro: proposition 1}
Define $\tau_a$ as the hitting time of $a$ as follows,
\begin{equation*}
    \tau_a = \inf\left\{s\in\left[0,\Delta\right]|X_s-X_0= a\right\}
\end{equation*}
A hitting time is also a stopping time. Then by applying the reflection principle (please refer to Theorem 2.19 of \cite{morters2010brownian}), the process $\{X^*:t\geq 0\}$, called Brownian motion $\{X_t:t\geq 0\}$ reflected at $\tau_a$, defined by 
\begin{align*}
    X^*_t &= X_t\mathbb{I}_{t\leq \tau_a}+(2X_{\tau_a}-X_t)\mathbb{I}_{t>\tau_a}\\
    &=X_t\mathbb{I}_{t\leq \tau_a}+(2a+2x_0-X_t)\mathbb{I}_{t>\tau_a}
\end{align*}
is also a Brownian motion. Thus,
\begin{align*}
    &\text{Pr}\Big(\tau_a\leq \Delta, X_\Delta\in[x_\Delta,x_\Delta+dx]\Big)\\=&\text{Pr}\Big(X_\Delta^*\in[2a+2x_0-x_\Delta-dx,2a+2x_0-x_\Delta]\Big)\\
    =&\text{Pr}\Big(X_\Delta^*-X_0\in[2a+x_0-x_\Delta-dx,2a+x_0-x_\Delta]\Big)\\
    =&\frac{dx}{\sqrt{2\pi \Delta}}\exp\Big\{-\frac{\left[2a-\left(x_\Delta-x_0\right)\right]^2}{2\Delta}\Big\}.
\end{align*}
Note that,
\begin{align*}
    &\text{Pr}\Big(X_\Delta\in[x_\Delta,x_\Delta+dx]\Big)\\=&\text{Pr}\Big(X_\Delta-X_0\in[x_\Delta-x_0,x_\Delta-x_0+dx]\Big)\\
    =&\frac{dx}{\sqrt{2\pi \Delta}}\exp\Big\{-\frac{(x_\Delta-x_0)^2}{2\Delta}\Big\}.
\end{align*}
Division between two equations above concludes the proof.
\end{proof}
\begin{proposition}
Let $X$ be the Brownian motion which starts at $X_0=x_0$, then the following equality holds for any $a>0$:
\begin{align*}
    &\Pr\left(\inf_{0\leq s\leq \Delta}X_s-X_0\leq -a|X_\Delta=x_\Delta\right)\\=&\begin{cases}
    \exp\left\{-\frac{2a}{\Delta}\left[a+\left(x_\Delta-x_0\right)\right]\right\}, &a> -(x_\Delta-x_0)\\1,&a\leq-(x_\Delta-x_0).\end{cases}
\end{align*}
\end{proposition}
\begin{proof}
A similar approach as in Proof \ref{pro: proposition 1} but define $\tau_{-a}=\inf\{s\in[0,\Delta]|X_s-X_0=-a\}$ and apply the reflection principle by defining the Brownian motion $\{X^*:t\geq 0\}$, the Brownian motion $\{X_t:t\geq 0\}$ reflected at $\tau_{-a}$, formally defined by
\begin{equation*}
    X_t^*=X_t\mathbb{I}_{t\leq -\tau_a}+(-2a+2x_0-X_t)\mathbb{I}_{t\geq -\tau_a}.
\end{equation*}

\end{proof}
\begin{proposition}
\label{prop: two sided distribution of Brownian bridge}
Let $X$ be defined as in Proposition \ref{prop: supremum of one side brownian bridge}, then the following inequality holds:
\begin{align*}
    &\Pr\Big(\sup_{0\leq s\leq \Delta}|X_s-X_0|\geq a| X_\Delta=x_\Delta\Big)\\\leq&\begin{cases}2\exp\Big\{-\frac{2a}{\Delta}[a-|x_\Delta-x_0|]\Big\}, &a>|x_\Delta-x_0|\\
    1, &a<|x_\Delta-x_0|
    \end{cases}
\end{align*}
\end{proposition}
\begin{proof}
\begin{align*}
&\text{Pr}\Big(\sup_{0\leq s\leq \Delta}|X_s-X_0|\geq a|X_\Delta=x_\Delta\Big)\\\leq&\text{Pr}\Big(\sup_{0\leq s\leq \Delta}X_s-X_0\geq a|X_\Delta=x_\Delta\Big)+\text{Pr}\Big(\inf_{0\leq s\leq \Delta}X_s-X_0\leq -a|X_\Delta=x_\Delta\Big)\\
=&\begin{cases}\exp\left\{-\frac{2a}{\Delta}\left[a-(x_\Delta-x_0)\right]\right\}+\exp\left\{-\frac{2a}{\Delta}\left[a+(x_\Delta-x_0)\right]\right\} &a>|x_\Delta-x_0|\\1, &a\leq |x_\Delta-x_0|\end{cases}\\
\leq&\begin{cases}2\exp\left\{-\frac{2a}{\Delta}\left[a-|x_\Delta-x_0|\right]\right\} &a>|x_\Delta-x_0|\\1, &a\leq |x_\Delta-x_0|\end{cases}
\end{align*}
\end{proof}
\subsection*{Proof of Lemma 1}
\begin{proof}
\begin{align*}
    &\text{Pr}\left(E_1>0\Big\vert\kappa=k>0\right)\\=&\mathbb{E}\left\{\mathbb{I}\left[\left(\prod^k_{j=1}\left(1+\frac{\Delta}{\eta}\left(\lambda\left(X_{0}\right)\right)-\lambda\left(X_{\tau_j}\right)\right)\right)>0\right]\Big\vert\kappa=k,\tau_1,\ldots,\tau_k\right\}\\
    \geq &\mathbb{E}\left\{\mathbb{I}\left[\max_{j\in\{1,\ldots, k\}}\frac{\Delta}{\eta}|\lambda\left(X_0\right)-\lambda\left(X_{\tau_j}\right)|<1\right]\Big\vert\kappa=k,\tau_1,\ldots,\tau_k\right\}\\
    =&\text{Pr}\left(\max_{j\in\{1,\ldots, k\}}|\lambda\left(X_0\right)-\lambda\left(X_{\tau_j}\right)|\leq \frac{\eta}{\Delta}\Big\vert\kappa=k,\tau_1,\ldots,\tau_k\right).
\end{align*}
We can obtain an upperbound for $\text{Pr}(E_1<0|\kappa=k)$ by
\begin{align*}
    \text{Pr}\left(E_1<0\Big\vert\kappa=k\right)&\leq\text{Pr}\left(\max_{j\in\{1,\ldots, k\}}|\lambda\left(X_0\right)-\lambda\left(X_{\tau_j}\right)|\geq \frac{\eta}{\Delta}\Big\vert \kappa=k,\tau_1,\ldots,\tau_k\right)\\
    &\leq \text{Pr}\left(\max_{j\in\{1,\ldots, k\}}|X_0-X_{\tau_j}|\geq \frac{\eta}{\Delta l}\Big\vert \kappa=k,\tau_1,\ldots,\tau_k\right)
\end{align*}
where we assume $\lambda(\cdot)$ is an $l-$Lipschitz function.
\begin{align}
&\text{Pr}\left(\max_{j\in\{1,\ldots, k\}}|X_0-X_{\tau_j}|\geq \frac{\eta}{\Delta l}\Big\vert \kappa=k,\tau_1,\ldots,\tau_k\right)\nonumber\leq \text{Pr}\left(\sup_{0\leq s\leq \Delta}|X_0-X_s|\geq \frac{\eta}{\Delta l}\right).
\end{align}
Applying Proposition \ref{prop: two sided distribution of Brownian bridge}, we have
\begin{equation*}
    \text{Pr}\left(E_1<0\vert \kappa>0\right)\leq\begin{cases} 2\exp\left\{-\frac{\frac{2\eta}{\Delta l}\left(\frac{\eta}{\Delta l}-|x_\Delta-x_0|\right)}{\Delta}\right\},&\frac{\eta}{\Delta l}\geq|x_\Delta-x_0|\\
     1,&\frac{\eta}{\Delta l}\leq|x_\Delta-x_0|\end{cases}
\end{equation*}
\end{proof}
\section{Expectation of the probability bound}
\label{sec: expectation of the probability bound}
This section establishes the unqualified bound (\eqref{eq:avg_bound}). The goal is to determine the following expectation for $Y=|X_\Delta-X_0|$ where $X_\Delta-X_0\sim\mathcal{N}(0,\Delta)$, (thus $Y$ is a half-normal random variable):
\begin{align*}
&\mathbb{E}\left\{\mathbb{I}\left[E<0\right]|\kappa>0\right\}\\
=&\mathbb{E}\left\{\mathbb{I}\left[E<0\right]\times \mathbb{I}\left[Y<\frac{\eta}{\Delta l}\right]+\mathbb{I}\left[E<0\right]\times \mathbb{I}\left[Y\geq\frac{\eta}{\Delta l}\right]|\kappa>0\right\}\\
    \leq&\mathbb{E}\left\{2\exp\left(-\frac{\frac{2\eta}{\Delta l}\left(\frac{\eta}{\Delta l}-Y\right)}{\Delta}\right)\mathbb{I}\left[Y< \frac{\eta}{\Delta l}\right]\right\}+\Pr\left(Y\geq \frac{\eta}{\Delta l}\right)\\
    =&\int_0^\infty 2\exp\left(-\frac{\frac{2\eta}{\Delta l}\left(\frac{\eta}{\Delta l}-y\right)}{\Delta}\right)\mathbb{I}\left[y< \frac{\eta}{\Delta l}\right]\times \frac{\sqrt{2}}{ \sqrt{\pi\Delta}}\exp\left(-\frac{y^2}{2\Delta}\right)dy+\Pr\left(Y\geq \frac{\eta}{\Delta l}\right)\\
    =&\int_0^\frac{\eta}{\Delta l} \frac{2\sqrt{2}}{\sqrt{\pi\Delta}}\exp\left(-\frac{2\eta^2}{\Delta^3l^2}\right)\times \exp\left(-\frac{\left(y-\frac{2\eta}{\Delta l}\right)^2-\frac{4\eta^2}{\Delta^2 l^2}}{2\Delta}\right)dy+\Pr\left(Y\geq \frac{\eta}{\Delta l}\right)\\
    =&\int_0^{\frac{\eta}{\Delta l}}\frac{2\sqrt{2}}{\sqrt{\pi\Delta}} \exp\left(-\frac{\left(y-\frac{2\eta}{\Delta l}\right)^2}{2\Delta}\right)dy + \Pr\left(\left\vert X_\Delta-X_0\right\vert\geq \frac{\eta}{\Delta l}\right) \\
    =&4\left[\Phi\left(\frac{2\eta}{\Delta^{\frac{3}{2}}l}\right)-\Phi\left(\frac{\eta}{\Delta^{\frac{3}{2}}l}\right)\right]+2\times\left(1-\Phi\left(\frac{\eta}{\Delta^{\frac{3}{2}}l}\right)\right)\\
    =&2+4\Phi\left(\frac{2\eta}{\Delta^{\frac{3}{2}}l}\right)-6\Phi\left(\frac{\eta}{\Delta^{\frac{3}{2}}l}\right).
\end{align*}
\section{Proof of Lemma 2}
\label{sec: lemma 2}
\begin{proof}
\begin{align}
    E_i &= \exp\left(-\Delta\lambda(X_{(i-1)\Delta})\right)\prod_{j=1}^{\kappa_i}\left(1+\frac{\lambda(X_{(i-1)\Delta})-\lambda(X_{\tau_j})}{l}\right)\nonumber\\
    &\leq \prod_{j=1}^{\kappa_i} \left(1+\left\vert X_{(i-1)\Delta}-X_{\tau_j}\right\vert\right)\nonumber\\
    &\leq\prod_{j=1}^{\kappa_i} \left(1+\max_{(i-1)\Delta \leq s\leq i\Delta}\left\vert X_s-X_{(i-1)\Delta}\right\vert \right) \nonumber
    \\
    &=    \left(1+\max_{0\leq s\leq \Delta}\left\vert B_s\right\vert \right)^{\kappa_i}=:F_i.
    \label{eq: estimator bound}
\end{align}
We truncate the Poisson estimate as $E^+_i=E_i\mathbb{I}_{A_i^c}$ and bound $\mathbb{I}_{A}\prod_{i=1}^mE_i$ as follows.
\begin{equation}
    \left\vert\mathbb{E}\left\{\mathbb{I}_A \prod_{i=1}^m E_i\right\}\right\vert\leq \mathbb{E}\left\{\prod_{i=1}^m E_i^2\right\}^{\frac{1}{2}}\mathbb{E}\left\{\mathbb{I}_A\right\}^{\frac{1}{2}}.
    \label{eq: omitted term bound}
\end{equation}
\newline
The term $\mathbb{E}\left\{\mathbb{I}_A\right\}^{\frac{1}{2}}$ can be bound using the union bound
\begin{equation}
    \mathbb{E}\left\{\mathbb{I}_A\right\}\leq \sum_{i=1}^m \mathbb{E}\left\{\mathbb{I}_{A_i}\right\}= m \mathbb{E}\left\{\mathbb{I}_{A_i}\right\}=m\times
    2\exp\left(-\frac{1}{2\Delta}\right)
    \label{eq: probability bound}
\end{equation}
The other term can be proved to be finite, i.e. $\mathbb{E}\left\{\prod_{i=1}^m E_i^2\right\}<\infty$.
Since the increment of Brownian motion $X$ is independent of each other, and using the inequality \eqref{eq: estimator bound}, one can show 
\begin{align*}
    \mathbb{E}\left\{\prod_{i=1}^mE_i^2\right\}\leq \mathbb{E}\left\{\prod_{i=1}^m F_i^2\right\}= \mathbb{E}\left\{F_i^2\right\}^m, \qquad \forall \quad i.
\end{align*}
It remains therefore to bound $\mathbb{E}\left\{F_i^2\right\}$. 
\begin{align*}
    \mathbb{E}\left\{F_i^2\right\}&=\mathbb{E}\left\{\left(1+\max_{0\leq s\leq \Delta}\left\vert B_s\right\vert \right)^{2\kappa_i}\right\}\\
    &= \mathbb{E}\left\{\sum_{k=0}^\infty\frac{(\Delta l)^ke^{-\Delta l}}{k!}\left(1+\max_{0\leq s\leq \Delta}\left\vert B_s\right\vert \right)^{2k}\right\}\\
    &= \exp(-\Delta l)\mathbb{E}\left\{\exp\left(\Delta l(1+\max_{0\leq s\leq \Delta} |B_s|)^2\right)\right\}\\
    &\leq \exp(-\Delta l)\mathbb{E}\left\{\exp\left(\Delta l(2+2\times\max_{0\leq s\leq \Delta} B_s^2)\right)\right\}
    \\
    &=\exp(\Delta l)\mathbb{E}\left\{\exp\left(2\Delta l\times\max_{0\leq s\leq \Delta} B_s^2\right)\right\}
    \\
    &=\exp(\Delta l) \times \int_0^{\infty}\Pr\left(\exp\left(2\Delta l\times\max_{0\leq s\leq \Delta} B_s^2\right)>w\right)dw\\
         &=\exp(\Delta l)\times\left[ 1+ \int_{1}^{\infty}\Pr\left(\exp\left(2\Delta l\times\max_{0\leq s\leq \Delta} B_s^2\right)>w\right)dw\right]\\
    &= \exp(\Delta l)\times\left[ 1+\int_{1}^\infty\Pr\left(\max_{0\leq s\leq \Delta} \left\vert B_s\right\vert >\sqrt{\frac{\log(w)}{2\Delta l}}\right)dw\right]\\
    &\leq\exp(\Delta l) \times\left[ 1+ \int_1^\infty 2\exp\left(-\frac{\log(w)}{4\Delta^2 l}\right)dw\right]
    \\
    &=\exp(\Delta l)\times\left[ 1+ \int_1^\infty 2w^{-\frac{1}{4\Delta^2 l}}dw\right]
    \\
    &= \exp(\Delta l) \times \left(\frac{1+4\Delta^2 l}{1-4\Delta^2 l}\right)
\end{align*}
where in the fourth last line we apply the inequality for running maximum of Brownian motion which starts at zero, i.e. $\text{Pr}\left(\max_{0\leq s\leq \Delta}\left\vert B_s\right\vert>a\right)\leq 2\exp\left(-\frac{a^2}{2\Delta}\right)$ for any positive number $a$.

Therefore,
\begin{align}
    \mathbb{E}\left\{E_1^2\cdots E_m^2\right\}\leq \exp\left(Tl\right) \times \left(\frac{1+4\Delta^2 l}{1-4\Delta^2 l}\right)^m.
    \label{eq: second moment bound}
\end{align}
Plugging \eqref{eq: probability bound} and \eqref{eq: second moment bound} into \eqref{eq: omitted term bound} concludes the proof.
\end{proof}
\section{Wald experiments}
\label{sec: wald experiments}
In this section, we wish to numerically show that the Wald estimate is biased, i.e. $\mathbb{E}^\theta(\hat{L}(\theta))/L(\theta)$ changes as $\theta$ changes where $K = \inf\{k>0:E^{(1)}+\ldots+E^{(k)}>0\}$, $\hat{L}(\theta)=\sum_{i=1}^K E^{(i)}$ and $L(\theta)=\mathbb{E}^\theta\left(G^\theta(X_0)\right)$.

Here is an example of which we know the true solution to. The dynamics that describes how one dimensional process $X$ evolves is given by
\begin{equation*}
    dX_t = bdt+dW_t
\end{equation*}
where $b$ is some constant and $W$ is a one dimensional Brownian motion. Hence 
\begin{equation*}
    X_t|X_{t'}=x \sim \mathcal{N}\left(x+b\times (t-t'), t-t'\right).
\end{equation*}
Thus, $\theta=b$ in this case. We can exactly calculate $L(\theta)$ for $\lambda(x)=x+10$,
\begin{equation*}
    L(\theta)=\mathbb{E}\left(\exp\left(-\int_0^T\lambda(X_t)dt\right)\Big\vert X_0=0\right)=\exp\left(-10T-\frac{b}{2}T^2+\frac{1}{6}T^3\right)
\end{equation*}
where the expectation is taken with respect to Brownian motion $X|X_0=0$. Note that $\int_0^T (10+X_0+bt+W_t-W_0) dt\sim \mathcal{N}\left(10T+\frac{b}{2}T^2, \frac{1}{3}T^3\right)$ for $X_0=W_0=0$.

We obtain $N=10^6$ independent samples $\hat{L}(b)$ for every value of $b$: 
\begin{equation*}
    \hat{L}'(b)=\frac{1}{N}\sum_{j=1}^N \sum_{i=1}^{K_j} E_j^{(i)} \approx \mathbb{E}^\theta(\hat{L}(b))
\end{equation*}
where $K_j = \inf\{k>0:E^{(1)}_j+\ldots+E^{(k)}_j>0\}$. Each $E_j^{(i)}$ is an independent sample where $E_j^{(i)}\leftarrow \mathrm{PE}(T,0,T,0)$. 
\begin{figure}[t!]
     \centering
     \begin{subfigure}[b]{0.32\textwidth}
         \centering
         \includegraphics[width=\textwidth]{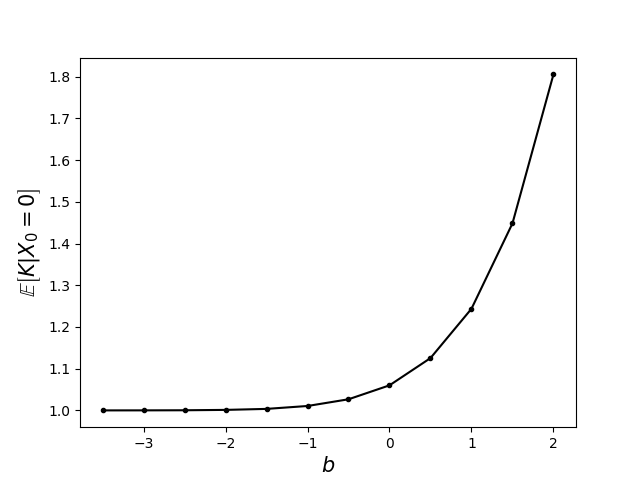}
         \caption{}
         \label{fig: draws}
     \end{subfigure}
     \hfill
     \begin{subfigure}[b]{0.32\textwidth}
         \centering
         \includegraphics[width=\textwidth]{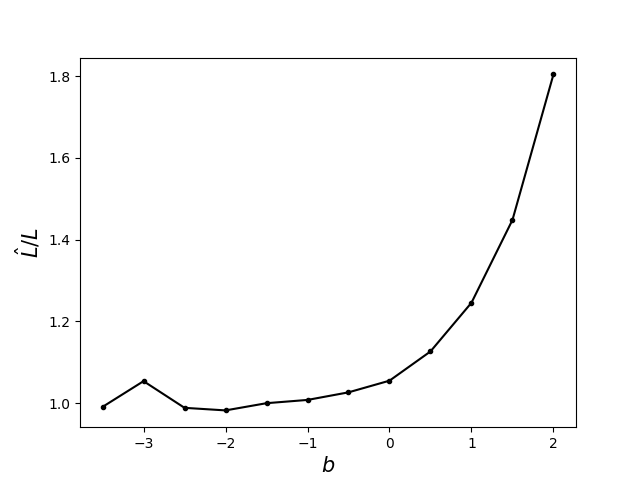}
         \caption{}
         \label{fig: wald ratio}
     \end{subfigure}
          \hfill
     \begin{subfigure}[b]{0.32\textwidth}
         \centering
         \includegraphics[width=\textwidth]{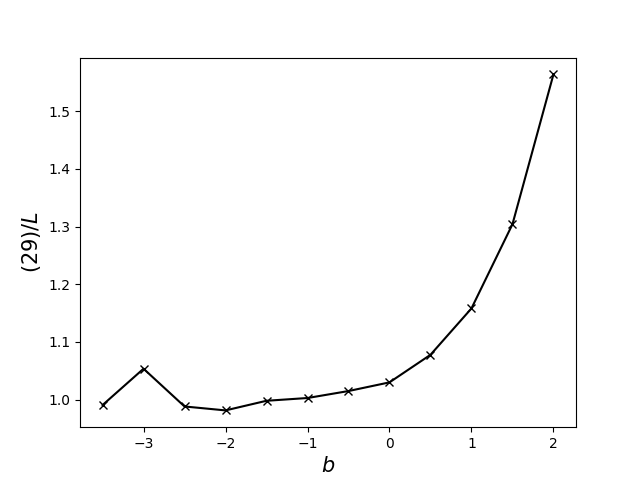}
         \caption{}
         \label{fig: wald ratio2}
     \end{subfigure}
     \caption{(a) Plot of $\mathbb{E}^\theta(K|X_0=0)$, (b) $\hat{L}'/L$ and (c) estimate \eqref{eq: wald ratio 2}/$L$ versus $\theta=b$.} 
\end{figure}
Figure \ref{fig: draws} shows that as $b$ increases, the number of draws to make Wald estimate positive increases. For this example $E$ is 
\begin{equation*}
    E = \exp\left(-T(X_0+10)\right)\prod_{i=1}^\kappa \left[1+\left[X_0-X_{\tau_i}\right]\right]
\end{equation*}
where $\kappa \sim \mathcal{P}o(T)$ and $\tau_1,\ldots, \tau_\kappa \sim \mathcal{U}(0, T)$ are i.i.d. samples. Larger $b$ (i.e. larger drift dragging the particle towards positive direction) increases the chances of meeting negative Poisson estimate. In Figure \ref{fig: wald ratio}, we notice a clear trend that the empirical ratio, $\hat{L}'(b)/L(b)$, increases with $b$. Finally we plot 
\begin{equation}
    \frac{1}{N}\sum_{j=1}^N \frac{1}{K_j}\sum_{i=1}^{K_j} E_j^{(i)}.
    \label{eq: wald ratio 2}
\end{equation}
\section{Exact computation of likelihood function}
\label{sec: exact computation of likelihood function}
This section is to determine the following likelihood function.
\begin{equation*}
    \mathbb{E}\left\{\exp\left[-\int^T_\tau \left(\alpha X_t+\beta\right) dt\right]\Big\vert X_\tau=x_0, X_T=x_1\right\}
\end{equation*}
The procedure can be splitted into 4 steps.
\begin{enumerate}
    \item  As $X_t|X_T$ is a Gaussian process, the Lebesgue integral is Gaussian random variable, see \cite{folland1999real}: approximate the given integral as Riemann sums and each Riemann sum is Gaussian and hence the limit will also be Gaussian.
    \begin{equation*}
        \left(\int^T_\tau \alpha X_t+\beta dt \Big\vert X_\tau=x_0, X_T=x_1 \right)\sim \mathcal{N}\left(\alpha\mu+\beta\left(T-\tau\right), \alpha^2\sigma^2\right)
    \end{equation*}

    \item Calculate mean $\mu$:
\begin{align*}
    \mathbb{E}\left[\int^T_\tau X_tdt\Big\vert X_\tau=x_0,X_T=x_1\right]&=\int^T_\tau\mathbb{E}\left[X_t\Big\vert X_\tau=x_0,X_T=x_1\right]dt\\
    &=\int^T_\tau x_0 +\frac{t-\tau}{T-\tau}\left(x_1-x_0\right)dt\\
    &=\frac{1}{2}\left(T-\tau\right)\left(x_0+x_1\right)
\end{align*}
\item Calculate variance $\sigma^2$:
\begin{align*}
    &\mathbb{E}\left[\left(\int^T_\tau X_tdt\right)\left( \int^T_\tau X_t dt\right)\Big\vert X_\tau=x_0,X_T=x_1\right]- \mu^2
    \\
    =&\mathbb{E}\left[\int_{\left[\tau,T\right]^2}X_uX_vdudv\Big\vert X_\tau=x_0,X_T=x_1\right]-\mu^2\\
    =&\int_{\left[\tau,T\right]^2}\mathbb{E}\left[X_uX_v\Big\vert X_\tau=x_0,X_T=x_1\right]dudv-\mu^2\\
    =&\int_{\left[\tau,T\right]^2}\text{cov}\left(X_u, X_v\right)+\mathbb{E}\left[X_u\Big\vert X_\tau=x_0,X_T=x_1\right]\times \mathbb{E}\left[X_v\Big\vert X_\tau=x_0,X_T=x_1\right]dudv -\mu^2\\ 
    =&\int_{[\tau,T]^2} \frac{(u\wedge v-\tau)(T-u\vee v)}{T-\tau} +\left(x_0 +\frac{u-\tau}{T-\tau}(x_1-x_0)\right)\times\left(x_0 +\frac{v-\tau}{T-\tau}(x_1-x_0)\right) dudv -\mu^2\\
    =& \int_\tau^T \int^v_\tau \frac{(u-\tau)(T-v)}{T-\tau} dudv +\int^T_\tau\int^T_v \frac{(v-\tau)(T-u)}{T-\tau}dudv\\
    =&\frac{(T-\tau)^3}{12}
\end{align*}
\item Calculate the likelihood:
\begin{align*}
    &\mathbb{E}\left\{\exp\left[-\int^T_\tau (\alpha X_t+\beta) dt\right]\Big\vert X_\tau=x_0,X_T=x_1\right\}\\=&\exp\left[-\frac{\alpha}{2}(T-\tau)(x_0+x_1)-\beta(T-\tau)+\frac{\alpha^2(T-\tau)^3}{24}\right]
\end{align*}
\end{enumerate}
Therefore, the exact likelihood for 
$t_{1:n_p}$ and $y_{{t_1}:{t_{n_p}}}$, where $n_p$ is the number of observations, is 
\begin{align}
    \mathcal{L}=\mathbb{E}&\left\{\left[\prod_{i=1}^{n_p}\left(X_{t_i}+10\right)g^\theta(y_{t_i}|X_{t_i})\right.\right.\left.\times\exp\left(-\frac{t_i-t_{i-1}}{2}\left(X_{t_{i-1}}+X_{t_i}\right)-10(t_{i}-t_{i-1})+\frac{(t_i-t_{i-1})^3}{24}\right)\right]\nonumber
    \\
    &\left.\times\exp\left(-\frac{T-t_{n_p}}{2}\left(X_{t_{n_p}}+X_T\right)-10(T-t_{n_p})+\frac{(T-t_{n_p})^3}{24}\right)\right\}.
    \nonumber
\end{align}
To find the ground truth for values of $n_p>2$, we use Algorithm \ref{alg: bootstrap particle filter} with line 8 using the exact evaluation (given by
\eqref{eq:exact_E_k}). This allows the computation of Monte Carlo estimate described in Section \ref{sec: benchmark}.
\begin{equation}
    E_k^{(i)}=\exp\left[-\frac{1}{2}\left(t_k^\Delta-t_{k-1}^\Delta\right)\left(X_{k-1}^\Delta+X_k^\Delta\right)-10\left(t_k^\Delta-t_{k-1}^\Delta\right)+\frac{\left(t_k^\Delta-t_{k-1}^\Delta\right)^3}{24}\right].
    \label{eq:exact_E_k}
\end{equation} 
\section{No Observation Case and Two Observation Case}
\label{sec: exact likelihood two observation}
The exact likelihood for no observation received within $[0,T]$ is
\begin{align}
    &\mathbb{E}\left\{\exp\left(-\int_0^T\lambda(X_s)ds\right)\right\}\nonumber\\
    =&\mathbb{E}\left\{\mathbb{E}\left[\exp\left(-\int_0^T\lambda(X_s)ds\right)\left\vert X_0=0, X_T\right.\right]\right\}\nonumber\\
    =&\mathbb{E}\left\{\exp\left(-\frac{TX_T}{2}-10T+\frac{T^3}{24}\right)\right\}\nonumber\\
    =&\int_{-\infty}^\infty \frac{1}{\sqrt{2\pi T}}\exp\left(-10T+\frac{T^3}{24}+\frac{T^3}{8}\right)\exp\left(-\frac{(x_T-\frac{T^2}{2})^2}{2T}\right)dx_T\nonumber\\
    =&\exp\left(-10T+\frac{T^3}{6}\right)\nonumber
\end{align}
and the exact likelihood for two observations received within $[0,T]$ is
\begin{align}
    \mathcal{L}=&\mathbb{E}\left\{\left[\prod_{i=1}^{2}\lambda\left(X_{t_i}\right)g^\theta\left(y_{t_i}\vert X_{t_i}\right)\exp\left(-\int_{t_{i-1}}^{t_i}\lambda(X_s)ds\right)\right] \exp\left(-\int_{t_2}^{T}\lambda\left(X_s\right)ds\right)\right\}\nonumber\\
    =&\int_{-\infty}^\infty 
    \int_{-\infty}^\infty
    \int_{-\infty}^\infty
    \left(v_1+10\right)\left(v_1+v_2+10\right)\times\frac{1}{2\pi \sigma_y^2}\exp\left(-\frac{(y_{t_1}-v_1)^2+(y_{t_2}-v_1-v_2)^2}{2\sigma_y^2}\right)\nonumber\\
    &\times\exp\left(-\frac{t_1v_1}{2}-\frac{t_2-t_1}{2}(2v_1+v_2)-\frac{T-t_2}{2}(2v_1+2v_2+v_3)-10T\right)\nonumber\\
    &\times\exp\left(\frac{t_1^3+(t_2-t_1)^3+(T-t_2)^3}{24}\right)\times\frac{1}{\sqrt{2\pi t_1}}\times\frac{1}{\sqrt{2\pi (t_2-t_1)}}\times\frac{1}{\sqrt{2\pi(T-t_2)}}\nonumber\\
    &\times\exp\left(-\frac{v_1^2}{2t_1}-\frac{v_2^2}{2(t_2-t_1)}-\frac{v_3^2}{2(T-t_2)}\right)dv_3 dv_2dv_1\nonumber
\end{align}
The first integral with respect to $v_3$ is
\begin{align*}
    \int_{-\infty}^\infty \exp\left(-\frac{T-t_2}{2}v_3-\frac{v_3^2}{2(T-t_2)}\right)dv_3=\sqrt{2\pi (T-t_2)}\exp\left(\frac{1}{8}\left(T-t_2\right)^3\right).
\end{align*}
The second integral with respect to $v_2$ is
\begin{align*}
    &\int_{-\infty}^\infty \left(v_1+v_2+10\right)\exp\left(-\frac{\left(y_{t_2}-v_1-v_2\right)^2}{2\sigma_y^2}\right)\\
    &\times\exp\left(-\frac{t_2-t_1}{2}v_2-\frac{T-t_2}{2}\times 2v_2\right)\times\exp\left(-\frac{v_2^2}{2(t_2-t_1)}\right)dv_2\\
    =&\int_{-\infty}^\infty (v_1+v_2+10)\exp\left(-\left(\frac{1}{2\sigma_y^2}+\frac{1}{2(t_2-t_1)}\right)v_2^2-\left(\frac{-y_{t_2}+v_1}{\sigma_y^2}-\frac{t_1+t_2-2T}{2}\right)v_2\right)\\&\times\exp\left(-\frac{y_{t_2}^2-2v_1y_{t_2}+v_1^2}{2\sigma_y^2}\right)dv_2\\
    =&\int_{-\infty}^\infty (v_1+v_2+10)\exp\left(-\frac{y_{t_2}^2-2v_1y_{t_2}+v_1^2}{2\sigma_y^2}+\frac{\mu_2^2}{2\sigma_2^2}\right)\exp\left(-\frac{(v_2-\mu_2)^2}{2\sigma_2^2}\right)dv_2\\
    =&\sqrt{2\pi \sigma_2^2}\left(\mu_2+v_1+10\right)\exp\left(-\frac{y_{t_2}^2-2v_1y_{t_2}+v_1^2}{2\sigma_y^2}+\frac{\mu_2^2}{2\sigma_2^2}\right)
\end{align*}
where 
\begin{equation*}
    \sigma_2^2=\left(\frac{1}{\sigma_y^2}+\frac{1}{t_2-t_1}\right)^{-1}, \qquad\mu_2=\sigma_2^2\left(\frac{y_{t_2}-v_1}{\sigma_y^2}+\frac{t_1+t_2-2T}{2}\right)=\sigma_2^2(av_1+b)
\end{equation*} 
and for 
\begin{equation*}
    a=-\frac{1}{\sigma_y^2},\qquad b=\frac{y_{t_2}}{\sigma_y^2}+\frac{t_1+t_2-2T}{2}.
\end{equation*}

The third integral with respect to $v_1$ is as follows,
\begin{align*}
    &\int_{-\infty}^\infty \left(v_1+10\right)\left(\mu_2+v_1+10\right)\exp\left(-\frac{\left(y_{t_1}-v_1\right)^2+y_{t_2}^2-2v_1y_{t_2}+v_1^2}{2\sigma_y^2}\right)\\&\times\exp\left(\frac{\sigma_2^2}{2}\left(a^2v_1^2+2abv_1+b^2\right)\right)\times\exp\left(-\frac{t_1v_1}{2}-\left(t_2-t_1\right)v_1-(T-t_2)v_1\right)\exp\left(-\frac{v_1^2}{2t_1}\right)dv_1
    \\
    =&\int_{-\infty}^\infty \left[(\sigma_2^2a+1)v_1^2+\left(10(\sigma_2^2a+1)+\sigma_2^2b+20\right)v_1+10\sigma_2^2b+100\right]\\
    &\times\exp\left(-\left(\frac{1}{\sigma_y^2}-\frac{a^2\sigma_2^2}{2}+\frac{1}{2t_1}\right)v_1^2+\left(\frac{2y_{t_1}+2y_{t_2}}{2\sigma_y^2}+ab\sigma_2^2+\frac{t_1-2T}{2}\right)v_1\right) \\
    &\times \exp\left(-\frac{y_{t_1}^2+y_{t_2}^2}{2\sigma_y^2}+\frac{b^2\sigma_2^2}{2}\right)dv_1\\
   =&\sqrt{2\pi \sigma_1^2}\left[(\sigma_2^2a+1)(\mu_1^2+\sigma_1^2)+\left(10(\sigma_2^2a+1)+\sigma_2^2b+20\right)\mu_1+10\sigma_2^2b+100\right]\\&\times\exp\left(-\frac{y_{t_1}^2+y_{t_2}^2}{2\sigma_y^2}+\frac{b^2\sigma_2^2}{2}+\frac{\mu_1^2}{2\sigma_1^2}\right)
\end{align*}
where $\sigma_1^2= \left(\frac{2}{\sigma_y^2}-a^2\sigma_2^2+\frac{1}{t_1}\right)^{-1}$ and $\mu_1=\sigma_1^2\left(\frac{y_{t_1}+y_{t_2}}{\sigma_y^2}+ab\sigma_2^2+\frac{t_1-2T}{2}\right)$

Therefore,
\begin{align*}
    \mathcal{L}=&\frac{1}{2\pi\sigma_y^2}\times\frac{\sigma_1\sigma_2}{\sqrt{t_1\left(t_2-t_1\right)}}\left[\left(\sigma_2^2a+1\right)\left(\mu_1^2+\sigma_1^2\right)+\left(10\left(\sigma_2^2a+1\right)+\sigma_2^2b+20\right)\mu_1+10\sigma_2^2b+100\right]\\
    &\times\exp\left(-\frac{y_{t_1}^2+y_{t_2}^2}{2\sigma_y^2}+\frac{b^2\sigma_2^2}{2}+\frac{\mu_1^2}{2\sigma_1^2}\right)
    \times\exp\left(-10T+\frac{t_1^3+(t_2-t_1)^3+(T-t_2)^3}{24}+\frac{1}{8}\left(T-t_2\right)^3\right) \nonumber.
\end{align*}
The exact likelihood are used to compute the relative MSE for Section \ref{sec: numerical experiments} and Appendix \ref{sec: empirical relationship between rel var and delta}.
\section{Empirical relationship between relative variance and $\Delta$}
\label{sec: empirical relationship between rel var and delta}
\begin{figure}[t!]
     \centering
     \includegraphics[width=0.55\linewidth]{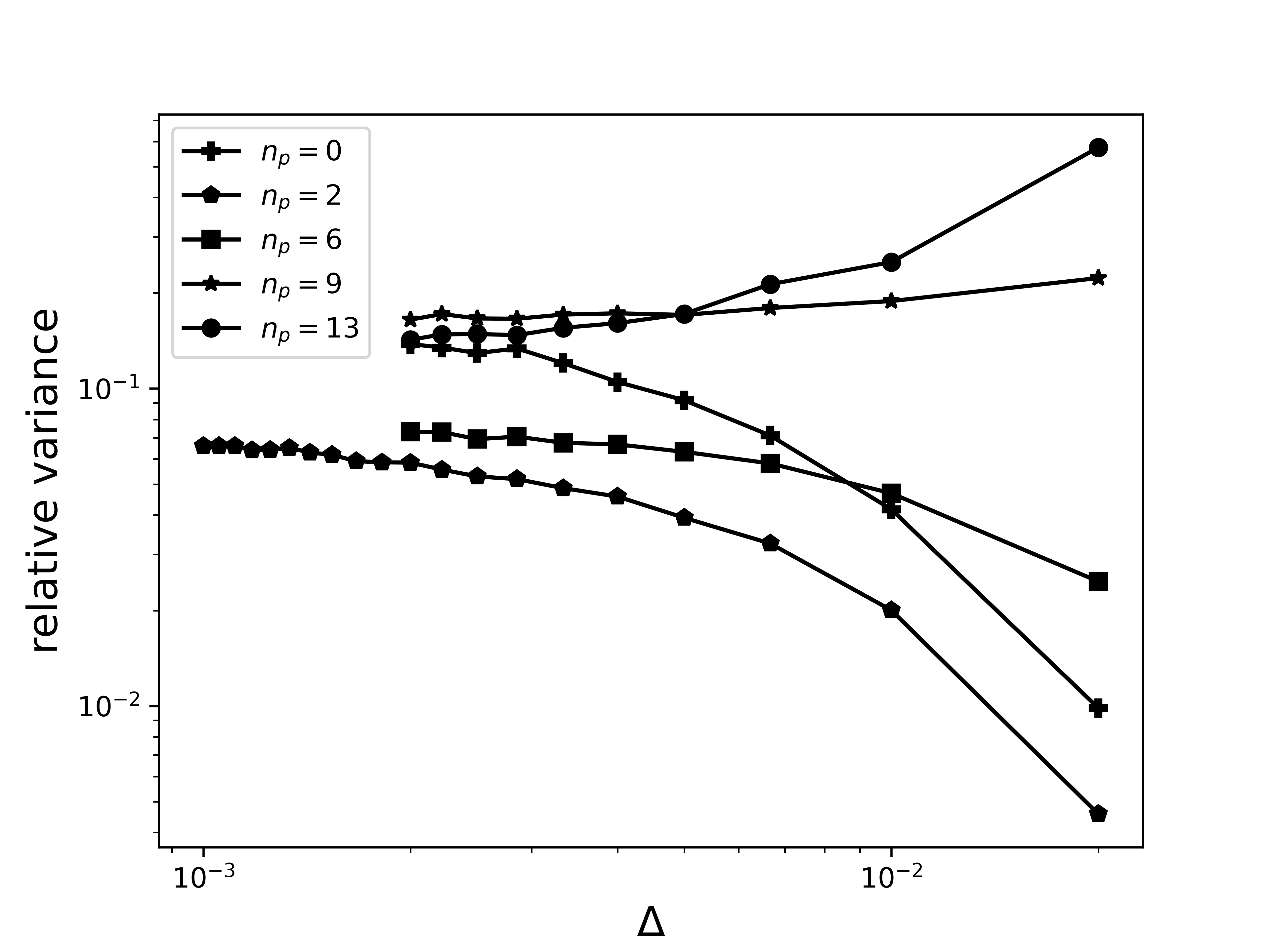}
     \caption{Plot of relative variance versus $\Delta$ for different values of $n_p$ and fixed $N=100$ in log scale.}
    \label{fig:relative variance vs Delta different np}
\end{figure}
Figure \ref{fig:relative variance vs Delta different np} reports the relationship between relative variance and $\Delta$ for different $n_p$ values and fixed $N=100$. For $n_p=0$ and $n_p=2$ cases, the exact likelihood is computed using solutions calculated in Section \ref{sec: exact likelihood two observation}, for other larger values of $n_p$, the Monte Carlo estimate $\mathcal{L}_{\text{MC}}$ is used in relative variance computation. Results show that the relationship between relative variance and $\Delta$ can be highly $n_p$-dependent. As $n_p$ of problem increases, the rate of change in relative variance becomes less positive when $\Delta$ approaches zero. A more general trend that applies to all values of $n_p$ is that the relative variance eventually becomes constant as $\Delta$ goes to zero.

\section{Thinning algorithm for creating data}
\label{sec: thinning algorithm}
This section describes the thinning algorithm we use to generate observation data. Please refer Algorithm \ref{alg: thinning algorithm} for details.
\begin{algorithm}[t!]
\DontPrintSemicolon
  
  \KwInput{$\lambda_{\max}=\lambda_0$, $T$}
    Generate $N\sim\mathcal{P}o(\lambda_{\max}T)$;\\
    Generate $t_1,t_2,\ldots, t_N\sim\mathcal{U}(0,T)$;\\
    Sort $t_1,t_2,\ldots, t_N$ and relabel them so that $t_1<t_2<\ldots<t_N;$\\
    Generate $X_0\sim \nu(x)$ and set $\tau=0$;
    \\
  \For{$i \in \{1:N\}$}{Propagate $X_{t_i}$ from previous $\tau$, i.e. $X_{t_i}\sim f_{t_i-\tau}^\theta(x_{t_i}|X_{\tau})$;
  \\
  Generate $U\sim\mathcal{U}(0,1)$;
  \\
  \If{$U\leq \lambda(X_{t_i})/\lambda_{\max}$}{Keep $t_i$ as a real observation time and set $\tau=t_i$;\\ Generate $y_{t_i}$ which is a realisation of $Y_{t_i}\sim g^{\theta}(y|X_{t_i})$ with Born and Wolf point spread function.}
  }
  \KwOutput{all pairs of $(t_i, y_{t_i})$}
\caption{Thinning Algorithm for simulating the observation times with intensity function $\lambda(X_t)$ on $[0,T]$}
\label{alg: thinning algorithm}
\end{algorithm}
\section{Additional experiments}
\label{sec: additional experiments}
\begin{figure}[t!]
     \centering
     \begin{subfigure}[b]{0.48\textwidth}
         \centering
         \includegraphics[width=\textwidth]{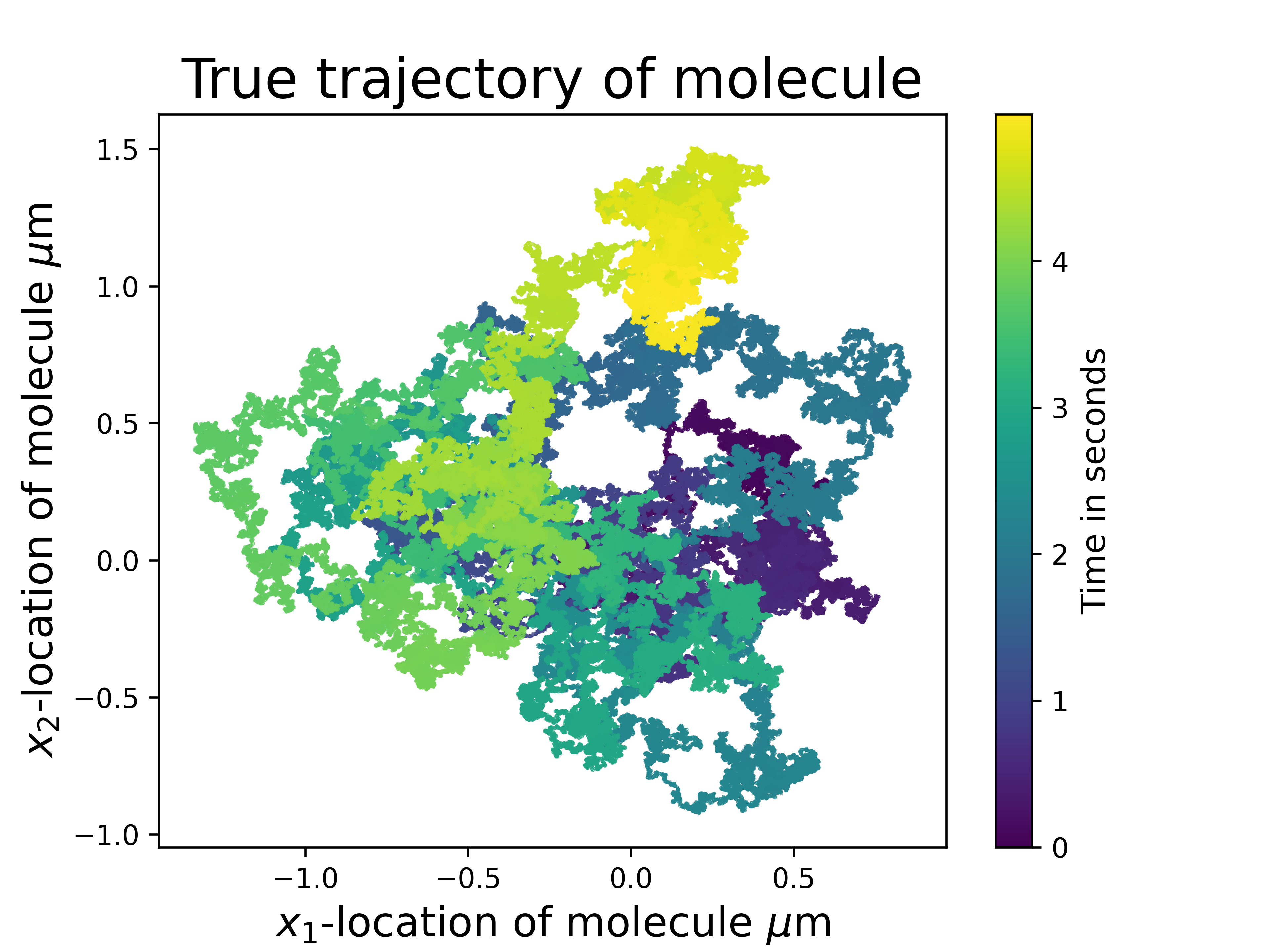}
         \caption{}
         \label{fig:true trajectory worse}
     \end{subfigure}
     \hfill
     \begin{subfigure}[b]{0.48\textwidth}
         \centering
         \includegraphics[width=\textwidth]{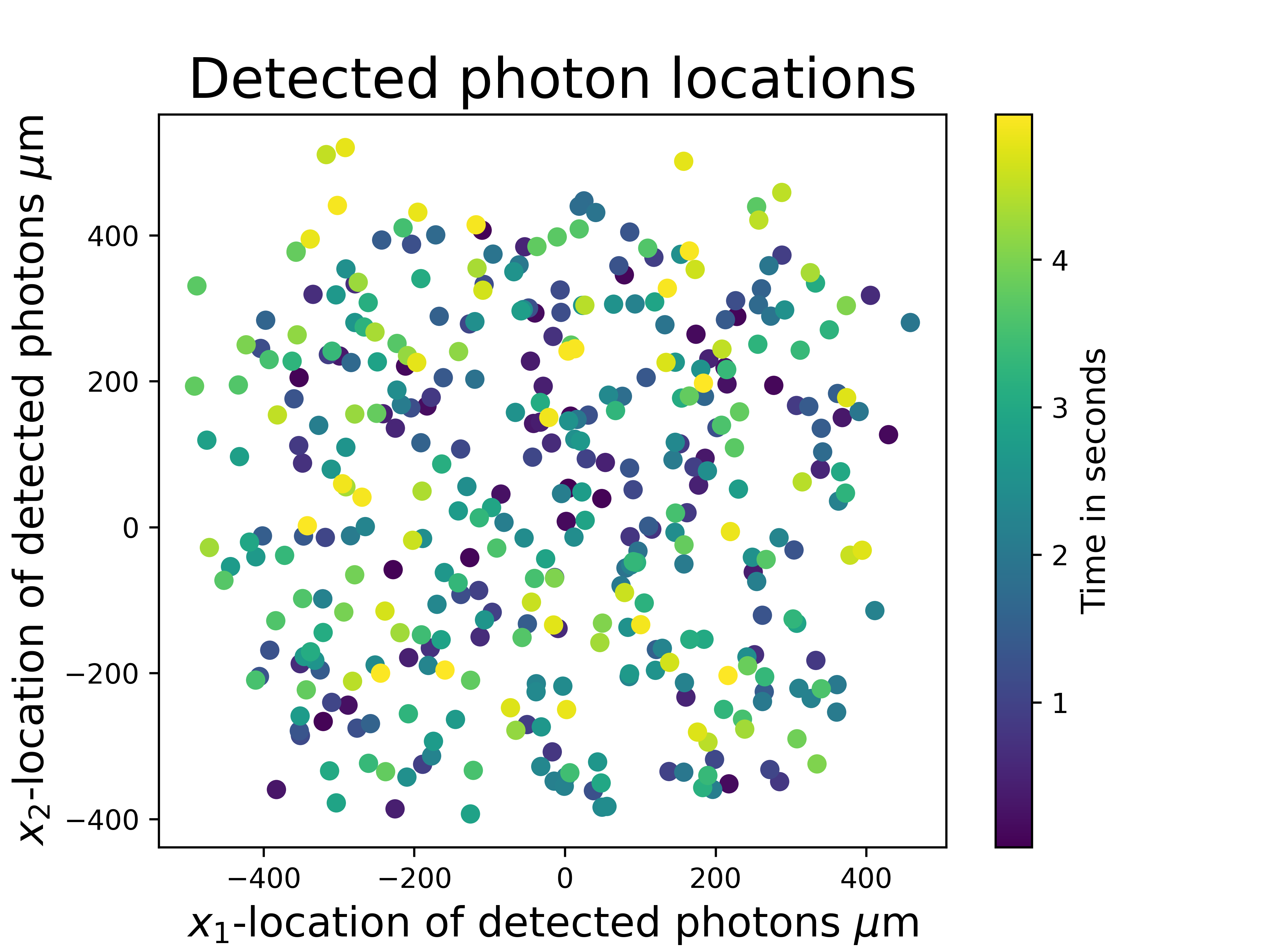}
         \caption{}
         \label{fig:BW model worse}
     \end{subfigure}
     \begin{subfigure}[b]{0.48\textwidth}
         \centering
         \includegraphics[width=\textwidth]{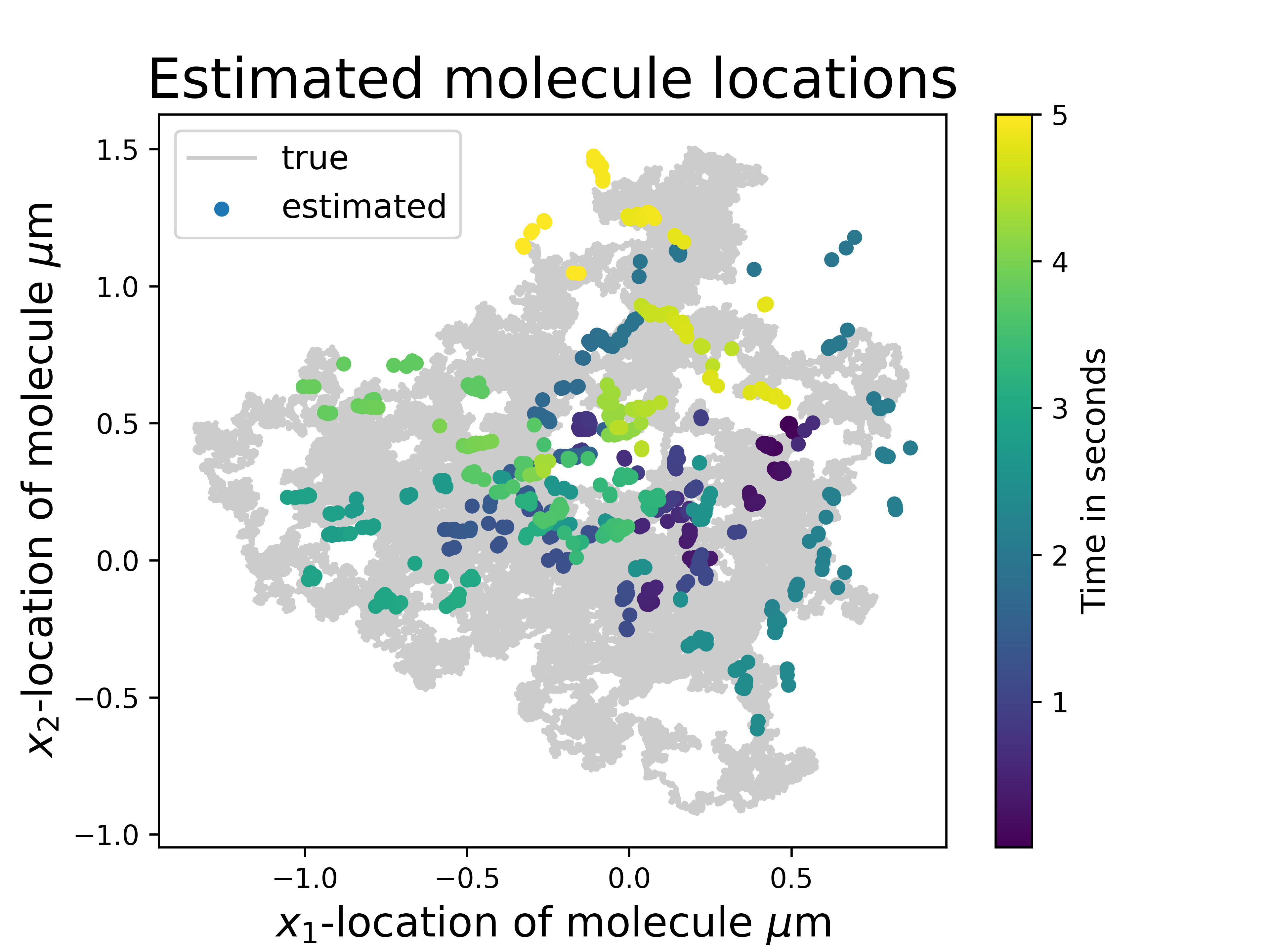}
         \caption{}
         \label{fig:BW_estimated worse}
     \end{subfigure}
     \begin{subfigure}[b]{0.48\textwidth}
         \centering
         \includegraphics[width=\textwidth]{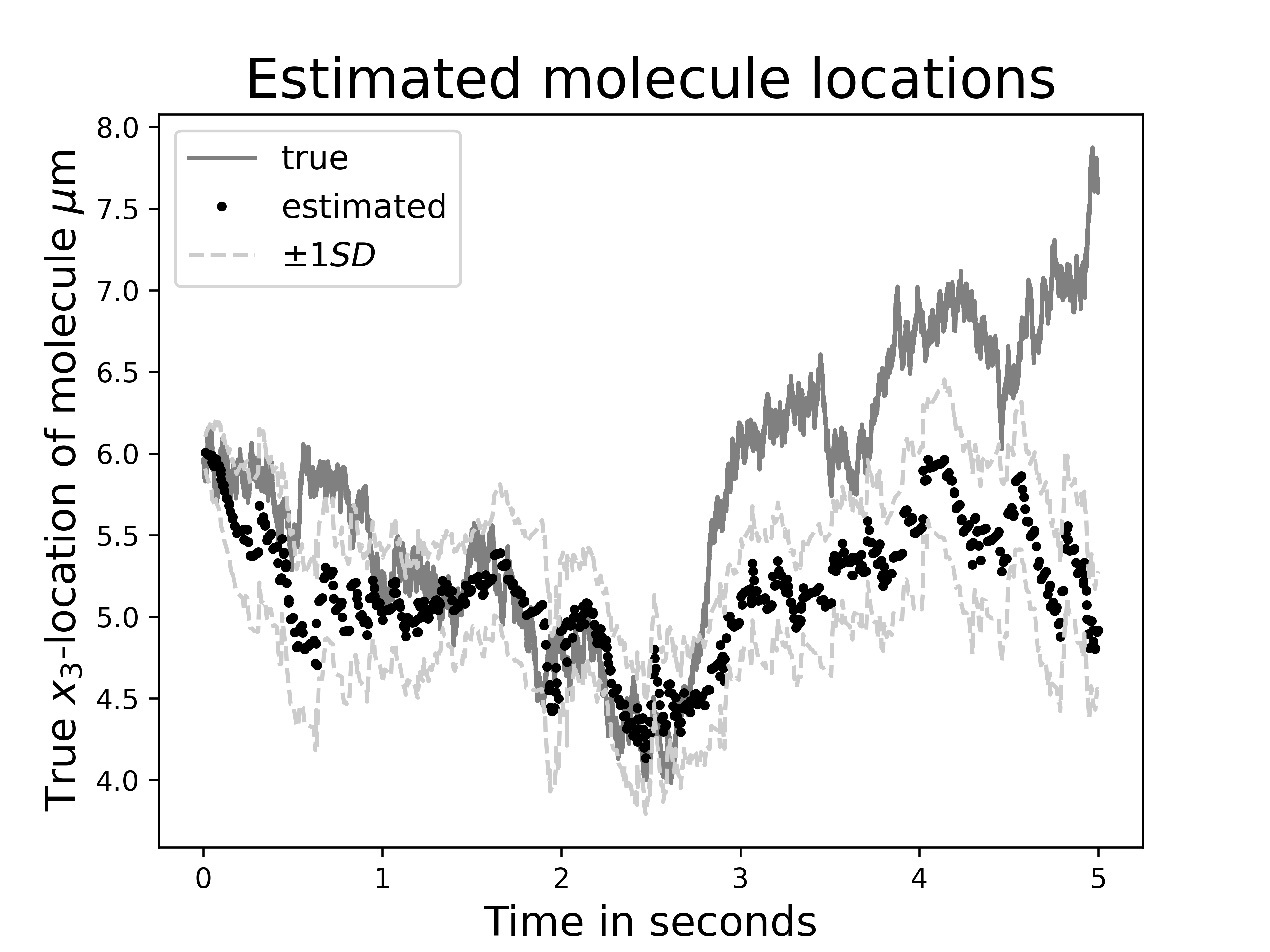}
         \caption{}
         \label{fig:x_3_estimated_BW worse}
     \end{subfigure}
     \caption{(a) True trajectory of a molecule; (b) observed photon locations; (c) estimated $(x_1, x_2)$ molecule locations  and (d) true $x_3$ molecule locations and estimated location.}
\end{figure}
In Figure \ref{fig:true trajectory}, we plot the true trajectory of a molecule, and simulate using parameters $\{\theta=(1.0, 1.0, 1.0)^\top, \mu=(0.5, 0.5, 6.0)^\top, p_0=0.01\}$ for time interval [0, 5.0]. Other parameters remain the same as in Section \ref{sec: 3D single molecule model}. Figure \ref{fig:BW_estimated} shows the filtered $(x_1, x_2)$ mean locations of molecules, which deviate from their right positions. Figure \ref{fig:x_3_estimated_BW} shows the filtered mean of $x_3$ and regions of $\pm$1 standard deviation together with the true state of $X_3$ at the observation times. In comparison to the Figure \ref{fig:BW_estimated} and \ref{fig:x_3_estimated_BW}, Figure \ref{fig:BW_estimated worse} and \ref{fig:x_3_estimated_BW worse} shows that higher values of $x_3$ degrade the estimation quality of particle filtering algorithm on the state of molecule and this is due to the exponential function structure of Born and Wolf image function which generates photons that are detected very far from the true molecule position.
\bibhang=1.7pc
\bibsep=2pt
\fontsize{9}{14pt plus.8pt minus .6pt}\selectfont
\renewcommand\bibname{\large \bf References}
\expandafter\ifx\csname
natexlab\endcsname\relax\def\natexlab#1{#1}\fi
\expandafter\ifx\csname url\endcsname\relax
  \def\url#1{\texttt{#1}}\fi
\expandafter\ifx\csname urlprefix\endcsname\relax\def\urlprefix{URL}\fi

 \bibliographystyle{chicago}      
 \bibliography{ref.bib}   
 \newpage

\end{document}